\documentclass[12pt]{amsart}
\usepackage{amsfonts,amssymb,graphicx,mystyle,color,amsmath}
\usepackage[colorlinks,citecolor={blue},linkcolor=blue, runcolor={blue}]{hyperref} 
\usepackage[round]{natbib}
\usepackage{array}
\usepackage{tikz,istgame}
\usepackage{enumitem}
\usepackage{comment}

\makeatletter
\def\subsection{\@startsection{subsection}{2}%
\z@{.7\linespacing\@plus\linespacing}{.5\linespacing}%
{\normalfont\itshape\centering}}
\makeatother

\def\a{\alpha} 
\def\b{\beta} 
\def\d{\delta} 
\def\e{\varepsilon} 
 
\def\g{\gamma} 
\def\l{\lambda} 
\def\s{\sigma} 
\def\t{\tau} 
\def\z{\zeta}

\def\D{\Delta} 
\def\G{\Gamma}  
\def\S{\Sigma}

  \def\N{\mathcal N}
 
\def\P{\mathcal P}   
\def\Re{\mathbb{R}}

\def\rho{\varrho}

\def\mG{\mathbb{G}}
\def\mS{\mathbb{S}}
\def\mP{\mathbb{P}}
\def\G{\mathcal{G}}
\def\e{\varepsilon}

\def\supp{\mathrm{Supp}}
\def\BR{\mathrm{BR}}
\def\cl{\mathrm{cl}}

\usepackage{fouriernc}
\usepackage[T1]{fontenc}

\makeatletter
\long\def\@ympar#1{%
  \@savemarbox\@marbox{\tiny #1}%
  \global\setbox\@currbox\copy\@marbox
  \@xympar}
\makeatother
\begin{document}%\openup 1\jot
\title[Generic Robustness of Equilibria]{Robust Equilibria in Generic Extensive-form Games$^*$}
\author[Lucas Pahl]{Lucas Pahl$^\dag$}
\author[Carlos Pimienta]{Carlos Pimienta$^\ddag$}
 \date{\today}
%\href{\href{run:./HE20230920.pdf}{Click here for the latest version}
\thanks{$^*$~We thank Srihari Govindan for drawing our attention to the conjecture we prove in this paper and for comments, and G V A Dharanan for proofreading the manuscript. We are grateful to Paulo Barelli, Joel Sobel, John Levy and Klaus Ritzberger for comments and suggestions.  We thank audiences at the Paris Game Theory Seminar, BGSE Workshop in Bonn, CES-Sorbonne, LSE-Mathematics, U. of Surrey and Royal Holloway U. of London. Lucas thanks research support from the Hausdorff Center for Mathematics (DFG project no. 390685813).
Carlos thanks research support from the Australian Research Council's Discovery Project DP190102629. }
\thanks{$^\dag$~\textsc{School of Economics, University of Sheffield, Sheffield, United Kingdom}}
\thanks{$^\ddag$~\textsc{School of Economics, The University of New South Wales, Sydney, Australia}}%

\begin{abstract}
We prove the 2-player, generic extensive-form case of the conjecture of \citet{GW1997a,GW1997b} and \citet{HH2002} stating that an equilibrium component is essential in every equivalent game if and only if the index of the component is nonzero. 
This provides an index-theoretic characterization of the concept of hyperstable components of equilibria in generic extensive-form games, first formulated by~\citet{KM1986}. We also illustrate how to compute hyperstable equilibria in multiple economically relevant examples and show how the predictions of hyperstability compare with other solution concepts. 
\end{abstract}

\maketitle

\defcitealias{GW2005}{GW}

%\begin{center}\href{https://www.dropbox.com/scl/fi/shgzmlf1f2tdzive0k3vh/HE20230928.pdf?rlkey=0f9ez1a5pbru3xs0n61akhwmf&dl=0}{\Large{Click here for the latest version}}\end{center}

\section{Introduction}
Hyperstable equilibria, as formulated in \citet{KM1986}, are those Nash equilibria that are robust to perturbations not only of the payoffs of the given game, but also to \textit{equivalent games} of the given game (i.e., games that are obtained by adding duplicate strategies). In other words, it is a solution concept characterized by two simple properties: (1) robustness to payoff perturbations (\textit{strategic essentiality}); (2) the principle that equivalent games should have equivalent solutions (also known as \textit{invariance}).

Hyperstable equilibria constitute an important step in Kohlberg and Mertens' stability program as the first equilibrium concept to satisfy the basic properties of existence, sequential rationality, invariance and iterated dominance, making it an appealing refinement of Nash equilibria. Even if hyperstable equilibria do not always satisfy admissibility, in generic extensive-form games hyperstable equilibrium outcomes do satisfy all those properties.\footnote{In generic extensive-form games, hyperstable equilibrium outcomes are always induced by a Mertens-stable set \citep{M1989} which, in turn, does satisfy all those properties.} However, hyperstability is a highly intractable concept. Not only does its definition require robustness to payoff perturbations in the base game - which is in itself hard to check - but also in every equivalent game to the base game.  

In this paper we provide a characterization of hyperstable equilibria that makes it a more operational concept, as we intend to illustrate in the last section with a series of examples. This characterization dispenses with the need of considering equivalent games to the base game and, in several instances, makes the problem of checking for robustness to payoff perturbations almost trivial. 

This gain in simplicity is achieved by formulating the characterization in terms of a concept of topological fixed point theory called the \textit{index of fixed points}, which we now explain. In algebraic topology, tools such as the fixed point index are developed to study and distinguish among fixed points. In particular, when studying fixed points in an abstract setting, it is also of interest to identify robustness properties of the fixed points with respect to perturbations of their fixed-point maps. A topologically essential fixed point of a map is a point such that every nearby map has a close-by fixed point. Topologically essential fixed points do not always exist,\footnote{Consider the identity map from the interval $[0,1]$ to itself. Any point in $[0,1]$ is a fixed point of the map, but if we consider uniform perturbations of this map to near-by maps, none of the points of $[0,1]$ is robust.} but considering the sets of fixed points of a given map, one can extend the definition of topological essentiality from points to \textit{connected components of fixed points}:\footnote{\citet{ON1953} shows that a minimal topologically essential fixed point set of an Euclidean space must be a connected component of the set of fixed points, thus justifying the formulation of topological essentiality to connected components of fixed points.} a connected component of the set of fixed points is topologically essential if for sufficiently small perturbations of the fixed-point map, there are fixed points close to the component. Topologically essential components can be identified by an integer number called their \textit{index} (cf. \citet{ON1953}). The process through which this number is assigned is rather involved and will be explained later in more detail, but the important fact is that a component of fixed points is topologically essential if and only if the index is non-zero. 

 In a game-theoretic context, and for some function whose fixed points coincide with the Nash equilibria of a game, an index can therefore be assigned to every component of Nash equilibria in mixed strategies, and simple properties of the index of fixed points guarantee that at least one topologically essential component of equilibria exists. The conjecture that the strategically essential components coincide with those with nonzero index \citep{GW1997a,GW1997b}, i.e. the topologically essential components, was disproved by \citet{HH2002} by means of a counterexample. They recast the conjecture and hypothesized that the hyperstable components are precisely those that have nonzero index. We show that this is indeed the case for generic two-player extensive-form games. Our notion of genericity implies the commonly used property of extensive-form games according to which each connected component of equilibria in mixed strategies induces a single outcome in an extensive-form game. Therefore, our characterization can be seen as identifying the Nash equilibrium outcomes of an extensive-form game which are hyperstable. While intuitively every topologically essential Nash equilibrium component should be hyperstable, this result implies that the space of perturbations considered in the definition of hyperstability - i.e. payoff perturbations -  is as rich as the one used by topological essentiality, i.e. the space of perturbations of a map whose fixed points are the equilibria of a game. 

A similar result was proved by \citet{GW2005}, henceforth \citetalias{GW2005}, for arbitrary games. However, to obtain such a characterization they require a strengthening of hyperstability beyond the combination of invariance and robustness to payoff perturbations. \citetalias{GW2005} call such a concept \textit{uniform hyperstability}.\footnote{~We define uniform hyperstability in Section~\ref{sec:main}.} 

As mentioned above, only information about the game under study is necessary to verify if a component is topologically essential. Indeed, \citet{S1974} shows that the index of equilibria in generic two-player games can be computed by a simple formula that only depends on parameters of the game. \citet{AB2009} generalized this formula to include non-generic two player games. Even if such formulas can be useful, for many relevant game-theoretical applications in the literature, it is enough to use the properties of the index to know if it is different from zero and, given our result, verify whether a component is hyperstable.\footnote{This is illustrated in Section \ref{sec:application}.}

From an axiomatic perspective, our result can be interpreted as providing an index theoretic characterization for the solution concept that, for each (generic two-player extensive-form) game selects all (closed and connected) subsets of Nash equilibria that satisfy (1) invariance and (2) robustness to payoff perturbations. A component of equilibria with non-zero index is known to be robust in every equivalent game, so a solution concept that selects the collection of equilibrium components with non-zero index is both invariant and robust to payoff perturbations. Conversely, if a solution concept selects all (closed and connected) sets of Nash equilibria that are invariant and robust to payoff perturbations then the selected solutions must be connected components of Nash equilibria.\footnote{~Suppose that at least one solution is a proper subset of a connected component. Then our arguments imply that an equivalent game exists in which an equivalent solution is not robust to payoff perturbations.} Furthermore, our main result implies that any zero index component is not robust in some equivalent game.  
Hence, all selected solutions are connected components with non-zero index.

In parallel to the developments in the strategic stability literature, the index of equilibria has also been used independently as a tool to select equilibria with interesting properties, especially with regards to their dynamic stability.  \cite{KR1994} formulated the index of equilibria for the replicator dynamics and presented some useful applications in game theory, including highlighting that non-zero index components of equilibria satisfy a number of desirable properties. \cite{DR2003} obtained a necessary condition for the asymptotic stability of equilibria in dynamic adjustment processes that is formulated in terms of the index of equilibria. More recently, \cite{AM2016}  argued for the selection of equilibria with index $+1$ on both experimental and theoretical grounds, while  \cite{GLP2023} provided game-theoretic characterization for these. A natural question stemming from this strand of the literature is: what does the index precisely mean in game-theoretic terms? More precisely, while one can use non-zero index equilibria as a selection criterion to guarantee interesting properties for the selected equilibria, such a criterion is purely topological and may imply properties that are game-theoretically undesirable. Our main result provides an answer to this question by formulating the precise combination of game-theoretical properties that characterize non-zero index equilibria and define the classical concept of hyperstability. 

This paper is organized as follows.
Section \ref{sec:prem} presents all the required notation and definitions. 
In particular, it introduces the necessary properties of the index that will be applied throughout the paper and states the main result. 
Section \ref{sec:excluded} defines one the main objects of the analysis---the \textit{excluded game} associated to an equilibrium outcome. 
Section \ref{SEC:GenAssump} lays down the precise genericity assumptions under which our main result holds. 
We prove our main result in Section \ref{sec:proof}.
Such a proof is divided in three steps and is accompanied by a running example to illustrate the main ideas. 
To finalize, Section \ref{sec:application}, contains an illustration on how to compute the index in many economically relevant examples, allowing to identify non-hyperstable equilibria. 
The Appendix contains additional definitions and proofs that do not appear in the main text.

\subsection{An example}\label{Sec:Example}
\begin{figure}[t!]
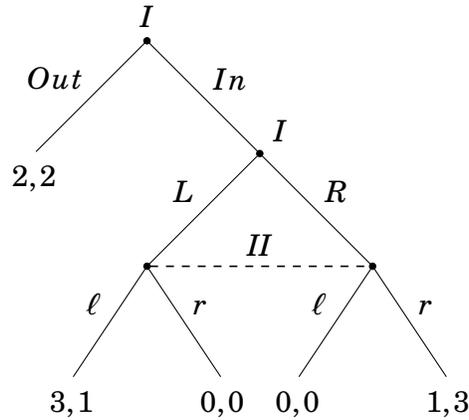

\caption{Entry game}\label{Entry}
%\hfill
%\begin{minipage}{0.45\textwidth}
\begin{istgame}
\centering

\setistmathTF{1}{1}{1}
\xtdistance{15mm}{30mm}
\istroot(0){I}
\istb{Out}[al]{2,2}
\istb{In}[ar]
\endist

\istroot(1)(0-2)<above right>{I}
\istb{L}[al]
\istb{R}[ar]
\endist

\xtdistance{15mm}{20mm}
\istroot(2)(1-1)
\istb{\ell}[al]{3,1}
\istb{r}[ar]{0,0}
\endist

\istroot(3)(1-2)
\istb{\ell}[al]{0,0}
\istb{r}[ar]{1,3}
\endist
\xtInfoset[dashed](2)(3){$\mathit{II}$}

\end{istgame}

\end{figure}

To illustrate the main ideas of the result and the proof, consider the entry game in Figure~\ref{Entry}.
To analyze the example, we use two simple facts about the index of equilibria.
The first one is that the index of a strict equilibrium is +1.
The second is that the sum of the indexes is always +1.
The game in Figure~\ref{Entry} has two components of Nash equilibria.
In the first component, player 1 moves $\mathit{In}$ and then both players play $(L,\ell)$ in the subgame.
Hence, this component is made of a single strategy profile that, furthermore, is a strict equilibrium as any deviation by either player leaves such player with a strictly lower payoff.
(Note as well that $(L,\ell)$ is also a strict equilibrium in the subgame.)
Since the index of a strict equilibrium is +1 and indexes must add up to one, the second component must have zero index.
This second component is such that player 1 moves $\mathit{Out}$ with probability one and player 2 plays $r$ in the subgame with probability at least 1/3.
We call the set of strategy profiles in the subgame such that player 2 plays in this way the \textit{supporting polytope} of the component because it supports on-path play, i.e., at any point in this component player 2 cannot affect the outcome and player 1 is at most indifferent between $\mathit{Out}$ and $\mathit{L}$ when player 2 plays $r$ with probability exactly equal to $1/3$.

We can think of the ``\textit{Out}''-component as \textit{excluding} the subgame, and of the supporting polytope as containing behavior in the subgame that supports that exclusion.
The supporting polytope contains two Nash equilibria of the subgame, $(R,r)$ and the mixed strategy profile $(\frac{3}{4}L+\frac{1}{4}R,\frac{1}{4}\ell+\frac{3}{4}r)$.
When considering the subgame in isolation, the former equilibrium has index +1 because it is strict, and the latter has index $-1$ because the sum of the indexes in the subgame must be +1 and the other two equilibria are strict ($(R,r)$ as just mentioned, and $(L,\ell)$).
If we extend the notion of index of a component of equilibria to the index of a (suitably chosen) neighborhood by adding the indexes of the components included in the neighborhood, we note that the index of any sufficiently small neighborhood of the supporting polytope must be zero.

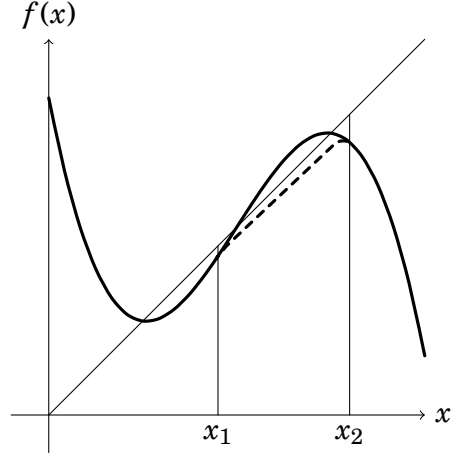
\begin{figure}[t]
\caption{Modifying function $f$ to eliminate two fixed points with indexes adding up to zero.}
\label{fig:perturbation}
\begin{tikzpicture}[domain=0:1,scale=5,line cap=round]
%\draw[very thin,color=gray] (-0.1,-0.1) grid (1.1,1.1);
%\newcommand{\f}[1]{(-9*(#1-0.25)*(#1-0.5)*(#1-0.75)+ #1 )}
\tikzset{declare function={g(\x)=(-9*(\x-0.25)*(\x-0.5)*(\x-0.75)+\x);}}

\def\a{0.45}
\def\b{0.80}

\draw[->] (-0.1,0) -- (1,0) node[right] {$x$};
\draw[->] (0,-0.1) -- (0,1) node[above] {$f(x)$};
\draw[] plot (\x,\x) node[right] {};
\draw[line cap=round, color=black,very thick] plot[smooth] (\x,{g(\x)}) node[right] {};

\draw[line cap=round, color=black, very thick, rounded corners, dashed] (\a,{g(\a)}) -- (\a+0.01,{g(\a+0.01)}) -- (0.78,0.74) -- (\b,{g(\b)});
\draw[-] (\a,\a) -- (\a,0) node[below] {$x_1$};
\draw[-] (\b,\b) -- (\b,0) node[below] {$x_2$};
\end{tikzpicture}
\end{figure}

Given a game, we show in Section~\ref{sec:perturbation} that if the index of a neighborhood is zero, then there exists a payoff perturbation of an equivalent game with no equilibrium inside the neighborhood and such that, payoffs \textit{outside} such a neighborhood are as close as we want to those in the original game.
Figure~\ref{fig:perturbation} illustrates.
The interval $(x_1,x_2)$ contains two fixed points in which the function cuts the $45^\circ$ line from below (index $-1$) and then from above (index +1).%
\footnote{~In this one-dimensional case, if $x^*$ is a fixed point, the index of $x^*$ is $\mathrm{Ind}_f(x^*)=\left.\mathrm{sign}\left(\frac{\partial(\mathit{Id}-f)(x)}{\partial x}\right\vert_{x=x^*}\right)$}
We can pull the graph of $f$ so that the new function whose graph coincides with the dashed line when $x\in(x_1,x_2)$ has no fixed point in $(x_1,x_2)$.
Thus, if $G$ is the subgame in Figure~\ref{Entry}, there exists an equivalent game $\bar{G}$ (i.e., a game which is obtained by adding finitely many pure strategies that are duplicates of mixed strategies of $G$) and a perturbation of it $\bar{G}^\delta$ with no equilibria in the (subset of the strategy profile equivalent to the) supporting polytope and such that payoffs outside that polytope are at least $\delta$-close to those in $\bar{G}$. The existence of $\bar{G}^\delta$ relies on the techniques developed in \cite{GW2005}, but it is not an exact application of their result, as our Appendix \ref{sec:proofLemmaExcluded} demonstrates.

With this information, consider the extensive-form game in Figure~\ref{fig:equiv}.
Nature moves first.
With probability $1-\varepsilon$ players play the original entry game but with the subgame $G$ replaced by the equivalent subgame $\bar{G}$.
With probability $\varepsilon$, players play the perturbed game $\bar{G}^\delta$ that does not have an equilibrium in the supporting polytope.
Player 1 observes the move of Nature, but player 2 does not know if she is playing because Nature chose the perturbed game $\bar{G}^\delta$ or because player 1 chose $\mathit{In}$.
This is captured by player 2's information set connecting nodes representing games $\bar{G}$ and $\bar{G}^\delta$.
Note that when $\varepsilon=0$ this extensive-form game is equivalent to the game in Figure~\ref{Entry}.\footnote{More precisely, the game with $\e = 0$ is equivalent to the original game in the sense that both have the same reduced normal forms.}
It can be proved that when $\varepsilon$ is small enough this equivalent game does not have any Nash equilibrium component in which player 1 plays $\mathit{Out}$.
Intuitively, if there were such outcome, player 2 would have to be playing in the (set equivalent to the) supporting polytope - so that 
$\mathit{Out}$ is a best reply for player 1. Knowing that, in such an equilibrium, if player 2 is called to move, then she knows it is because Nature chose $\bar{G}^\delta$. But such a game does not have an equilibrium in the supporting polytope.
We can similarly prove that there is no sequence of equilibria in which player 1's strategy converges to playing $\mathit{Out}$ with probability one as $\varepsilon$ converges to zero.
Hence, when $\varepsilon$ is small enough, Figure~\ref{fig:equiv} gives us a perturbation of a game equivalent to the entry game of Figure~\ref{Entry} with no equilibrium close to the zero-index ``\textit{Out}''-component.

In the rest of the paper we show that under some generic conditions every key step in this example generalizes. 
In particular, given a two-player game and a component that induces a unique outcome, we can always define an ``excluded game'' (which is not necessarily a subgame) whose supporting polytope (capturing behavior in the excluded game that prevents it from being reached) has zero index in such an excluded game.
And with that, we can construct a game equivalent to the original and a perturbation of its payoffs analogous to Figure~\ref{fig:equiv} with the property that the resulting game has no equilibrium close to the zero index component.

\begin{figure}
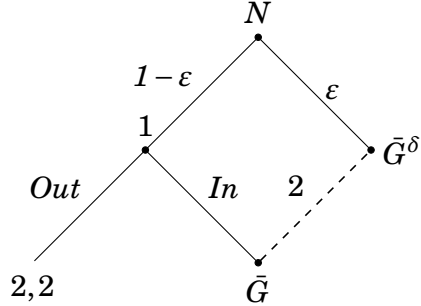

\caption{A game equivalent to that in Figure~\ref{Entry} when $\varepsilon=0$}
\label{fig:equiv}
\begin{istgame}
\xtdistance{15mm}{30mm}
\setistmathTF{1}{1}{1}
\istroot(n)(0,0){N} % names the root as (n) at (0,0)
\istb{\mathit{1-\varepsilon}}[al] % endpoint will be (0-1), automatically
\istbA*<level distance=\xtlevdist>{\mathit{\varepsilon}}[r]{\bar{G}^{\delta}}[r]
\endist % end of simple (parent-child) structure

\istroot(1)(n-1){1}
\istb{\mathit{Out}}[al]{2,2} % endpoint will be (1-1), automatically
\istb*{\mathit{In}}[ar]{\bar{G}}[b] % endpoint will be (0-2), automatically
\endist

\xtInfoset[dashed](1-2)(n-2){2}[al]
\end{istgame}
\end{figure}

\section{Preliminaries}\label{sec:prem}

\subsection{Equivalent strategies, equivalent normal-form games} 
A \textit{finite two-player normal-form game} $\mathbb{G}\equiv(S_1,S_2,\mathbb{G}_1,\mathbb{G}_2)$ is a four-tuple where, for each $n=1,2$, player $n$'s finite set of pure strategies is $S_n$ and $\mathbb{G}_n: S_1\times S_2 \to \Re$ is player $n$'s payoff function.
As usual, $\S_n\equiv\D(S_n)$ is player $n$'s set of mixed strategies and $S\equiv S_1\times S_2$ and $\S\equiv\S_1\times\S_2$ are, respectively, the sets of pure and mixed strategy profiles. 
We also denote by $\mathbb{G}_n$ the multilinear extension of player $n$'s payoff function to the set of mixed strategy profiles $\S$.

Two strategies $\s_n, \s'_n \in \S_n$ are \textit{equivalent} if for both $m=1,2$ we have $\mG_m(\s_n, s_{-n}) = \mG_m(\s'_n, s_{-n})$ for all $s_{-n}\in S_{-n}$.
Given $\mG = (S_1,S_2, \mG_1, \mG_2)$, the \textit{reduced normal-form} of $\mG$ is a normal-form game $\mG' = (S'_1,S'_2,\mG_1,\mG_2)$ that is obtained from $\mG$ by eliminating pure strategies that are equivalent to some existing mixed strategy.
That is, $S'_n \subseteq S_n$ and if $s_n \in S_n \setminus S'_n$ then there is $\s_n \in \S'_n\equiv\Delta(S'_n)$ that is equivalent to $s_n$. Note that, up to relabeling of strategies, the reduced normal-form of a finite game is unique.

\begin{definition} 
Two games $\bar{\mathbb{G}}$ and $\mathbb{G}$ are \textit{equivalent} if they have the same reduced normal-form.
\end{definition}

Given two equivalent games $\mathbb{G}=(S_1,S_2, \mG_1, \mG_2)$ and $\bar{\mathbb{G}}=(\bar{S}_1,\bar{S}_2, \bar{\mG}_1, \bar{\mG}_2)$ we extend the notion of equivalence between strategy profiles in the same game to equivalence of strategy profiles in equivalent games.
Say that $\bar \s_n \in \bar \S_n$ and $\s_n \in \S_n$ are \textit{equivalent} if there exists a strategy $\s'_n$ in their (common) reduced normal-form game such that $\bar\s_n$ is equivalent to $\s'_n$ (both viewed as strategy profiles in $\bar \mG$) and $\s_n$ is equivalent to $\s'_n$ (both viewed as strategy profiles in $\mG$). 
If $\mG$ and $\bar{\mG}$ are two equivalent games we say profile $\s \in \S$ is equivalent to profile $\bar{\s} \in \bar{\S}$ if for each player $n$, the strategy $\s_n$ is equivalent to $\bar{\s}_n$.  
A subset $T\subset\S$ is \textit{equivalent} to $\bar{T}\subset\bar{\S}$, if for each $\s \in T$ there exists $\bar{\s}\in\bar{T}$ such that $\s$ is equivalent to $\bar{\s}$ and for each $\bar{\s}\in\bar{T}$ there exists $\s \in T$ such that $\bar{\s}$ is equivalent to  $\s$.

\subsection{Extensive-form games}\label{extensiveformgames}

We introduce notation and basic definitions.
For a formal definition of extensive-form game with perfect recall see, e.g., \citet{OR1994}.
Consider a  \textit{two-player finite game tree with perfect recall} $\Gamma \equiv  (T,\prec,U,A,\rho)$.
The set of nodes is $T$ and $\prec$ denotes the precedence relation in the tree. 
The set of terminal nodes is $Z \subset T$.
The collection $U$ is a partition of $T \setminus Z$ into information sets of players and Nature.
The set $U_n \subset U$ is the collection of information sets for player $n=0,1,2$, where player 0 represents Nature. 
(Every element in $U_0$ is a singleton.)
The set of all actions in the game is $A$ and $A_n(u)$ represents player $n$'s set of actions available at her information set $u \in U_n$.
Let $A_n\equiv\bigcup_{u\in U_n}A_n(u)$ be the entire set of player $n$'s actions.
Moves of Nature are given by the function $\rho$ that, to every $u\in U_0$, assigns a completely mixed probability distribution on $A_0(u)$. 

For $n=1,2$, player $n$'s set of pure strategies is $S_n\equiv\big\{s : U_n \rightarrow A_n \mid s(u) \in A_n(u) \big\}$ and her mixed strategy set is $\Sigma_n \equiv \Delta(S_n)$. 
Hence, we have the sets of strategy profiles $S\equiv S_1\times S_2$ and $\S=\S_1\times\S_2$.
Since we only consider games with perfect recall, Kuhn's theorem implies that we can work with behavior strategies whenever convenient. Given player $n=1,2$, a behavior strategy $b_n = (b_n(i))_{i \in A_n}$ satisfies $(b_n(i))_{i\in A_n(u)}\in\Delta(A_n(u))$ for every $u\in U_n$. 
We let $B_n$ be player $n$'s set of behavioral strategies.

For a fixed two-player game tree $\Gamma$, the space of games is $\mathcal{G} \equiv \mathbb{R}^{2\vert Z\vert}$.
A game $G \in \mathcal{G}$ assigns payoff $G_n(z)$ to player $n$ at final node $z$. We refer to the extensive-form game defined by $\Gamma$ and terminal payoffs $G$ simply by $G$.
The space of outcomes is $\Delta(Z)$, where an outcome $Q\in\Delta(Z)$ assigns probability $Q(z)$ to $z$. 

Given action $a\in A$ the set of terminal nodes that come after action $a$ is $Z(a)$.
Similarly, given an information set $u\in U$ we let $Z(u)$ be the set of terminal nodes that come after some action available at $u$.
We denote expressions $Q(Z(a))$ and $Q(Z(u))$ simply as $Q(a)$ and $Q(u)$.  
Given a profile of mixed strategies $\s\in\S$ the induced outcome when players play according to $\s$ is $\mP(\cdot \mid \s) \in \D(Z)$.
With abuse of notation we also let $G:\Delta(Z)\to\mathbb{R}$ denote the expected utility function associated with game $G\in\mathcal{G}$.
The normal-form representation of $G$ is $\mG=(S_1,S_2,\mG_1,\mG_2)$ where, for each $n=1,2$, we have $\mG_n(s)=G_n(\mathbb{P}(\cdot\mid s))$ for all $s\in S$.

\subsection{Index Theory}\label{indextheory}
The fixed point index contains information about the robustness of fixed points of a map when such a map is perturbed to a nearby map. 
Since Nash equilibria are fixed points, we can apply index theory to them (cf. \citealp{KR1994}).
The classical introduction to index theory usually requires some concepts from algebraic topology. 
For the purposes of this paper, this can be bypassed without much hindrance. 
Results in this section can be found in \citet[pp 245-265]{AM2018} and \citet{LP2023}.

Let $\mG$ be a normal-form game with mixed strategy $\S$. Given a neighborhood $O$ of $\S$ suppose $f: O \to \S$ is a differentiable map. Let $d_{f}$ be the displacement of $f$, i.e., $d_{f}(\s) = \s - f(\s)$. 
Then the fixed points of $f$ are the zeros of $d_{f}$. Suppose now that the Jacobian of $d_f$ at a zero $\s$ of $f$ is nonsingular. We assign the index $+1$ to $\s$ if the determinant of the Jacobian of $d_{f}$ is positive or $-1$ if such a determinant is negative. For the next definition, given a subset $\mathcal{O}$ of $\S$, we denote by $\text{cl}_{\S}(\mathcal{O})$ the closure of $\mathcal{O}$ with respect to $\S$. %

\begin{definition}
An open neighborhood $\mathcal{O} \subset \S$ of a component of equilibria (in mixed strategies) $K$ of normal-form game $\mG$ is \textit{admissible} if every equilibrium of $\mG$ in $\text{cl}_{\S}(\mathcal{O})$ belongs to $K$. 
\end{definition}

When a finite game $\mG$ has a component of equilibria $K$ that consists of more than a single strategy profile, we extend the definition of the index as follows.
Take a continuous map $g_{\mathbb{G}}: \S \to \S$ such that the fixed points of $g_{\mathbb{G}}$ are the Nash equilibria of game $\mG$ and $g_{\mathbb{G}}$ continuously depends on the payoffs of $\mathbb{G}$.
An example of such a map is the map~\cite{JN1951} constructed to prove existence of equilibria in mixed strategies. 
Consider a neighborhood $O$ of $\S$ and $r: O \to \S$ a retraction to the closest point in $\S$. 
Let $\mathcal{O} \subset O$ be an open neighborhood of $K$ in the affine space generated by $\S$, whose closure contains no other fixed point of $g$. 
We approximate $(g_{\mathbb{G}} \circ r)$ uniformly by a differentiable function $f: O \to \S$ without fixed points on the boundary of $\mathcal{O}$ and such that the displacement of $f$ at any fixed point has nonsingular Jacobian. For any sufficiently close approximation, the sum of the indexes of the (isolated) fixed points of $f$ in $\mathcal{O}$ is constant and can be taken as the definition of the index of the component $K$.\footnote{Under the continuous dependence of the payoffs of $\mathbb{G}$, this definition of index is independent of the particular map $g_\mathbb{G}$ used, it only depends on the game $\mG$ (cf. \cite{DG2000}).}

We can now define the index with respect to the best-reply correspondence of game $\mG$ following the same procedure as in \citetalias{GW2005}. Consider now $\mathcal{O} \cap \S$. For notational convenience, we will denote this intersection from now on as $\mathcal{O}$. 
Let $W$ be an open neighborhood of $\mathrm{Graph}(\BR^{\mG})$ such that $W \cap \big\{(\sigma, \sigma) \in \S \times \S \mid \sigma \in \text{cl}_{\S}(\mathcal{O}) - \mathcal{O}\big\} = \emptyset$. 
There exists $\tilde W \subset W$ a neighborhood of $\text{Graph}(\BR^{\mG})$ such that any two continuous functions $f_0$ and $f_1$ from $\S$ to $\S$ whose graphs are in the neighborhood $\tilde W$ are homotopic by a homotopy $H: [0,1] \times \S \rightarrow \S$ with $\text{Graph}(H(t,\cdot)) \subset \tilde W$ for all $t \in [0,1]$ (cf. \citealp{AM1989}).
Take a continuous map $f : \S \rightarrow \S$ with $\text{Graph}(f) \subset \tilde W$. 
We define the \textit{best-reply index} of component $K$, denoted $\text{Ind}_{\BR^{\mG}}(K)$, as the fixed point index of the continuous map $f\vert_{\mathcal{O}}: \mathcal{O} \rightarrow  \S$. 
The choice of the neighborhood $W$ and the homotopy property of the index (see \citealp{D1972}, Chapter VII, 5.15) imply that the index of the component is the same for any continuous map with graph in the neighborhood $\tilde W$. 
We note that defining the index of a component from, say, the map that Nash used in \citep{JN1951} or from the best-reply correspondence are two distinct ways of defining the index, which can be shown to be equivalent (cf. \citetalias{GW2005}).
In addition, this process to define the index applies with insignificant changes if the correspondence between the simplices of strategies is contractible valued instead of convex valued (cf. \citealp{AM1989}).%
\footnote{~A topological space $X$ is contractible if there exists a continuous map $T: [0,1] \times X  \to X$ and $x_0 \in X$ such that $T(0,\cdot) =\mathrm{id}_{X}$ and  $T(1, \cdot) = x_0$.} 
This fact will play an important role in Section~\ref{sec:step1}, as we need to consider selections of a best-reply correspondence which are not necessarily convex-valued, but are contractible-valued.

One can generalize the definition of the best-reply index of a component of equilibria to the \textit{best-reply index of an admissible neighborhood}, by using the exact same procedure as in the previous paragraph. This yields the index of a neighborhood by summing the indexes of the components of equilibria which are contained in it. For convenience, whenever we refer to the index of a component or an admissible neighborhood, it will be implicit that we refer to the best-reply index. 

We are now ready to recall a few known properties of the index of equilibria which we will use in the proof of our main result. The proofs that the index satisfies such properties can be found in \cite{AM2018} or \citetalias{GW2005}.

\begin{enumerate}[start=1,label={\bfseries I.\arabic*}]

\item\label{I1}\textit{The index of an admissible neighborhood is locally constant with respect to payoff perturbations of the game}. 
Formally, fix an admissible neighborhood $O$ in the mixed strategy set of a finite game $\mG$. 
Then there exist $\bar\d>0$, such that for any $0 \leq \d \leq \bar \d$ and a $\d$-payoff-perturbation $\mG^\d$ of game $\mG$, the index of $O$ (with respect to $\mG^\d$) is constant.

\item\label{I2} \textit{The index of an equilibrium component is invariant to equivalent presentations of a game}. 
If $K$ is a component of equilibria of $\mG$ with index $c$, then for any equivalent game $\bar \mG$, the index of the equivalent component $\bar K$ is also $c$.

%\medskip

\item\label{I3} \textit{The index of a component is invariant to the deletion of strictly inferior replies to the component}. 
If $K$ is a component of equilibria with index $c$ of game $\mG$, then deleting from the normal-form of game $\mG$ the pure strategies of player $n$ which are strictly inferior replies to every profile in the component $K$ yields a new game $\mG'$ with the same component $K$ as an equilibrium component in $\mG'$ and with the same index $c$. 
\end{enumerate} 

We also need three well-known properties of the index. 
For our purposes, we particularize their statements as follows: 
%\medskip
\begin{enumerate}[start=4,label={\bfseries I.\arabic*}]
\item\label{I4} \textit{Multiplication}: 
Let $\BR^{\mG}: \S \rightrightarrows \S$ and $\BR^{\mG'}: \S' \rightrightarrows \S'$ be the best-reply correspondences of (respectively) games $\mG$ and $\mG'$.
Let $\BR^{\mG} \times \BR^{\mG'}$ be the correspondence taking $(\s,\t) \in \S \times \S'$ to $\BR^{\mG}(\s) \times \BR^{\mG'}(\t)$. 
If $O \times O' \subset \S \times \S'$ contains no fixed points of $\BR^{\mG} \times \BR^{\mG'}$ in its boundary, then $O$ (respectively $O'$) has no fixed points of $\BR^{\mG}$ (respectively $\BR^{\mG'}$) in its boundary, and the index of $O \times O'$ with respect to  $\BR^{\mG} \times \BR^{\mG'}$ is the multiplication of the indexes of $O$ (with respect to  $\BR^{\mG}$) and of $O'$ (with respect to  $\BR^{\mG'}$).

\item\label{I5} \textit{Commutativity}: 
Let $F: \S \rightrightarrows \S$ be an upper hemicontinuous correspondence that is  nonempty, compact and convex valued.
Let $e:\S\to\S'$ be a continuous map with left-inverse $q:\S'\to\S$.
If $X$ is a component of fixed points of $e\circ F\circ q : \S' \rightrightarrows \S'$ then $q(X)$ is a component of fixed points of $F$ and their indexes agree.

\item\label{I6} \textit{Excision}: 
Let $F: \S \rightrightarrows \S$ be an upper hemicontinuous correspondence that is  nonempty, compact and convex valued.
Suppose $\tilde{\mathcal{O}}$ and $\mathcal{O}$ are both admissible neighborhoods in $\S$ whose closures in $\S$ contain the same fixed points of $F$. Then the index of $\mathcal{O}$ and that of $\tilde{\mathcal{O}}$ with respect to $F$ are identical.

\end{enumerate}

We conclude this section with the following result.
In particular, it implies that if a component has zero index, then some ending node is reached with probability zero.

\begin{proposition}\label{nonzeroindex}
If an equilibrium outcome $Q$ induced by a component of equilibria in mixed strategies $K$ has full support, then $K$ has non-zero index.  
\end{proposition}

\begin{proof}
See Appendix~\ref{app:nonzeroindex}.
\end{proof}

\subsection{Main Result}\label{sec:main}

We recall the definition of Hyperstability. 

\begin{definition}
A component $K$ of equilibria in mixed strategies of a normal-form game $\mG$ is \textit{hyperstable} if for each equivalent game $\bar \mG$ and every $\e>0$, there exists $\d>0$ such that any $\d$-payoff-perturbation $\bar \mG^{\d}$ of $\bar \mG$ has an equilibrium which is $\e$-close to $\bar K$ (the equivalent component in $\bar \mG$ to $K$). 
\end{definition}

Hyperstability is a property first formulated in \cite{KM1986} to refine Nash equilibria.
It implies several desirable properties.%
\footnote{~Hyperstable components of equilibria always contain proper (and therefore sequential) equilibria, satisfy invariance, and are robust to payoff perturbations.}
In Section~\ref{sec:proof} we prove the following result.

\begin{theorem}\label{mainthm}
Fix a two-player game tree with perfect recall.
Apart from a lower-dimensional, semi-algebraic set of payoffs, an equilibrium component is hyperstable if and only if it has non-zero index. 
\end{theorem}

\citetalias{GW2005} also offer a characterization of non-zero index Nash equilibrium components.
However, their characterization requires a strengthening of hyperstability.
Namely, \citetalias{GW2005} show that a component has non-zero index if and only if it is \textit{uniformly hyperstable}, that is, if for each $\e>0$, there exists $\d>0$ such that for any equivalent game $\bar \mG$ and any $\d$-perturbation $\bar \mG^{\d}$, there exists an equilibrium of $\bar \mG^{\d}$ which is within $\e$ of $\bar K$ (the equivalent component of $\bar \mG$). 
Note that $\d$ is fixed across all equivalent games, making uniform hyperstability a stronger concept than the combination of robustness to payoff perturbations and invariance. Nonetheless, our result implies that for generic two-player extensive-form games, a component is uniformly hyperstable if and only if it is hyperstable.

\section{Excluded Games}\label{sec:excluded}
The ``$\mathit{Out}$''-component in the game in Figure~\ref{Entry} has player 1 excluding the subgame.
In turn, behavior in the subgame as prescribed by such a component prevents player 1 from profiting by deviating and choosing $\mathit{In}$.
In general, observed behavior crucially depends on the robustness of unobserved behavior.
As we discussed in the Introduction, the reason the ``$\mathit{Out}$''-component has zero index is that it specifies a set of strategy profiles in the subgame that, when analyzed relative to such a subgame, has itself zero index.
In this section, we generalize this insight along with the notion of ``excluded game'' so that, given an equilibrium component, we can examine the interaction between the two players that occurs both on-path and off-path.

Take a two-player finite game tree with perfect recall $\Gamma$ and let $\mathcal{G}=\Re^{2|Z|}$ be the space of payoffs over terminal nodes. 
For payoffs $G \in \mathcal{G}$, let $\mathbb{G}$ be the induced normal form.
From now on, fix an equilibrium component $K$ that induces a unique outcome $Q \in \D(Z)$. 
Given the equilibrium component $K$, an information set $u$ is said to be \textit{on-path} if $Q(u)>0$.
Player $n$'s collection of information sets that are on-path is $U_n^+\equiv\{u_n\in U_n\mid Q(u)>0\}$.
Similarly, an information set is \textit{off-path} if $Q(u)=0$.
Player $n$'s collection of information sets that are off-path of $Q$ is $U_n^0\equiv U_n\setminus U^+_n$.

Let $Z^0$ be the set of terminal nodes which have probability $0$ under $Q$. 
We define
\begin{equation*}
S^n_n \equiv \big\{ s_n \in S_n \mid \text{ there exists } \s \in K\text{ with } \mP(Z^0 \mid s_n, \s_{-n} ) >0 \big\}
\end{equation*}
and call it player $n$'s set of \textit{observable deviations}. 
It consists of those player $n$'s pure strategies in which, at some information set on-path according to $K$, player~$n$ deviates and plays an action that has zero probability under $Q$.
As usual, $\S^n_n\equiv\D(S^n_n)$. 

Let 
\begin{equation}\label{eq:strat_pos}
S^{+}_n\equiv \Big\{s^+_n:U^+_n\to A_n\mid s^{+}_n(u)\in A_n(u)\text{ for every }u\in U^+_n\Big\}
\end{equation}
and $\S^+_n\equiv\Delta(S^+_n)$. Note that these are not pure strategies because they assign choices only to on-path (of $K$) information sets.

The subset $S^{\otimes}_n \subset S^+_n$ is the collection of assignments of actions at on-path information sets that only assign actions that are taken in the equilibrium component $K$, that is
\begin{equation*}
S^{\otimes}_n\equiv \Big\{s^{\otimes}_n\in S^+_n\mid Q(s^{\otimes}_n(u))>0\text{ for every }u\in U^+_n \Big\}, 
\end{equation*}

where the symbol $Q(s^{\otimes}_n(u))$ follows the notational convention established in section \ref{extensiveformgames}.

Define $\S^\otimes_n\equiv\Delta(S^\otimes_n)$ and $\S^\otimes\equiv \S^\otimes_1\times\S^\otimes_2$.
Analogously to \eqref{eq:strat_pos}, define also
\begin{equation}\label{eq:strat_zero}
S^0_n\equiv\Big\{s^{0}_n:U^0_n\to A_n\mid s^{0}_n(u)\in A_n(u)\text{ for every }u\in U^0_n\Big\}
\end{equation}
together with $\S^0_n\equiv\Delta(S^0_n)$.

We will now split the mixed strategies of a player between an on- and off-path part. Given some player $n$'s pure strategy $s_n\in S_n$ the restriction of $s_n$ to information sets in, respectively, $U^+_n$ and $U^0_n$ is $q_n^{+}(s_n)$ and $q_n^0(s_n)$. These define maps $q^+_n: S_n \to S^{+}_n$ and $q_n^0: S_n \to S^{0}_n$. Linear interpolation extends $q^+_n$ to a map from $\S_n$ to $\S^+_n$. Analogously, $q^0_n$ is extended to a map from $\S_n$ to $\S^0_n$. Given a mixed strategy $\s_n \in \S_n$, $q^+_n(\s_n)$ is, therefore, the marginal of $\s_n$ over $S^+_n$; and $q^0_n(\s_n)$ is the marginal of $\s_n$ over $S^0_n$. 
In particular, if $s_n \in S_n \setminus S^n_n$, then $s_n$ can be written as a pair $(s^{\otimes}_n, s^0_n) \in S^{\otimes}_n \times S^0_n$ and we let $q^{\otimes}_n(s_n) = s^{\otimes}_n$. Similarly to above, $q^{\otimes}_n: \D(S_n \setminus S^n_n) \to \S^{\otimes}_n$ is the affine map for which $q^{\otimes}_n(\s^{\otimes}_n)$ is the marginal of $\s_n \in \D(S_n \setminus S^n_n)$ over $\S^{\otimes}_n$. Let $q^+\equiv q^+_1\times q^+_2$, $q^\otimes\equiv q^\otimes_1\times q^\otimes_2$, and $q^0\equiv q^0_1\times q^0_2$.

Given any pair $(\s^+_n,\s^0_n)\in\S^+_n\times\S^0_n$ let $\s^+_n*\s^0_n$ represent the product strategy $\s_n$ that satisfies $\s_n(s^+_n,s^0_n)=\sigma^+_n(s^+_n)\cdot\sigma^0_n(s^0_n)$ for every $s_n=(s^+_n,s^0_n)$.

If $\s_n$ is a product strategy then, by definition, there is a pair $(\s^+_n, \s_n^0) \in \S^+_n \times \S^0_n$ such that $\s_n = \s^+_n*\s^0_n$.

Obviously, not every mixed strategy is a product strategy over $S^{+}_n \times S^0_n$.
Nonetheless, for any mixed strategy of a player there exists an outcome-equivalent product strategy. 
This observation is formalized below and follows directly from Kuhn's Theorem.

\begin{lemma}\label{lm:kuhn}
For every $\s_n\in\S_n$ there exists a pair $(\hat{\s}^+_{n},\hat{\s}^0_{n})\in\S^+_{n}\times\S^0_{n}$ such that, given any $\s_{-n} \in \S_{-n}$, the strategy $\hat{\s}^+_{n}*\hat{\s}^0_{n}$ induces the same outcome as $\s_n$.

\end{lemma}

\begin{definition}[Excluded game]\label{def:excluded}
Fix $\hat\t\in K$ such that $\hat\tau_n=\hat\tau^\otimes_n*\hat\tau^0_n$ with $\hat\tau^\otimes_n\in\mathrm{int}(\S^\otimes_{n})$, where the interior is relative to the affine space generated by $\S^\otimes_{n}$.
Player $n$'s \textit{excluded game} is the normal-form game in which player $n$'s strategy set is $S_n^n$, player $-n$'s strategy set is $S^0_{-n}$ and the payoff function for player $m=1,2$ is defined for each profile $(s^n_n,s^0_{-n})\in S^n_n\times S^0_{-n}$ as
\begin{equation}\label{payoffsexcluded}
{\mG}^n_m(s^n_n, s^0_{-n}) \equiv \sum_{ s^{\otimes}_{-n} \in S^{\otimes}_{-n}}\hat{\t}^{\otimes}_{-n}({s}^{\otimes}_{-n})\mG_m(s^n_n, s^{\otimes}_{-n}, s^0_{-n}).
\end{equation}
\end{definition}
Since the equilibrium component $K$ induces a unique probability distribution $Q$, the payoff function~\eqref{payoffsexcluded} does not depend on the particular choice of $\hat\tau\in K$, on the fact that $\hat\t_n$ is a product (as $q_n^{\otimes}(\hat\t_n)$ can be used instead), or on the fact that $\hat\tau^\otimes_n$ is chosen (for convenience) to be in the interior of $\S^\otimes_{n}$.

\begin{remark}
We note two particular cases. First, when no information set of player $-n$ is off-path in component $K$ then player $n$'s excluded game is just a decision problem. Second, if $K$ is such that player $n$ does not have any observable deviation, as player 2 in the ``\textit{Out}''-component of Figure~\ref{Entry}, then we say that player $n$ does not have an excluded game.
Also note that while in Figure~\ref{Entry} player 1's excluded game coincides with the proper subgame, in general, an excluded game is not necessarily a subgame. 
\end{remark}

Let $K^n_{-n}$ be the subset of player $-n$'s mixed strategies in player $n$'s excluded game such that player $n$ obtains a payoff that is no larger than what she obtains under the equilibrium component $K$:
\begin{equation*}
K^n_{-n}\equiv\Big\{\s^0_{-n}\in \S^0_{-n} \mid\text{ for all }s^n_n \in S^n_n\text{ we have } \mG^n_n(s^n_n, \s^0_{-n}) \leq G_{n}(Q) \Big\}.
\end{equation*}
And let $\tilde\partial	K^n_{-n}$ be the subset of $K^n_{-n}$ where player $n$ has at least one deviation for which she is indifferent:
\begin{equation*}
\tilde\partial	K^n_{-n}\equiv \Big\{ \s^0_{-n} \in K^n_{-n} \mid \text{ for some }s^n_n \in S^n_n\text{ we have }\mG^n_n(s^n_n, \s^0_{-n}) = G_{n}(Q) \Big\}.
\end{equation*}
The set $\tilde{\partial}K^n_{-n}$ contains the boundary of $K^n_{-n}$ relative to $\S^0_{-n}$ but, in general, it is not equal to it. 
Player $n$'s \textit{supporting polytope} is $K^n\equiv\S^n_n \times K^n_{-n}$. 
We also define $\tilde\partial K^n\equiv\S^n_n \times\tilde\partial K^n_{-n}$. Later we show that under a generic choice of terminal payoffs for the game tree $K^{n}$ is a full dimensional neighborhood in the strategy set of the excluded without equilibria in its boundary realtive to this strategy set, thus admiting an index.

\begin{example}\label{ex:A}

In Figure~\ref{Entry}, the supporting polytope $K^1$ to the ``Out'' component is $\D(\{L,R\})\times\big\{\a\ell+(1-\a)r\mid 0\leq \a\leq2/3\big\}$.
\end{example}

\begin{remark}
 
Proposition \ref{nonzeroindex} implies that if a zero-index component of equilibria in mixed strategies $K$ induces a unique outcome then such an outcome must not have full support. 
This result does not require any genericity assumption and it is a consequence of the two-player environment. 
Furthermore, assume that the component $K$ is such that for both $n=1,2$ we have $\tilde\partial K^{n} = \emptyset$. 
Then every strategy $s^n_n \in S^n_n$ pays strictly less than $G_n(Q)$ against any equilibrium strategy of player $-n$.
Hence, we can eliminate the deviations $s^n_n$ from $\mG$ for both players leaving the index of component $K$ invariant (cf. Property \ref{I3}). 
The resulting game obtained after this elimination has an extensive-form where the outcome $Q$ induced by $K$ is completely mixed. 
Proposition \ref{nonzeroindex} then implies that $K$ has non-zero index. 
Therefore, any zero-index equilibrium component $K$ is such that $\tilde\partial K^{n}\neq \emptyset$ for some $n=1,2$.
\end{remark}

\begin{definition}[Included game]
The \textit{included game} $\mG^{\otimes}$ associated with component $K$ is the normal-form game in which player $n$'s strategy set is $S_n^{\otimes}$ and payoff functions for player $m=1,2$ are given by:
\begin{equation*}
{\mG}^\otimes_m(s^\otimes_n, s^\otimes_{-n}) \equiv \mG_m(s_n, s_{-n}), 
\end{equation*}
where $(s_n,s_{-n}) \in (q^{\otimes}_n\times q^{\otimes}_{-n})^{-1}(s^{\otimes}_{n},s^{\otimes}_{-n})$. 
\end{definition}

Note that the payoffs in $\mG^{\otimes}$ are well-defined since they do not depend on the particular choice of $(s_n, s_{-n})\in (q^{\otimes}_n\times q^{\otimes}_{-n})^{-1}(s^{\otimes}_{n},s^{\otimes}_{-n})$.

\section{Genericity Assumptions}\label{SEC:GenAssump}
In Section~\ref{sec:proof} we fix a game $G$ with game tree $\Gamma$ and a zero-index component $K$ and prove that $K$ is not hyperstable provided it satisfies the following two assumptions:

\begin{enumerate}[start=1,label={\textbf{A.\arabic*}}]

\item\label{P1} The outcome associated to $K$ is unique.
Moreover, after eliminating all branches and nodes from $\Gamma$ which have zero probability under the equilibrium outcome induced by $K$, the set $q^\otimes(K)$ is a component of equilibria of $G^{\otimes}$.

\item\label{P2} If  $K$ induces a unique outcome in which player $n \in \{1,2\}$ has an excluded game, $K^n$ is a full-dimensional polytope in $\S^n$ and every equilibrium payoff of $\mathbb{G}^n$ to player $n$ in $K^n$ is strictly lower than player $n$'s equilibrium payoff induced by $K$.

\end{enumerate}

Assumption \ref{P1} is a standard property of generic extensive-form games (see \citealp{KW1982} or \citealp{GW2001}).
For our purposes, note that the definition of excluded game relies on the component inducing a unique outcome.
The new assumption is \ref{P2}. The statement concerning the full-dimensionality of polytope $K^n$ in $\S^n$ is satisfied generically.%
\footnote{~For generically chosen terminal payoffs of a fixed game tree, any equilibrium component in mixed strategies induces a unique outcome and contains an equilibrium profile for which any observable deviation by any player is a strictly inferior reply to that profile (cf. the first paragraph of the proof of Theorem 4.2 in \citealp{GW2002}). 
This guarantees full-dimensionality of $K^{n}$.} Assumption $\ref{P2}$ as a whole tells us that $K^n$ is the closure of an admissible neighborhood in $\S^n$ and, therefore, admits an index.  It is a technical assumption in order to well-define the index of the neighborhood.

The next example shows that \ref{P1} does not imply \ref{P2}.

\begin{example}

\begin{figure}
\caption{Game G}\label{game3}
\begin{minipage}{0.50\textwidth}
\centering
\begin{istgame}[scale=0.8, every node/.style={scale=0.8}]
\xtdistance{15mm}{45mm}
\setistmathTF{1}{1}{1}
\istroot(o)(0,0){1} % names the root as (n) at (0,0)
\istb{T}[al] % endpoint will be (0-1), automatically
\istb{B}[ar]
\endist % end of simple (parent-child) structure

\xtdistance{15mm}{30mm}

\istroot(n)(o-2){N}
\istb{\frac{1}{2}}[al]% endpoint will be (n-1), automatically
\istb{\frac{1}{2}}[ar]% endpoint will be (n-2), automatically
\endist

\xtdistance{15mm}{15mm}

\istroot(p21)(o-1)
\istb{\ell}[al]{1,0}% endpoint will be (n-1), automatically
\istb{r}[ar]{0,0}% endpoint will be (n-2), automatically
\endist

\istroot(p22)(n-1)
\istb{\ell}[al]{0,0}% endpoint will be (n-1), automatically
\istb{r}[ar]{1,1}% endpoint will be (n-2), automatically
\endist

\istroot(p23)(n-2){2}
\istb{a}[al]{2,2}% ORIGINALLY {x,2}
\istb{b}[ar]{3,0}% endpoint will be (n-2), automatically
\endist

\xtInfoset[dashed](p21)(p22){2}[ar]
\end{istgame}
\end{minipage}
\hfill
\begin{minipage}{0.45\textwidth}
\centering
%player 1's excluded game in the ``$T$''-component:
Excluded game $\mG^1$ in the ``$T$''-component:

\setlength{\extrarowheight}{.1em}
\begin{tabular}{c|c|c|}
\multicolumn{1}{c}{}&\multicolumn{1}{c}{$\ell a$}&\multicolumn{1}{c}{$\ell b$}  \\ \cline{2-3}
$B$&$1,1$&$3/2$,0 \\ \cline{2-3}
\end{tabular}

\end{minipage}
\end{figure}

Consider the game in Figure \ref{game3}.

The game has two equilibrium outcomes, each of which is associated to two different equilibrium components.
In the first one, player 1 plays $T$ and player 2 plays $\ell a$. 
In the second, player 1 plays $B$ and player $2$ plays $ra$. 
This game satisfies \ref{P1} but does not satisfy \ref{P2}.
Recall that in the excluded game $\mG^1$ of player 1 associated to the first component, the non-deviating player 2 plays, at  all information sets which are on equilibrium-path the exact equilibrium distribution - therefore, player 2 must play $\ell$ in the leftmost information set. In this excluded game, the deviating player 1 is a dummy player (she only has one strategy $S^1_1 = \{B\}$) and player 2 has two strategies (i.e., $S^1_2 = \{ \mathit{\ell a}, \mathit{\ell b}\}$).

If player 2 plays $\ell a$, the payoff in the excluded game is $(1,1)$, and if player 2 plays $\mathit{\ell b}$, the payoff is $(3/2,0)$. 
Therefore, $(B,\mathit{\ell a})$ is the obvious equilibrium of the excluded game $\mG^1$, which gives player 1 the equilibrium payoff of the component, violating \ref{P2}. 

\end{example}

The next proposition implies that restricting to games that satisfy Assumptions~\ref{P1} and \ref{P2} is a mild constraint.

\begin{definition}\label{def:genericity}
A subset $\mathcal{G}' \subseteq \mathcal{G}$ is \textit{generic} if its complement is a lower-dimensional semi-algebraic set.  
\end{definition}

\begin{proposition}\label{prop:genericity}
The subset $\mathcal{G}' \subseteq \Re^{N|Z|}$ for which properties \ref{P1} and \ref{P2} are satisfied is generic.  
\end{proposition}
\begin{proof}
It is well-known that Assumption~\ref{P1} is generic in the space of extensive-form games \citep{KW1982,GW2001}. The proof that generic two-player games satisfy Assumption~\ref{P2} is relegated to Appendix~\ref{app:genericity}.
\end{proof}

\color{black}

\section{Proof of Theorem~\ref{mainthm}}\label{sec:proof}
It is already known that a non-zero index component is hyperstable.
It can be proved by observing that equivalent components have the same index and that a component with non-zero index is robust to sufficiently small payoff perturbations.

\footnote{~See, e.g., \citetalias{GW2005}.}
Henceforth, we focus on showing that a hyperstable component must have non-zero index. 

The strategy of the proof is to show that if component $K$ has zero index then $K$ is not hyperstable. 
Thus, we need to find a game $\bar \mG$ equivalent to $\mG$ and a neighborhood $\bar V$ of $\bar K$ (the component in $\bar \mG$ equivalent to $K$) such that for any $\varepsilon>0$, there is an $\varepsilon$-payoff-perturbation $\bar \mG^{\varepsilon}$ of $\bar \mG$ with no equilibrium in $\bar V$. 
We proceed in three steps.
In Step 1, we modify the original game $\mG$ to an auxiliary game to show that for some $n=1,2$, the supporting polytope $K^{n}$ has zero index in the excluded game $\mG^n$.
In Step 2, we construct a game equivalent to $\mG^n$ and a perturbation so that the perturbed game does not have an equilibrium in the subset of the strategy profiles equivalent to $K^{n}$.
This is used in Step 3 to construct a game equivalent to the original game $\mG$ and a perturbation that shows that the original zero index component $K$ is not hyperstable. Prior to each step we present the key conceptual details about the step and overall strategy of the argument.

\subsection{Step 1: For some player \texorpdfstring{$n$}{n}, the supporting polytope has index 0 in \texorpdfstring{$\mathbb{G}^n$}{Gn}}\label{sec:step1}

From Proposition~\ref{nonzeroindex} we know that at least one player has an excluded game.
We proceed under the assumption that both players have an excluded game as the proof can be easily adapted to the case in which only one has it. Observe that $K^n$ is full-dimensional in $\S^n$ and admissible with respect to $\mG^n$ (cf. \ref{P2}). Therefore, it has a well-defined index in $\mG^n$. We also assume that for each player $n=1,2$ the supporting polytope $K^n$ contains an equilibrium of ${\mG}^n$.
If for some $n=1,2$ the excluded game ${\mG}^n$ does not have an equilibrium in $K^n$ then such a supporting polytope has index $0$ and we can move to Section~\ref{sec:perturbation}.

We now prove that there exists $n=1,2$ such that $K^n$ has index $0$ in ${\mG}^n$.
The objective is to use the multiplication property of the index (cf. \ref{I4}) to express the index of $K$ as the product of the indexes of $K^1$ and $K^2$ in their corresponding excluded games so that if the index of $K$ is zero then either the index of $K^1$ or the index of $K^2$ is zero as well.
With that in mind, we perturb the original game so that each player is forced to play an observable deviation (i.e. an element of the strategy set in their excluded games) with vanishing probability. 
We then show that the non-deviating player can effectively best-reply by separately responding to the event in which the other player is forced to play as in her excluded game, and to the event in which the other player plays according to the equilibrium path defined by the component $K$.
This defines a selection of the best-reply that, after some manipulations invoking the properties of the index, can be expressed as a product of the best replies in $\mG^{\otimes}$, $\mG^{1}$, and $\mG^{2}$.

%As observed above, recall that every $s_n \in S_n \setminus S^n_n$ can be written as $s_n = (s^{\otimes}_n, s^0_n) \in S^{\otimes}_n \times S^{0}_n$.

Let us define the auxiliary extensive-form game $G^{\e_1,\e_2}$ in Figure~\ref{Auxgame2normalform}.
First, Nature randomizes between three states, $\theta_0$ with probability $1-\e_1-\e_2$, $\theta_1$ with probability $\e_1$, and $\theta_2$ with probability $\e_2$. 
For each $n=1,2$, player $n$ observes if Nature has selected $\theta_n$ but cannot distinguish between $\theta_0$ and $\theta_{-n}$.
If Nature chooses $\theta_n$ then player $n$ chooses a member of $S^n_n$.
In turn, if Nature chooses either $\theta_0$ or $\theta_{-n}$ then player $n$, without observing player $-n$'s move, chooses an element from $S_n$.
Payoffs are inherited from the original game $\mG$.

\begin{figure}[t]
\caption{Game ${G}^{\e_1, \e_2}$}\label{Auxgame2normalform}

\vspace{1em}

\begin{istgame}[scale=0.8, every node/.style={scale=0.8}]
\xtdistance{20mm}{60mm}
\istroot(n1){Nature}
\istb{P(\theta_1)=\varepsilon_1}[al]
\istb{P(\theta_0)=1-\varepsilon_1-\varepsilon_2}[fill=white]
\istb{P(\theta_2)=\varepsilon_2}[ar]%{\bar{S}_1}[[yshift=-25mm]below] 
\endist

\xtdistance{20mm}{14mm}
\istroot(11)(n1-1){1}
\istb
\istb[white]{\dots}[black]
\istb[thick]
\istb
\endist
\xtActionLabel(11)(11-3){s^1_1\in S^1_1}[yshift = -10pt, fill=white]

\istroot(12)(n1-2){}
\istb
\istb[thick]
\istb[white]{\dots}[black]
\istb
\endist
\xtActionLabel(12)(12-2){s_1\in S_1}[xshift = 0pt, yshift = -10pt, fill=white]

\istroot(13)(n1-3){}
\istb
\istb[thick]
\istb[white]{\dots}[black]
\istb
\endist

\xtInfoset[dashed](12)(13){1}[a]
\xtInfoset[dashed](11-1)(12-4){2}[a]
\xtInfoset[dashed](13-1)(13-4){2}[a]

\istroot(2a)(11-1){}
\istb[white]{\dots}[black]
\endist

\istroot(2b)(11-4){}
\istb[white]{\dots}[black]
\endist

\istroot(2c)(12-1){}
\istb[white]{\dots}[black]
\endist

\istroot(2d)(12-4){}
\istb[white]{\dots}[black]
\endist

\istroot(2e)(13-1){}
\istb[white]{\dots}[black]
\endist

\istroot(2f)(13-4){}
\istb[white]{\dots}[black]
\endist

\xtdistance{20mm}{14mm}
\istroot(21)(11-3){}
\istb
\istb<level distance= 22mm>[thick]{}{\genfrac{}{}{0pt}{0}{{\mG}_1(s^1_1,s_2),}{{\mG}_2(s^1_1,s_2)}}
\istb[white]{\dots}[black]
\istb
\endist

\istroot(22)(12-2){}
\istb
\istb<level distance= 22mm>[thick]{}{\genfrac{}{}{0pt}{0}{\mG_1(s_1,s_2),}{\mG_2(s_1,s_2)}}
\istb[white]{\dots}[black]
\istb
\endist
\xtActionLabel(22)(22-2){s_2\in S_2}[yshift = -16pt, fill = white]

\istroot(23)(13-2){}
\istb
\istb[white]{\dots}[black]
\istb<level distance= 22mm>[thick]{}{\genfrac{}{}{0pt}{0}{{\mG}_1(s_1, s^2_2),}{{\mG}_2(s_1 ,s^2_2)}}
\istb
\endist
\xtActionLabel(23)(23-3){s^2_2\in S^2_2}[yshift = -14pt, xshift = 4pt, fill = white]

\end{istgame}

\end{figure}

Nature, by choosing $\theta_n$, forces player $n$ to play from the same strategy set as in her excluded game. 
Nonetheless, players have their entire strategy set available after $\theta_0$.
Note that game ${G}^{0,0}$ is equivalent to the original game $\mG$. 
We denote the component equivalent to $K$ in ${G}^{0,0}$ by $K^{0,0}$. 
Since ${G}^{0,0}$ is equivalent to $G$, component $K^{0,0}$ has also index $0$ (cf. Property \ref{I2}). 
Fix a neighborhood $\mathcal{O}^{0,0}$ of $K^{0,0}$, such that every equilibrium of ${G}^{0,0}$ in the closure of $\mathcal{O}^{0,0}$ belongs to $K^{0,0}$.
Since the index is locally constant with respect to payoff perturbations, taking $\e_1, \e_2>0$ sufficiently small implies that the neighborhood $\mathcal{O}^{0,0}$ is admissible for the game ${G}^{\e_1, \e_2}$ and that it has index~$0$.   
For notational convenience, we work with behavioral strategies. 
A behavioral strategy for player $n$ in game $G^{\e_1,\e_2}$ is a pair $(\t_n,\zeta^n_n)\in\S_n \times \S^n_n$.

\begin{lemma}\label{auxStep1}
For $\e_1, \e_2>0$ sufficiently small, every equilibrium $(\t_1, \zeta^1_1,\t_2,\zeta^2_2)\in\mathcal{O}^{0,0}$ of game $G^{\e_1,\e_2}$ satisfies $\supp(\t_n)\cap S^n_n=\varnothing$ for each $n=1,2$. In addition, any $s^n_n \in S^n_n$ is a strictly inferior reply (after $\theta_0$) to any such equilibrium. 
\end{lemma}

\begin{proof} 
The proof is divided in two cases. For the first, let $(\varepsilon^k_1)_{k}, (\varepsilon^k_2)_k$ be positive sequences converging to $0$ and $(\t^{k}_1, \zeta^{1,k}_1,\t^{k}_2, \zeta^{2,k}_2)_{k\in\mathbb{N}} \subseteq \mathcal{O}^{0,0}$ a sequence of equilibria of $G^{\varepsilon^k_1, \varepsilon^k_2}$.
We prove that for $k$ large enough $\supp(\tau^k_1)\cap S^1_1=\varnothing$ and that, after $\theta_0$, every member of $S^1_1$ is a strictly inferior reply to $\tau^k_2$. 
Passing to subsequences if necessary, we can assume that the sequence of strategy profiles converges to $(\tau_1,\zeta^1_1,\tau_2,\zeta^2_2)$.  
%following series converge, $\sigma^k \to \s$; $\tau^k \to \tau$; $\zeta^{1,k}_1 \to \zeta^1_1$. 
In addition, writing $\tau^k_1 = (1-\alpha_k)\tau^{\times,k}_1 + \alpha_k \tau^{1,k}_1$, with  $\alpha_k \geq 0$, $\supp(\tau^{\times,k}_1) \subseteq (S_1 \setminus S^1_1)$ and $\supp(\tau^{1,k}_1) \subseteq S^1_1$, we also assume $\psi^k \equiv \varepsilon^k_1[\varepsilon^k_1 + (1-\varepsilon^k_1)\alpha_k]^{-1} \to \psi\in [0,1]$. Observe that the support of $\t_2$ is a subset of $(S_2 \setminus S^2_2)$, since $(\tau_1,\zeta^1_1,\tau_2,\zeta^2_2) \in K^{0,0}$.

%We need the following claim.
%
\begin{claim}\label{claimaux} 
Suppose Supp$(\tilde \t_2) \subseteq (S_2 \setminus S^2_2)$.  If $q_2^\otimes(\tilde{\tau}_2)$ is equivalent to $q_2^\otimes(\tau_2)$ in $\mG^{\otimes}$, then for $n=1,2$ we have
$\mG_n(\tau^{\times, k}_1,\tilde\t_2)=\mG_n(\tau^{\times, k}_1,\t_2)$.
\end{claim}
\begin{proof}[Proof of Claim \ref{claimaux}] 
Note that $(\tau^k_1,\tau^k_2)\to (\tau_1,\tau_2) \in K$ so that $\supp(\t_2) \cap S^2_2 = \emptyset$. 
Since $q_2^\otimes(\tilde{\tau}_2)$ is equivalent to $q_2^\otimes(\tau_2)$ in $\mG^{\otimes}$ both $q_2^\otimes(\tilde{\tau}_2)$ and $q_2^\otimes(\tau_2)$ give the same payoff against any player 1's strategy $\t^\otimes_1$ in $\mG^{\otimes}$. 
Given $\t^{\times,k}_1 \in \D(S_1 \setminus S^1_1)$, both $q_2^\otimes(\tilde{\tau}_2)$ and $q_2^\otimes(\tau_2)$ yield the same payoff against $q^{\otimes}_1(\t^{\times,k}_1)$. 
Therefore, $\mG_n(\t^{\times,k}_1, \tilde \t_2) = \mG^{\otimes}_n(q^{\otimes}(\tilde \t_2), q^{\otimes}_1(\t^{\times,k}_1)) = \mG^{\otimes}_n(q^{\otimes}(\t_2), q^{\otimes}_1(\t^{\times,k}_1)) = \mG_n(\t^{\times,k}_1, \t_2)$.
\end{proof}

For $k$ sufficiently large, $\supp(\tau_2)\subseteq\supp(\tau^k_2)$. 
For such a $k$, the limit $\t_2$ maximizes player 2's expected payoff at her information set following $\theta_0$ and $\theta_1$. 
For $x_2 \in \S_2$, player 2's expected payoff is
\begin{equation}\label{eq:maxx2}
[\varepsilon^k_1 + (1-\varepsilon^k_1 - \varepsilon^k_2)\alpha_k]\Big[\psi^k {\mG}_{2}(\zeta^{1,k}_1, x_2) + (1-\psi^k)\mG_2(\tau^{1,k}_1, x_2) \Big] +  \\
(1-\varepsilon^k_1 - \varepsilon^k_2)(1-\alpha_k)\Big[\mathbb{G}_2(\tau^{\times,k}_1, x_2)\Big].
\end{equation}

Let $A_2\subset\S_2$ be defined as: $\tilde \t_2 \in A_2$ if and only if $q^\otimes_2(\tilde{\t}_2)$ is equivalent to $q_2^{\otimes}(\tau_2)$ in $\mG^\otimes$. Since $\tau_2$ maximizes expression~\eqref{eq:maxx2}, Claim~\ref{claimaux} implies that $\tau_2$ also solves 
\begin{equation*} 
\max_{x_2\in A_2}{\psi^k {\mG}_{2}(\zeta^{1,k}_1, x_2) + (1-\psi^k)\mG_2(\tau^{1,k}_1, x_2)}. 
%=\\ 
%\max_{x^0_2} \Big\{\left[\varepsilon^k_1 + (1-\varepsilon^k_1)\alpha_k\right]
%\Big[\psi^k {\mG}^1_{2}(\zeta^{1,k}_1, x^0_2) + %(1-\psi^k)\mG_2(\tau^{1,k}_1,\tau^{\otimes}_2, x^0_2) \Big] 
%+(1-\varepsilon^k_1)(1-\alpha_k)\Big[\mG_2(\tau^{\times,k}_1, \tau^{\otimes}_2, x^0_2) \Big] \Big{\}}.
\end{equation*}
From Lemma~\ref{lm:kuhn}, player 2 has a product strategy $\tau^\otimes_2*\tau^0_2$ that is equivalent to $\tau_2$ in $\mG$. 
Since $(\t_1, \t_2) \in K$, then $\t^{\otimes}_2$ is equivalent to $\hat \t^{\otimes}_2$ in $\mG^{\otimes}$. 
By the definition of payoff functions in player 1's excluded game, $\tau^0_2$ solves
\begin{equation*} 
\max_{x^0_2\in\S_2^0}\psi^k {\mG}^1_{2}(\zeta^{1,k}_1, x^0_2) + (1-\psi^k)\mG^1_2(\tau^{1,k}_1, x^0_2).
\end{equation*}
That is, for $k$ large enough, $\t^0_2$ is a best-reply  against $\psi^k \zeta^{1,k}_1 +(1-\psi^k)\tau^{1,k}_1$ in player 1's excluded game ${\mG}^1$.
Of course, $\t^0_2$ is also a best-reply  against the limit $\psi \zeta^{1}_1 +(1-\psi)\tau^{1}_1$. 
 
In turn, $\zeta^{1,k}_1$ is a best-reply (when player 1's strategy is constrained to $\S^1_1$) against $\t^{k}_2$ in ${\mG}$ for every $k$.
By continuity, $\zeta^{1}_1$ is a constrained best-reply against $\t_2$ in $\mG$ and, therefore, also against the equivalent product strategy $\tau^\otimes_2*\tau^0_2$, which implies that $\zeta^{1}_1$ is a best-reply against $\t^0_2$ in $\mG^1$. 
Using an analogous argument, if $\a_k >0$ for all $k$, we also conclude that $\t^1_1$ is a best-reply  against $\t^0_2$ in $\mG^1$.
If we do not have $\a_k>0$ for all $k$ then either there exists a subsequence for which it holds or $\psi=1$. 
In either case, $(\t^0_2,\psi \zeta^{1}_1 +(1-\psi)\tau^{1}_1)$ is an equilibrium of ${\mG}^1$.

From genericity property \ref{P2}, no equilibrium of ${\mG}^1$ gives the same payoff as the equilibrium outcome $Q$ to player 1. 
Therefore, $\t^0_2 \notin \tilde\partial K_2$ and, for every $s^1_1 \in S^1_1$ we have $\mG_1(s^1_1, \t_2) = \mG^1_1(s^1_1, \t^0_2)< G_1(Q)$. 
We conclude that for sufficiently large $k$, every member of $S^1_1$ is a strictly inferior reply to $\t^k_2$. So for $k$ large enough, $\supp(\t^k_1)\subseteq S_1 \setminus S^1_1$ which finishes the proof of the first case. 

For the second case, consider positive sequences $(\varepsilon^k_i)_{k}, i=1,2$ converging to $0$ and $(\sigma^k)_{k \in \mathbb{N}} \subset \mathcal{O}^{0,0}$ a sequence of equilibria of $G^{\varepsilon^{k_1}, \varepsilon^k_2}$. For sufficiently large $k$, using an argument symmetric to the first case, $\supp(\tau^k_2) \subseteq (S_2 \setminus S^2_2)$ and $S^2_2$ is a set of strictly inferior replies after $\theta_0$ to $\tau^k_1$. 
From these two cases, for sufficiently small $\varepsilon_1, \varepsilon_2>0$, Lemma \ref{auxStep1} follows. 
\end{proof}

Take $\e_1, \e_2>0$ as small as prescribed by Lemma \ref{auxStep1}. 
Every equilibrium of ${G}^{\e_1, \e_2}$ in $\mathcal{O}^{0,0}$ is such that any strategy that assigns positive probability to some member of $S^n_n$ after $\theta_0$ is a strictly inferior reply to the equilibrium. 
With abuse of notation, denote also by $G^{\e_1, \e_2}$ the game that is obtained by eliminating $S^1_1$ and $S^2_2$ from players' actions at their information set after $\theta_0$. Since they are always inferior replies, their elimination does not affect the index of $\mathcal{O}^{0,0}$.

For $\d>0$, define
\begin{equation*}
U_{n}^\d \equiv\Big \{ \t_{n} \in \D(S_{n} \setminus S^{n}_{n}) \mid \mG_{-n}(s^{-n}_{-n}, \t_{n}) < G_{-n}(Q) - \d \text{ for all } s^{-n}_{-n} \in S^{-n}_{-n}\Big\}.
\end{equation*}
Since the component $K$ induces a unique distribution $Q$, the component that induces $Q$ in $\mG^{\otimes}$ is a rectangle $R_1 \times R_2$ with $R_n \subseteq \S^{\otimes}_n$.%
\footnote{~The game $\mG^{\otimes}$ is the normal form of the extensive-form game obtained from $G$ by eliminating nodes and branches that have zero probability under $Q$. 
Therefore, the outcome $Q$ is completely mixed in this extensive-form and is induced by a unique component of mixed strategies in $\mG^{\otimes}$. 
Hence, as in the proof of Lemma \ref{nonzeroindex}, $Q$ is induced by a unique enabling strategy profile in the enabling-form of $\mG^{\otimes}$.
(See Appending~\ref{app:nonzeroindex} for the definition of enabling strategies and enabling-form.) 
This implies that the component in $\mG^{\otimes}$ inducing $Q$ in $\mG^{\otimes}$ is a rectangle.}
From this we obtain a $\d$-neighborhood $B^{\d}_{n}\subseteq \S^\otimes_n$ of $R_n$ such that  $B^{\d} \equiv B^{\d}_1 \times B^{\d}_2$ is a neighborhood of $R_1 \times R_2$ and does not contain any other equilibrium of $\mG^{\otimes}$.
Let $V^{\d}_1\equiv(q^{\otimes}_1)^{-1}(B^{\d}_1)$, $V^{\d}_2\equiv(q^{\otimes}_2)^{-1}(B^{\d}_2)$, and $V^{\d}\equiv V^{\d}_1 \times V^{\d}_2$. Fix in addition $\d>0$ sufficiently small such that for both $n=1,2$ all equilibria of player $n$'s excluded game in $K^n$ induce a payoff smaller than $G_n(Q) - \d$ (such $\d>0$ exists due to Assumption \ref{P2}).

\begin{claim}\label{admiss} 
For $\e_1, \e_2>0$ sufficiently small,

\begin{enumerate} 

\item  $\mathcal{O}^{\d}\equiv([V^{\d}_1 \cap U^{\d}_{1}] \times \S^1_1)\times ([V^{\d}_2 \cap U^\d_2] \times \S^2_2) $ is admissible for $G^{\e_1, \e_2}$;

\item The symmetric difference $\mathcal{O}^{\d} \Delta  {\mathcal{O}}^{0,0}$ contains no equilibrium of $G^{\e_1,\e_2}$.

\end{enumerate} 
\end{claim}

\begin{proof}
We start by proving admissibility of $\mathcal{O}^{\d}$ for $G^{\e_1, \e_2}$. Let $(\e^k_n)_{k \in \mathbb{N}}$ converge to $0$ for $n=1,2$ such that $(\zeta^k, \t^k)$ is an equilibrium of $G^{\e^k_1, \e^k_2}$ in $\cl(\mathcal{O}^{\d})$, where the closure is taken with respect to $(\times_{n=1,2}\D(S_n \setminus S^n_n) \times \S^1_1 \times \S^2_2)$.
Consider a limit $(\zeta, \t)$. 
By \ref{P1}, $q^{\otimes}(\t)$ is an equilibrium of $\mG^{\otimes}$ that induces $Q$. 
Therefore, for $k$ large enough, $q^{\otimes}(\t^k) \in B^{\d}$, which implies $\t^k \in V^{\d}$. 
%Hence, for some $n$, we must have $\mG_{-n}(\zeta^{-n}_{-n}, \t_n) = G_{-n}(Q) -\d$. 
Recall that $\t_n \in \D(S_n \setminus S^{n}_{n})$ and take a product strategy $\t^{\otimes}_n*\t^{0}_n$ equivalent to $\t_n$ in $\mG$. 
We claim that $(\zeta^{-n}_{-n}, \t^0_n)$ is an equilibrium of the excluded game $\mG^{-n}$. 
It is clear that $\zeta^{-n}_{-n}$ is a best-reply of player $-n$ against $\t^0_n$ in $\mG^{-n}$.
It remains to show that $\t^0_n$ is a best-reply of player $n$ against $\zeta^{-n}_{-n}$ in $\mG^{-n}$.  Since $\t^k_n$ is an equilibrium strategy in $G^{\e^k_1, \e^k_2}$, it solves
\begin{equation}\label{expected}
\max_{x_n \in \D(S_n \setminus S^n_n)}\e^k_{-n}\mG_n(x_n, \zeta^{-n,k}_{-n})+(1-\e^k_1-\e^k_2)\mG_n(x_n, \t^{k}_{-n}).
\end{equation}

For $k$ sufficiently large, $\supp(\t_n)\subseteq\supp(\t^k_n)$, so $\t_n$ is also a solution to~\eqref{expected}. 
Note that for every $x_n \in \D(S_n \setminus S^n_n)$ we have $\mG_n(x_n, \t^{k}_{-n}) = \mG^{\otimes}_n(q^{\otimes}_n(x_n), q^{\otimes}_{-n}(\t^k_{-n}))$. Therefore, since $\t^{\otimes}_n$ is equivalent to $\hat \t^{\otimes}_n$ in $\mG^{\otimes}$, we obtain $\mG_n(\t_n, \zeta^{-n,k}_{-n}) = \mG^{-n}_n(\t^0_n, \zeta^{-n,k}_{-n})$. 
Taking limits in $k$, we obtain that $\t^0_n$ is a best-reply against $\zeta^{-n}_{-n}$ in $\mG^{-n}$. 
That is, $(\t^0_n ,\zeta^{-n}_{-n})$ is an equilibrium of $\mG^{-n}$ as claimed. This implies $\mG^{-n}_{-n}(\t^0_n, \zeta^{-n}_{-n}) < \mG_{-n}(Q) - \d$, by Assumption \ref{P2} and our choice of $\d$. 
Hence, $\mG_{-n}(\t_n, \zeta^{-n}_{-n}) = \mG^{-n}_{-n}(\t^0_n, \zeta^{-n}_{-n}) < \mG_{-n}(Q) - \d$. 
That is, $\t^k_n \in U^{\d}_n$ for sufficiently large $k$. 
This proves that for $\e_1, \e_2 >0$ small enough $\mathcal{O}^{\d}$ is admissible for $G^{\e_1, \e_2}$.

We now prove that for sufficiently small $\e_1, \e_2>0$, $\mathcal{O}^{\d}$ contains the equilibria of $G^{\e_1, \e_2}$ that are contained in $\mathcal{O}^{0,0}$. It is sufficient to prove that for any convergent sequence $(\zeta^k, \tau^k)_{k \in \mathbb{N}} \subset \mathcal{O}^{0,0}$ of equilibria $G^{\e^k_1, \e^k_2}$ with limit $(\zeta, \tau)$ and $\e^k_n$ converging to $0$ for $n=1,2$, the sequence is eventually in $\mathcal{O}^{\d}$. First observe that the limit must satisfy $\mG_{-n}(\zeta^{-n}_{-n}, \t_n) \leq G_{-n}(Q)$. Take a product strategy $\t^{\otimes}_n * \t^0_n$ equivalent to $\t_n$ in $\mG$. We claim that $(\zeta^{-n}_{-n}, \t^0_n)$ is an equilibrium of the excluded game $\mG^{-n}$. It is clear that $\zeta^{-n}_{-n}$ is a best-reply of player $-n$ against $\t^0_n$ in $\mG^{-n}$. It remains to show that $\t^0_n$ is a best-reply of player $n$ against $\zeta^{-n}_{-n}$ in $\mG^{-n}$.  Since $\t^k_n$ is an equilibrium strategy in $G^{\e^k_1, \e^k_2}$, it solves \eqref{expected}. Continuing the argument in the same fashion as in the first part, we obtain that $(\t^0_n, \zeta^{-n}_{-n})$ is an equilibrium of $\mG^{-n}$. This implies $\mG^{-n}_{-n}(\t^0_n, \zeta^{-n}_{-n}) < \mG_{-n}(Q) - \d$, by Assumption \ref{P2} and our choice of $\d$. Hence, $\mG_{-n}(\t_n, \zeta^{-n}_{-n}) = \mG^{-n}_{-n}(\t^0_n, \zeta^{-n}_{-n}) < \mG_{-n}(Q) - \d$. Thus  $(\zeta^k, \tau^k)$ belongs to $\mathcal{O}^{\d}$ for sufficiently large $k$. 

The proof of the result that for sufficiently small $\e_1, \e_2>0$, $\mathcal{O}^{0,0}$ contains the equilibria of $G^{\e_1, \e_2}$ that are contained in $\mathcal{O}^{\d}$ can be obtained from noticing that any limit $(\zeta, \t)$ of a sequence of equilibria in $\mathcal{O}^\d$ satisfies $\t \in V^\d$ and $\t \in U^\d$. \end{proof}

We now construct an expression for the index of $\mathcal{O}^{\d}$ that allows us to conclude that for some $n=1,2$, the supporting polytope $K^n$ has zero index with respect to $\mG^n$. 
When player $n$ best-replies against $(\t_{-n}, \zeta^{-n}_{-n})$ in game ${G}^{\e_1, \e_2}$ at her information set after $\theta_0$, player $n$ chooses $\t_{n} \in \D(S_n \setminus S^n_n)$ that maximizes 
\begin{equation}\label{expected2}
\e_{-n}\mG_n(\t_n, \zeta^{-n}_{-n})+(1-\e_1-\e_2)\mG_n(\t_n, \t_{-n}).
\end{equation}
Since actions in $S^n_n, n=1,2$ are no longer available after $\theta_0$, the function $\mG_n(\t_n, \t_{-n})$ only depends on $q_n^\otimes(\t_n)$ and $q_{-n}^\otimes(\t_{-n})$.
In order to compute the best-reply $\t_n$, player $n$ only requires $\z^{-n}_{-n}$ and $\t^{\otimes}_{-n}$.
The correspondence that assigns to each $(\zeta_{-n}^{-n},\tau_{-n}^{\otimes})$ the optimal $\t_n \in\D(S_n \setminus S^n_n)$ is 
\begin{equation*}
\mathrm{\BR}^{\otimes, \e_1,\e_2}_n:\S^{-n}_{-n}\times\S^{\otimes}_{-n}\rightrightarrows \D(S_n \setminus S^n_n).
\end{equation*}
In addition, define the correspondence $\BR^{\times}_n:\D(S_{-n}\setminus S^{-n}_{-n})\rightrightarrows \S^n_n$  by:
\begin{equation*} 
\BR^{\times}_n(\t_{-n})\equiv 
\arg\max_{\zeta^n_n \in \D(S^n_n)}\mG_n(\tau_{-n},\zeta^{n}_{n}).
\end{equation*}
Define $\BR^{\times}\equiv\BR^{\times}_1\times\BR^{\times}_2$.
Let $\mathrm{id}_n$ be the identity in $\S_n^n$.
Player $n$'s best-reply correspondence $\BR_n^{\e_1,\e_2}$ in game $G^{\e_1,\e_2}$ is
\begin{equation*}
\BR^{\times}_n
\times
\left[
\BR_n^{\otimes,\e_1,\e_2}
\circ(\mathrm{id}_{-n}\times q_{-n}^{\otimes})
\right]
:\S^{-n}_{-n}\times\Delta(S_{-n}\setminus S^{-n}_{-n})
\rightrightarrows
\S^{n}_{n}\times\Delta(S_{n}\setminus S^{n}_{n}).
\end{equation*}
And the best-reply correspondence in $G^{\e_1,\e_2}$ is $\BR^{\e_1,\e_2}\equiv\BR_1^{\e_1,\e_2}\times\BR_2^{\e_1,\e_2}$.

Instead of using $\BR^{\e_1,\e_2}$ directly to compute the index of $\mathcal{O}^\d$ we use a selection of $\BR^{\e_1,\e_2}$.
That is, a new correspondence whose graph is a subset of the graph of $\BR^{\e_1,\e_2}$.
To define it, consider first the correspondence $\phi^{\e_1,\e_2}_n:\S^{-n}_{-n}\times\S^{\otimes}_{-n}\rightrightarrows\S^\otimes_n\times\S^0_n$ that assigns to each $(\zeta^{-n}_{-n},\t^{\otimes}_{-n})\in\S^{-n}_{-n}\times\S^{\otimes}_{-n}$ the set of pairs $(\t^\otimes_n,\t^0_n)\in\S^\otimes_n\times\S^0_n$ that satisfy
\begin{align}
\t^{\otimes}_n &\in q^{\otimes}_n \circ \BR_n^{\otimes,\e_1,\e_2}(\zeta^{-n}_{-n},\t^{\otimes}_{-n}) \text{, and}\\
\t^0_n &\in\arg\max_{\tilde{\t}^0_n\in\S^0_n} \mathbb{G}_n(\hat{\t}^{\otimes}_n * \tilde{\t}^0_n, \zeta^{-n}_{-n}).
\end{align}
Let $e_{n}:\S_n^{\otimes}\times\S_n^0\to\D(S_n\setminus S^n_n)$ be defined by $e_n(\t^\otimes_n,\t^0_n)\equiv\t^\otimes_n*\t^0_n$.
Define also $w_{n}:\S^n_n\times\Delta(S_n\setminus S^n_n)\to\S^n_n\times\S^{\otimes}_n\times\S^0_n$ by $w_n(\zeta^n_n, \t_n) \equiv (\zeta^n_n, q^{\otimes}_n(\t_n), q^0_n(\t_n))$. 
As usual, $e \equiv e_1 \times e_2$, $w \equiv w_1 \times w_2$, and $\phi^{\e_1,\e_2} \equiv \phi^{\e_1,\e_2}_1 \times \phi^{\e_1, \e_2}_2$.
Consider the correspondence 
\begin{equation*}
\Phi^{\e_1, \e_2} \equiv (e \circ \phi^{\e_1, \e_2} \circ w) : \S^1_1 \times \S^2_2\times \Delta(S_1 \setminus S^1_1)\times\Delta(S_2 \setminus S^2_2) \rightrightarrows \Delta(S_1 \setminus S^1_1)\times \Delta(S_2 \setminus S^2_2).
\end{equation*}

\begin{lemma}
$\Phi^{\e_1, \e_2}$ is a nonempty, contractible, and compact valued, upper hemicontinuous correspondence. 
\end{lemma}

\begin{proof} 
The fact that $\Phi^{\e_1, \e_2}$ is nonempty, compact-valued, and upper hemicontinuous follows from standard arguments.
We show that values are contractible.  
Fix $(\t,\zeta)$ and some $\check{\t} \in \Phi^{\e_1, \e_2}(\t,\zeta)$. Then $\check{\t} = \check \t^{\otimes} * \check{\t}^0$.
Consider the homotopy $J = J_1 \times J_2: \Phi^{\e_1, \e_2}(\t,\zeta) \times [0,1] \to \Phi^{\e_1, \e_2}(\t,\zeta)$ that for each $\tilde\t=(\tilde{\t}_1,\tilde{\t}_2)=(\tilde \t^{\otimes}_1 * \tilde \t^0_1,\tilde \t^{\otimes}_2 * \tilde \t^0_2) \in \Phi^{\e_1, \e_2}(\t,\zeta)$ is defined by $J_n(t, \tilde \t_n) = \left[(1-t) \tilde \t^{\otimes}_n + t \check{\t}^{\otimes}_n\right] * \left[(1-t) \tilde \t^{0}_n + t \check{\t}^{0}_n\right]$ for each $n=1,2$. 
The homotopy is well-defined. 
Moreover,  $J_n(0, \tilde \t^{\otimes}_n * \tilde \t^0_n) = \tilde \t^{\otimes}_n * \tilde \t^0_n$ and $J_n(1, \tilde \t^{\otimes}_n * \tilde \t^0_n) = \check{\t}^{\otimes}_n *  \check{\t}^0_n$, which shows that $\Phi^{\e_1, \e_2}(\t,\zeta)$ is contractible valued. 
\end{proof}

\begin{lemma} 
$\BR^{\times} \times \Phi^{\e_1, \e_2}$ is a selection of $\BR^{\e_1, \e_2}$.
\end{lemma}

\begin{proof} 
Let $(\check{\zeta}, \check{\t}) \in [\BR^{\times} \times \Phi^{\e_1, \e_2}](\zeta, \tau)$. 
%We need to show $(\check{\zeta}, \check{\t}) \in \BR^{\e_1, \e_2}(\zeta, \tau)$.
Since $\check{\zeta}\in\BR^\times(\zeta,\tau)$ we only need to show $\check{\tau}_n\in\BR_n^{\otimes,\e_1,\e_2}(\zeta^{-n}_{-n},q_{-n}^{\otimes}(\tau))$  for $n=1,2$.
By definition, $\check{\t}_n = \check{\t}^{\otimes}_n * \check{\t}^0_n$, where
\begin{equation*}
\check{\t}^{\otimes}_n \in q^{\otimes}_n\left[\BR_n^{\otimes,\e_1,\e_2}(\zeta^{-n}_{-n},q_{-n}^{\otimes}(\t_{-n}))\right].
\end{equation*}

Therefore, there exists
\begin{equation*}
\t'_n \in\arg\max_{\tilde{\t}_n}\Big\{ \e_{-n} \mG_{n}(\tilde \t_n, \zeta^{-n}_{-n}) + (1-\e_1 -\e_2)\mG_{n}(q^{\otimes}_n(\tilde \t_n), q^{\otimes}_{-n}(\t_{-n}))\Big\}
\end{equation*}
such that  $q^{\otimes}_n(\t'_n) = \check{\t}^{\otimes}_n$. 
Take $\t^*_n = \t^{*,\otimes}_n * \t^{*,0}_n$ equivalent to $\t'_n$ in $\mG$. 
Then $\check{\t}^{\otimes}_n$ must be equivalent to $\t^{*, \otimes}_n$ in $\mG^{\otimes}$. 
Recall $\supp(\t^{*, \otimes}_n) \subseteq \supp(\hat{\t}^{\otimes}_n)$ because $\hat{\t}^{\otimes}_n$ is completely mixed. 
Hence,
\begin{equation*}
\check{\t}^0_n \in \arg\max_{\t^0_n}\sum_{s^{\otimes}_n, s^0_n} \hat{\t}^{\otimes}_n(s^{\otimes}_n)\t^0_n(s^0_n) \mG_n(s^{\otimes}_n, s^0_n, \zeta^{-n}_{-n}) \subseteq \arg\max_{\t^0_n}\sum_{s^{\otimes}_n, s^0_n} \t^{*,\otimes}_n(s^{\otimes}_n)\t^0_n(s^0_n) \mG_n(s^{\otimes}_n, s^0_n, \zeta^{-n}_{-n}).
\end{equation*}
That is,
\begin{equation*}
%\check{\t}^{\otimes}_n * \check{\t}^0_n \in\BR_n^{\otimes,\e_1,\e_2}(\zeta_{-n},q_{-n}^{\otimes}(\t_{-n}))
\check{\t}^{\otimes}_n * \check{\t}^0_n \in \arg\max_{\tilde{\t}_n}\Big\{ \e_{-n} \mG_{n}(\tilde \t_n, \zeta^{-n}_{-n}) + (1-\e_1 -\e_2)\mG_{n}(q^{\otimes}_n(\tilde \t_n), q^{\otimes}_{-n}(\t_{-n}))\Big\},
\end{equation*}
which shows $\check{\t}_n\in\BR_n^{\otimes,\e_1,\e_2}(\zeta^{-n}_{-n},q_{-n}^{\otimes}(\tau_{-n}))$  for $n=1,2$ and, therefore, $(\check{\zeta}, \check{\t}) \in \text{BR}^{\e_1, \e_2}(\zeta, \tau)$.
\end{proof}

Letting $\BR^{\otimes}$ be the best-reply correspondence in the game $\mG^{\otimes}$, let us define $\Phi \equiv \Phi^{0,0} = e \circ (\BR^{1}_2 \times \BR^{\otimes} \times \BR^2_1) \circ w$. 
Similarly to $\Phi^{\e_1,\e_2}$, standard arguments show that $\Phi$ is upper hemicontinuous as well as nonempty, compact, and contractible valued.

\begin{lemma} 
The following statements hold:
\begin{enumerate}
\item The neighborhood $\mathcal{O}^{\d}$ is admissible for $\BR^{\times} \times \Phi$. 
\item For $\e_1$, $\e_2>0$ sufficiently small, $\mathcal{O}^{\d}$ is admissible for $\BR^{\times} \times \Phi^{\e_1, \e_2}$.
\item For $\e_1$, $\e_2>0$ sufficiently small, $\BR^{\times} \times \Phi^{\e_1, \e_2}$ assigns the same index to $\mathcal{O}^{\d}$ as $\BR^{\times} \times \Phi$ does.
\end{enumerate} 
\end{lemma}

\begin{proof} 
To prove (1), let $(\zeta, \t)$ be a fixed point of $\BR^{\times} \times \Phi$, where $(\zeta, \tau) \in \cl(\mathcal{O}^{\d})$. 
Then $\t_n = \t^{\otimes}_n * \t^0_n$ and $\t^{\otimes}_n \in B^{\d}$. 
This implies that $\t^{\otimes}_n$ is equivalent to $\hat{\t}^{\otimes}_n$ in $\mG^{\otimes}$. 
Therefore, $\zeta^{-n}_{-n} \in\BR^{\times}_{-n}(\t_n) =\BR^{-n}_{-n}(\t^0_n)$.
%$\mG_{-n}(\tilde{\zeta}^{-n}_{-n}, \t_n) = \mG^{-n}_{-n}(\tilde{\zeta}^{-n}_{-n}, \t^0_n)$, which implies that $\zeta^{-n}_{-n} \in\BR^{-n}_{-n}(\t^0_n)$. 
Moreover, by construction, $\t^0_n \in\BR^{-n}_n(\zeta^{-n}_{-n})$.  
Hence, $(\t^0_n, \zeta^{-n}_{-n})$ is an equilibrium of $\mG^{-n}$. 
Since $\mG_{-n}(\zeta^{-n}_{-n}, \t_n) \leq G_{-n}(Q) - \d$ we obtain $\mG^{-n}_{-n}(\zeta^{-n}_{-n}, \t^0_n) \leq \mG_{-n}(Q) - \d$.
%\marginpar{Is the sentence starting with ``Since'' necessary?}
By \ref{P2} and our choice of $\d$, we have $\mG^{-n}_{-n}(\zeta^{-n}_{-n}, \t^0_n) < G_{-n}(Q) - \d$ and finally $\mG_{-n}(s^{-n}_{-n}, \t_n) < G_{-n}(Q) - \d$ for every $s^{-n}_{-n} \in S^{-n}_{-n}$. 
That is, $(\zeta, \t) \in \mathcal{O}^{\d}$, which proves admissibility of $\mathcal{O}^{\delta}$. 

To prove (2), fix a neighborhood $W$ of the graph of $\BR^{\times} \times \Phi$ such that $W \cap \{ (\t,\zeta,\t,\zeta) \mid (\t,\zeta) \in \cl(\mathcal{O}^{\d}) \setminus \mathcal{O}^{\d} \} = \emptyset$. 
For $\eta>0$, define the set of $\eta$-replies against $\t^{\otimes}\in\S^{\otimes}$ by $\BR^{\otimes, \eta}(\t^{\otimes}) \equiv \big\{ \tilde{\t}^{\otimes}\in\S^{\otimes} \mid\mG^{\otimes}_n(\tilde{\t}^{\otimes}_n, \t^{\otimes}_{-n}) \geq \text{max}_{s^{\otimes}_n} \mG^{\otimes}_{n}(s^{\otimes}_n, \t^{\otimes}_{-n}) - \eta \text{ for }n=1,2 \big\}$. 
Fix $\eta>0$ sufficiently small such that $\mathrm{Graph}\big(\BR^{\times} \times [e \circ (\BR^1_2 \times \BR^{\otimes, \eta} \times \BR^{2}_1) \circ w]\big) \subseteq W$. 
Take now $\e_1, \e_2 >0$ sufficiently small such that for each $n=1,2$, $s_n = (s^{\otimes}_n, s^0_n)$ and $s^{-n}_{-n}$,
\begin{equation*}
\e_{-n}\mG_{n}(s_n, s^{-n}_{-n}) + (1-\e_1-\e_2) \mG^{\otimes}_{n}(s^{\otimes}) \geq \mG^{\otimes}_n(s^{\otimes}) - \eta.
\end{equation*}
Then, $(\tilde \zeta, \tilde \t) \in [\BR^{\times} \times \Phi^{\e_1, \e_2}](\zeta, \tau)$ implies $(\zeta, \t, \tilde \zeta, \tilde \t) \in \mathrm{Graph}\big(\BR^{\times} \times e \circ [\BR^{1}_2 \times \BR^{\otimes, \eta} \times \BR^2_1] \circ w\big)$, which implies $\mathrm{Graph}(\BR^{\times} \times \Phi^{\e_1, \e_2}) \subseteq W$. 
This shows that, for sufficiently small $\e_1, \e_2>0$, the neighbourhood $\mathcal{O}^{\d}$ is admissible for $\BR^{\times} \times \Phi^{\e_1, \e_2}$.

To prove (3), consider the neighborhood $W$ defined in the previous paragraph. 
There exists then a neighborhood $V \subseteq W$ of $\mathrm{Graph}(BR^{\times} \times \Phi)$ such that any two continuous maps whose graphs are contained in $V$ are homotopic by a homotopy contained in $W$. 
By our choice of $W$, any two such maps assign the same index to $\mathcal{O}^{\d}$.
Using a reasoning analogous to the one of the previous item, take then $\e_1$ and $\e_2$ sufficiently small so that $\mathrm{Graph}(\BR^{\times} \times \Phi^{\e_1, \e_2}) \subseteq V$. 
Approximate $\BR^{\times}\times \Phi^{\e_1, \e_2}$ by a continuous function $f$ with $\mathrm{Graph}(f) \subset V$ such that the index it assigns to $\mathcal{O}^{\d}$ determines the index of $\mathcal{O}^{\d}$ with respect to $\BR^{\times} \times \Phi^{\e_1, \e_2}$. This index is then equal to the index of $\mathcal{O}^{\d}$ with respect to $\BR^{\times} \times \Phi$. This concludes the proof.\end{proof}

Let $h_n(\t_n, t) = (1-t)\t_n + t[\hat{\t}^{\otimes}_n * q^{0}_n(\t_n)]$ and define 
\begin{align*}
H: \left(\prod_{n=1,2} \left[\Delta(S_n \setminus S^n_n) \times \S^{n}_n \right]\right) \times [0,1] &\rightrightarrows \prod_{n=1,2} \left[\Delta(S_n \setminus S^n_n) \times \S^{n}_n \right]\text{ by}\\
H(\t,\zeta, t) &= \BR^{\times}(h(\t,t)) \times \Phi (\t, \zeta)\text{ for every }(\t,\zeta, t). 
\end{align*}
The map $H$ is a homotopy of correspondences.
It is a nonempty, contractible, compact valued, upper hemicontinuous correspondence. Letting $\bar{e} (\zeta, \t^{\otimes}, \t^0) = (\zeta, e(\t^{\otimes}, \t^{0}))$, observe that for $t=1$, the homotopy equals $\bar{e} \circ (\BR^{1} \times \BR^{\otimes} \times \BR^2) \circ w$ and for $t=0$ it equals $\BR^{\times} \times \Phi$.

\begin{lemma} 
$\mathcal{O}^{\d}$ is admissible for $H$ and therefore $\bar{e} \circ (\BR^{1} \times \BR^{\otimes} \times \BR^2) \circ w$  assigns the same index to $\mathcal{O}^{\d}$ as $\BR^{\times} \times \Phi$. 
\end{lemma}
\begin{proof} 
Let $(\tau,\zeta) \in \cl(\mathcal{O}^{\d})$ satisfy $(\tau,\zeta)\in\BR^{\times}(h(\t, t)) \times \Phi(\t,\zeta)$ for some $t \in [0,1]$. 
Then $\t_n = \t^{\otimes}_n * \t^0_n$. 
First note that $\t \in V^{\d}$, since $\t_n \in \Phi_n(\t, \zeta)$ and the strategy $\t^{\otimes}_n$ is equivalent to $\hat{\t}^{\otimes}_n$. 
It follows $\zeta^n_n\in \BR^{\times}_n(h_{-n}(\t_{-n},t)) = \BR^{\times}_n(\hat{\t}^{\otimes}_{-n} * \t^0_{-n}) = \BR^{n}_{-n}(\t^0_{-n})$. 
On the other hand, $\t^0_{-n} \in \BR^{n}_{-n}(\zeta^n_n)$ so that $(\zeta^n_n, \t^0_{-n})$ is an equilibrium of $\mG^{n}$. 
By \ref{P2} and our choice of $\d$ we have $\mG^{n}_n(\zeta^n_n, \t^0_{-n}) < G_n(Q) - \d$. 
Therefore, $\mG_n(s^n_n, \t_{-n}) =\mG^n_n(s^n_n, \t^0_{-n}) < G_{n}(Q) - \d$ for every $s^n_n$. 
%Hence, for each $s^n_n, \mG_n(s^n_n, \t_{-n}) = \mG^{n}_n(\zeta^n_n, \t^0_{-n}) < G_{n}(Q) - \d$. 
This proves $(\zeta, \tau) \in \mathcal{O}^{\d}$.
\end{proof}

To finish Step 1, let
\begin{equation*}
\hat{U}^{\d}_{-n} \equiv \Big\{ \t^0_{-n} \in \S^0_{-n} \mid\mG^n_n(s^n_n, \t^0_{-n}) < G_n(Q) - \delta\text{ for all } s^n_n\in S^n_n  \Big\} \text{ and }  
\hat{\mathcal{O}}^{\d} \equiv w^{-1}\Big[\S^1_1 \times \S^2_2 \times B^{\d} \times \hat{U}^{\d}\Big].
\end{equation*}

\begin{lemma}
$\hat{\mathcal{O}}^{\d}$ is admissible for $\bar{e} \circ (\BR^{1} \times \BR^{\otimes} \times \BR^2) \circ w$ and contains the same fixed points of $\bar{e} \circ (\BR^{1} \times \BR^{\otimes} \times \BR^2) \circ w$ as $\mathcal{O}^{\d}$.  
\end{lemma}

\begin{proof} 
Let $(\tau, \zeta) \in \cl(\hat{\mathcal{O}}^{\d})$ be a fixed point of $\bar{e} \circ (\BR^1 \times \BR^{\otimes} \times \BR^2) \circ w$. 
Then $\t_n = \t^{\otimes}_n * \t^0_n$ and $\t^{\otimes}_n \in B^{\d}_n$, which implies that $\t^{\otimes}_n$ is equivalent to $\hat{\t}^{\otimes}_n$ in $\mG^{\otimes}$.
The profile $(\zeta^n_n, \t^0_{-n})$ is an equilibrium of $\mG^n$ so that $\mG^{n}_n(\zeta^n_n, \t^0_{-n}) < G_n(Q) - \d$ because of \ref{P2} and our choice of $\d$. Therefore, for every $s^n_n \in S^n_n$ we have $\mG^n_n(s^n_n, \t^0_n) < G_n(Q) - \d$. 
Hence $(\t,\zeta) \in \hat{\mathcal{O}}^{\d}$ which proves that $\hat{\mathcal{O}}^{\d}$ is admissible. 

If $(\t,\zeta) \in \hat{\mathcal{O}}^{\d}$ is a fixed point of $\bar{e} \circ (\BR^1 \times \BR^{\otimes} \times \BR^2) \circ w$ then $\t_n = \t^{\otimes}_n * \t^0_n$, so that $\t^{\otimes}_n \in B^{\d}_n$ and $\t^{\otimes}_n$ is equivalent to $\hat{\t}^{\otimes}_n$. 
Moreover, $(\zeta^n_n, \t^0_{-n})$ is an equilibrium of $\mG^n$ and $\mG_n(s^n_n, \t_{-n}) = \mG^n_n(s^n_n, \t^0_{-n}) < G_n(Q) - \d$, for every $s^n_n$. 
Therefore, $(\zeta, \t) \in \mathcal{O}^{\d}$. 
On the other hand, if $(\zeta,\t) \in \mathcal{O}^{\d}$ is a fixed point then again $\t_n = \t^{\otimes}_n * \t^0_n $ and $\t^{\otimes}_n \in B^{\d}_n$ so that $\t^{\otimes}_n$ is equivalent to $\hat{\t}^{\otimes}_n$ in $\mG^{\otimes}$. 
Therefore, $\mG^n_n(s^n_n, \t^0_n) = \mG_n(s^n_n, \t^0_n) < G_n(Q) - \d$ for every $s^n_n$. 
That is, $\t^0_n \in \hat{U}^{\d}_n$ and $(\zeta,\t) \in \hat{\mathcal{O}}^{\d}$.
\end{proof} 

Note now that Property \ref{I6} implies that $\bar{e} \circ (\BR^{1} \times\BR^{\otimes} \times \BR^2) \circ w\vert_{\hat{\mathcal{O}}^{\d}}$ has index zero. 
In turn, Property \ref{I5} implies that the index of $\BR^{1}\vert_{\S^1_1 \times \hat{U}^{\d}_2} \times \BR^{2}\vert_{\S^2_2 \times \hat{U}^{\d}_1} \times \BR^{\otimes}\vert_{B^{\d}}$ is zero. 
Given Proposition~\ref{nonzeroindex}, Property \ref{I4} implies that for some $n=1,2$, the index of $\BR^n\vert_{\S^n_n \times \hat{U}^{\d}_{-n}}$ is zero. 
Since $\BR^n$ is the best-reply correspondence of the excluded game $\mG^n$, the index of the supporting polytope in player $n$'s excluded game is zero.

\begin{example}\label{ex:B}
We continue with Example~\ref{ex:A}.
The ``\textit{Out}''-component induces supporting polytope $K^1=\D(\{L,R\})\times\big\{\a\ell+(1-\a)r\mid 0\leq \a\leq2/3\big\}$ in player 1's excluded game. 
Note that such a supporting polytope contains two equilibria, $(R,r)$ and the mixed $(\frac{3}{4}L+\frac{1}{4}R,\frac{1}{4}\ell+\frac{3}{4}r)$.
As anticipated in the Introduction, the former has index $+1$ while the latter has index $-1$.
Therefore, the supporting polytope has index 0.
\end{example}

\subsection{Step 2: Perturbing the excluded game so that it has no equilibria in the supporting polytope}\label{sec:perturbation}

From Step 1, assume without loss of generality that $K^1$ has index $0$.
Lemma~\ref{lm:excluded} below guarantees that, when that is the case, there is a game $\bar\mG^1$ equivalent to $\mG^1$ and a payoff perturbation of $\bar{\mG}^1$ such that, letting $\bar K^1$ be the set of strategies equivalent to $K^1$ in $\bar \mG^1$, the resulting perturbed game does not have an equilibrium in $\bar{K}^1$.
Such a perturbed game also satisfies some additional properties that are used in Section~\ref{sec:step3}.

To state the Lemma, define $\S^1_2\equiv\Delta(S^0_2)$ so that we can write the set of mixed strategy profiles in $\mG^1$ as $\S^1=\S^1_1\times\S^1_2$.
Define also $K^{1,\eta}_2 \equiv \big\{ \s^1_2 \in \S^1_2 \mid \forall s^1_1 \in S^1_1, \mG^1_1(s^1_1, \s^1_2) \leq G_1(Q) - \eta \big\}$. 

\begin{lemma}\label{lm:excluded}
There exists a game $\bar\mG^1 = (\bar S^1_1, \bar S^0_2, \bar{\mG}^1_1, \bar{\mG}^1_2)$ equivalent to $\mG^1$ such that:
\begin{enumerate}
\item \label{1} Player $1$'s pure strategy set $\bar{S}^1_1$ is the collection of vertices of a polyhedral refinement $\mathcal{P}_1$ of a simplicial subdivision of $\bar{\S}^1_1$.
\item \label{2} Player 2's pure strategy set $\bar{S}^0_2$ is the collection of vertices of a polyhedral refinement $\mathcal{P}_2$ of $\S^1_2$ for which no polyhedron of $\mathcal{S}_2$ that intersects $\tilde{\partial}K_2^1$ has a point in common with $K_2^{1,\eta}$.
\item\label{part:payoffs} 
There exists $\delta>0$, a $\d$-perturbation $\bar{\mG}^{1,\d}$ of $\bar{\mG}^1$ and, for each $\e>0$, an $\e$-perturbation  $\bar{\mG}^{1,\d,\e}$ of $\bar{\mG}^{1,\d}$ such that for $n=1,2$:
\begin{align*}
\lim_{\e\to 0}\bar\mG^{1, \d,\e}_n(s_1,s_2)&=\bar\mG^{1,\d}_n(s_1,s_2)=\bar\mG^1_n(s_1,s_2)\text{ if }s_2\not\in K_2^{1,\eta},\text{ and}\\
\lim_{\e\to 0}\bar\mG^{1, \d, \e}_n(s_1,s_2)&=\bar\mG^{1,\d}_n(s_1,s_2)\geq\bar\mG^{1}_n(s_1,s_2)\text{ if }s_2\in K_2^{1,2\eta}.
\end{align*}
\item\label{part:support} For each $n=1,2$ and $\e>0$ small enough, if mixed strategies $\s_n\neq \tilde{\s}_n$ are equivalent in $\bar\mG^{1,\d}$ and $\s_n$'s support is within the vertices of a polyhedron of $\mathcal{P}_n$ whereas $\tilde \s_n$'s is not, then, for every mixed strategy $\s_{-n}$, we have $\bar\mG^{1,\d,\e}_n(\s_n,\s_{-n})>\bar\mG^{1,\d,\e}_n(\tilde{\s}_n,\s_{-n})$.
\item\label{part:noeq} For $\e>0$ small enough, the game $\bar\mG^{1, \d,\e}=(\bar{S}^1_1,\bar{S}^0_2,\bar\mG_1^{1,\d,\e},\bar\mG_2^{1,\d,\e})$ does not have an equilibrium in $\bar{K}^1$.
\end{enumerate}
\end{lemma}

\begin{proof}
See Appendix~\ref{sec:proofLemmaExcluded}. 
\end{proof}

The construction of $\bar{\mG}^1$ in Lemma \ref{lm:excluded} follows closely the construction presented in Theorem 3 in \citetalias{GW2005}. 
However, there are a few differences stemming from the fact that the supporting polytope may contain more than one equilibrium component of the excluded game. Parts \eqref{1} and \eqref{2} simply record that the strategies obtained in $\bar{\mG}^1$ are duplicates originating from vertices of a polyhedral subdivision of the strategy sets. This is needed so that \eqref{part:payoffs} and \eqref{part:support} can be meaningfully stated and used in Section~\ref{sec:step3}. Chiefly, in Part \eqref{part:payoffs}, $s_2$ can be viewed as a vertex of the polyhedral subdivision $\mathcal{P}_2$ of $\S^1_2$; and \eqref{part:support} refers explicitly to the vertices of the polyhedral subdivision. The main difference with the construction in \citetalias{GW2005} is \eqref{part:payoffs}. 
Within $\bar{K}^{1,2\eta}$ the payoff perturbations are not necessarily arbitrarily small, i.e., the inequality may be strict.
This contrasts with \citetalias{GW2005} where all payoff perturbations can be made as small as desired.%
%\footnote{~This is an essential feature.  
%As in \citetalias{GW2005}, we obtain the payoff-perturbations from nearby-maps to the best-reply correspondence of the excluded game of player 1  that do not have fixed points on $K^1$. As this set possibly contains many equilibrium components of the excluded game, such a nearby map might be far away from the best-reply correspondence in $K^{1}$. This in turn translates to ``large'' payoff perturbations in $K^{1}$.}
\footnote{~This is a critical difference.  
As in \citetalias{GW2005}, we obtain the payoff-perturbations from a map that does not have fixed points on $K^1$. 
In \citetalias{GW2005} such a map is chosen to be ``close'' to the best-reply correspondence at every point of the domain.
However, in our setting, it is possible for $K^1$ to contain many equilibrium components with nonzero index.
In those cases, a map without fixed point on $K^1$ can be  ``close'' to the best-reply correspondence in the complement of $K^1$ but ``far'' from it on $K^1$.
(See Figure~\ref{fig:perturbation} for an illustration. 
The dotted line is far from the solid line in the interval $(x_1,x_2)$ so that it does not cross the $45^\circ$ line.)
This translates to ``large'' payoff perturbations in $K^{1}$.}

\begin{example}\label{ex:C}
In Example~\ref{ex:A} we established that the supporting polytope $K^1=\D(\{L,R\})\times\big\{\a\ell+(1-\a)r\mid 0\leq \a\leq2/3\big\}$ has index 0.
Therefore, we can find a game equivalent to the subgame in Figure~\ref{Entry} and a perturbation so that the perturbed game has no equilibrium in $K^1$. Since such a subgame is relatively simple, we can accomplish this by directly perturbing it to the following game:
\begin{center}
%Note that the perturbation is only allowed to be ``large'' inside the supporting polytope.

\setlength{\extrarowheight}{.1em}
\begin{tabular}{c|c|c|}
\multicolumn{1}{c}{}&\multicolumn{1}{c}{$\bar{\ell}$}&\multicolumn{1}{c}{$\bar{r}$}  \\ \cline{2-3}
$\bar{L}$&$3,1$&$2,0$ \\ \cline{2-3}
$\bar{R}$&$0,0$&$1,3$ \\ \cline{2-3}
\end{tabular}
\vspace{1em}
\end{center}
\end{example}

\subsection{Step 3: The equivalent game and its perturbation}\label{sec:step3}
We construct a game $\bar{G}$ equivalent to $G$ and a perturbation of $\bar{G}$ that has no equilibrium in a neighborhood of the equivalent component of equilibria $\bar{K}$.
The main idea is to perturb $G$ so that, with vanishing probability, player 1 is forced to play the game $\bar{\mG}^{1,\delta,\e}$ that has been defined in Step 2 and that has no equilibrium in $\bar{K}^1$. We show that, in this perturbed game, players' behavior induce an equilibrium of $\bar{\mG}^{1,\delta,\e}$ so that player 2's strategy cannot belong to the supporting polytope. This, in turn, implies that the perturbed game of $\bar{G}$ cannot have an equilibrium in the fixed original neighborhood of $\bar{K}$.

%Choose $\eta>0$ as in Step 2 and take fine enough simplicial subdivisions of $\S^0_2$ such that no simplex that intersects $\tilde{\partial}K_2^1$ has a point in common with $K_2^{1,\eta}$.
%Let $\l>0$ be the mesh of the resulting simplicial complex and apply Lemma~\ref{lm:excluded}.
Henceforth, extend the original game $\mG$ to the equivalent game $\bar\mG$ in which player 1's strategy set is $S_1\cup\bar{S}^1_1$ and player 2's strategy set is $S_2\cup \bar{S}_2$, where
\begin{equation*}
\bar{S}_2\equiv \Big\{\hat{\tau}^{\otimes}_2 * \bar{s}^0_2\mid \bar{s}^0_2\in\bar{S}^0_2\Big\}.
\end{equation*}
The extension is by equivalence, i.e., by viewing each strategy in $\bar{S}^1_1$ as a duplicate of a strategy in $\S^1_1$, and similarly, viewing each strategy in $\bar{S}_2$ as a duplicate of a strategy in $\S_2$. Let $\bar\mG_1$ and $\bar\mG_2$ be the corresponding bilinear extensions of $\mG_1$ and $\mG_2$ to $\D(S_1 \cup \bar{S}^1_1) \times \D(S_2\cup \bar{S}_2)$.

For a given $\e>0$, we define the extensive-form game $\bar{G}^{\e}$ whose normal-form is $\bar \mG^\e$. 
The extensive-form is represented in Figure~\ref{Auxgame2}.
Nature moves first and chooses between $\kappa_1$ with probability $(1-\e)$ and $\kappa_2$ with probability $\e$. 
Player 1 observes Nature's choices but player 2 does not.
%After $\kappa_1$, player 1 chooses between $\mathit{In}$ and $\mathit{Out}$.
%After choosing $\mathit{In}$, player 1 has to choose an action from $\bar{S}^1_1$.
%Otherwise, after choosing $\mathit{Out}$, player 1 chooses an action from $S_1\setminus S^1_1$.
After $\kappa_1$, player 1 chooses from $(S_1 \setminus S^1_1) \cup \bar{S}^1_1$.
In turn, if Nature chooses $\kappa_2$, player 1 has to choose an action from the set $\bar{S}^1_1$.
Player $2$ observes neither Nature's move nor player $1$'s moves, and selects from $S_2\cup\bar{S}_2$.
In the payoff assignment of Figure~\ref{Auxgame2}, the symbol $\mathbb{1}_{[X]}$ represents the indicator function that equals 1 if $X$ is true and 0 otherwise.
Furthermore, the value of $B>0$ is chosen sufficiently large such that $-B<\min_{\bar{s}_2\in \bar{S}_2}\bar\mG^{1, \d,\varepsilon}_2(\bar{s}^1_1, \bar{s}_2)$ for every $\bar{s}^1_1\in\bar{S}^1_1$ and $\e>0$.

%Intuitively, the game tree is such that, by selecting $\kappa_1$, Nature \textit{allows} player 1 to deviate from the original equilibrium component.
%Then, the (enlarged) set of deviations $\bar S^1_1$ is available if and only if player~1 chooses $\mathit{In}$.
%In turn, by selecting $\kappa_2$, Nature \textit{forces} player~1 to deviate as the set of the choices available after $\kappa_2$ is $\bar S^1_1$.

\begin{figure}
\caption{Game $\bar{G}^{\e}$}\label{Auxgame2}

\vspace{1em}
\begin{istgame}[scale=0.8, every node/.style={scale=0.8}]
\xtdistance{10mm}{70mm}
\istroot(n){Nature}
\istb{P(\kappa_1)=1-\varepsilon}[al] 
\istb{P(\kappa_2)=\varepsilon}[ar]%{\bar{S}_1}[[yshift=-25mm]below] 
\endist

\xtdistance{20mm}{15mm}
\istroot(11)(n-1){1}
\istb
\istb[white]{\dots}[black]
\istb[thick]
\istb
\endist
\xtActionLabel(11)(11-3){s_1 \in (S_1 \setminus S^1_1)\cup\bar{S}^1_1}[xshift=-6pt,yshift=-12pt, fill=white]

\istroot(12)(n-2){1}
\istb
\istb[thick]
\istb[white]{\dots}[black]
\istb
\endist
\xtActionLabel(12)(12-2){\bar{s}^1_1 \in \bar{S}^1_1}[xshift=-2pt, yshift=-12pt,fill=white]

\xtInfoset[dashed](11-1)(12-4){2}[xshift=0mm, above]

\istroot(2a)(11-1){}
\istb[white]{\dots}[black]
\endist

\istroot(2b)(11-4){}
\istb[white]{\dots}[black]
\endist

\istroot(2e)(12-1){}
\istb[white]{\dots}[black]
\endist

\istroot(2f)(12-4){}
\istb[white]{\dots}[black]
\endist

\istroot(22)(11-3){}
\istb
\istb[white]{\dots}[black]
%\istb<level distance= 35mm>[thick]{}{\genfrac{}{}{0pt}{0}{\mG_1(\bar s^1_1, s_2),}{\mG_2(\bar s^1_1, s_2)+\mathbb{1}_{[s_2\in\bar{S}_2]}\e(1-\g(s_2))}}
\istb<level distance= 22mm>[thick]{}{\genfrac{}{}{0pt}{0}{\bar\mG_1(s_1, s_2),}{\bar\mG_2(s_1, s_2)}}
\istb
\endist

\istroot(23)(12-2){}
\istb
\istb[white]{\dots}[black]
%\istb<level distance= 25mm>[thick]{}{\genfrac{}{}{0pt}{0}{\mathbb{1}_{[s_2\in \bar{S}_2]}\bar \mG^{1,\d,\e}_1(\bar s^1_1, q^0_2(s_2))-\mathbb{1}_{[s_2\in S_2]}B,}{\mathbb{1}_{[s_2\in\bar{S}_2]}\bar \mG^{1,\d,\e}_2(\bar s^1_1, q^0_2(s_2)) - \mathbb{1}_{[s_2\in S_2]}B}}
\istb<level distance= 22mm>[thick]{}{\genfrac{}{}{0pt}{0}{\mathbb{1}_{[s_2\in \bar{S}_2]}\bar \mG^{1, \d,\e}_1(\bar s^1_1, q^0_2(s_2)),}{\mathbb{1}_{[s_2\in\bar{S}_2]}\bar \mG^{1, \d,\e}_2(\bar s^1_1, q^0_2(s_2)) - \mathbb{1}_{[s_2\in S_2]}B}}
\istb
\endist

\xtActionLabel(22)(22-3){s_2 \in S_2 \cup \bar{S}_2}[xshift=-7pt, yshift=-14pt, fill=white]

\end{istgame}
\end{figure}

When $\varepsilon=0$, we write the game $\bar{G}^0$ simply as $\bar{G}$ and note that it is equivalent to the original game $G$. 
Let $\bar{K}$ be the equilibrium component equivalent to $K$ in $\bar{G}$ and let $\bar{\mathcal{O}}$ be a neighborhood of $\bar{K}$ in the mixed strategy set of $\bar{G}^{\e}$ whose closure contains no equilibrium of $\bar{G}^{\e}$ in its boundary. Since $\bar{\mathcal{O}}$ contains no equilibria of $\bar{G}$ in its boundary, this is also the case for the game $\bar{G}^{\e}$ for sufficiently small $\e>0$.
We conclude the proof of Theorem \ref{mainthm} by proving that for $\e>0$ small enough there is no equilibrium of $\bar{G}^{\e}$ in $\bar{\mathcal{O}}$.

Similarly to Step 1, we work with behavioral strategies of game $\bar{G}^{\e}$.
A behavioral strategy profile is a member of $\D( (S_1 \setminus S^1_1) \cup\bar{S}^1_1)\times\bar{\S}^1_1\times\D(S_2\cup \bar{S}_2)$.
Suppose to the contrary that there exists a sequence $(\e_k)_{k \in \mathbb{N}}$ of positive numbers converging to $0$ and a sequence $(\sigma^k)_{k \in \mathbb{N}} \subset \bar{\mathcal{O}}$ converging to $\s\in \bar{K}$ such that for every $k$ the profile $\sigma^k$ is an equilibrium of $\bar{\mG}^{\varepsilon_k}$. %with the induced strategy $\tau^k_1$ after $\kappa_1$ putting probability $\alpha_k>0$ in $In$. 
Passing to a subsequence if necessary, assume $(\sigma_1^k,\sigma_2^k)=(\tau^{k}_1,\bar\zeta^{1,k}_1,\s^k_2)\to (\t_1,\bar\zeta^1_1,\s_2)$.
Write $\tau^{k}_1\equiv(1-\alpha_k)\tau^{\times,k}_1+\alpha_k\bar\tau^{1,k}_1$ where $\tau^{\times,k}_1\in\D(S_1\setminus\bar{S}^1_1)$, $\bar\tau^{1,k}_1\in\D(\bar{S}^1_1)$ and $\alpha_k\in[0,1]$.
Similarly, write $\s^k_2 \equiv (1-\beta_k) \bar{\t}^{k}_2 + \beta_k {\t}^k_2$, where $\bar{\t}^{k}_2 \in \Delta(\bar{S}_2)$, ${\t}^k_2\in \S_2$, and $\beta_k\in[0,1]$.
Let $\varphi^k \equiv \e_k[\e_k + (1-\e_k) \alpha_k]^{-1}$.
Again, passing to the corresponding subsequence if necessary, assume $(\t^{\times,k}_1,\bar\t^{1,k}_1)\to(\t^{\times}_1,\bar\t^1_1)$, $(\bar\t^k_2,\t^k_2)\to(\bar\t_2,\t_2)$, $\b^k\to\b$ and $\varphi^k \to \varphi$.

\begin{claim}\label{auxLemma3}
For $k$ large enough, $\alpha_k =0$.  
\end{claim}
\begin{proof}%[Proof of Claim \ref{auxLemma3}]
To the contrary, suppose there is a subsequence of $(\s^k)_{k \in \mathbb{N}}$ for which the corresponding $\alpha_k$'s are all strictly positive. 
%Passing to subsequences if necessary, we can assume that for each $k$, $\a_k>0$.
From Lemma~\ref{lm:kuhn} there is a product strategy $\s^{\otimes}_2*\s^0_2$ in $\mG$ that is equivalent to $\s_2$ in $\bar \mG$.
First observe $\s^0_2\in\tilde{\partial}K^1_{2}$.
Otherwise, $\s^0_2 \in K^1_2 \setminus \tilde{\partial}K^1_{2}$ which, for $k$ sufficiently large,  implies $\alpha_k=0$ since $\t^k_1$ is an equilibrium strategy.
Player 2's expected payoff against $\s^k_1$ in $\bar{G}^{\e_k}$ for a typical strategy $y_2 = (1-{b})\bar{x}_2 + {b}{x}_2$, with $\bar{x}_2 \in \Delta(\bar{S}_2)$, ${x}_2 \in \S_2$, and ${b} \in [0,1]$ is

\begin{multline}\label{expectedplayer2}
[\varepsilon_k + (1-\varepsilon	_k)\alpha_k]
\Bigg(
\varphi^k
\Big[
(1-{b})\bar\mG^{1, \d,\e_k}_2(\bar\zeta^{1,k}_1, q^0_2(\bar{x}_2)) - {b}B
\Big] +
(1-\varphi^k)
\bar\mG_2(\bar\tau^{1,k}_1, y_2)
\Bigg)+\\
(1-\varepsilon_k)(1-\alpha_k)\bar\mG_2(\tau^{\times, k}_1, y_2).
\end{multline}

Player 2's equilibrium strategy $\s^k_2$ maximizes \eqref{expectedplayer2}.
The limit strategy is $\s_2 = (1-\b)\bar \t_2 + \b \t_2$.
For $k$ sufficiently large, $\s_2$ also maximizes \eqref{expectedplayer2} because $\supp(\s_2) \subseteq \supp(\s^k_2)$. 
Denote by $A_2^\otimes$ the subset of strategies $\sigma'_2\in\Delta((S_2 \setminus S^2_2)\cup\bar{S}_2)$ for which $q_2^\otimes(\s'_2)$ is equivalent to $\hat{\t}^\otimes_2$ in game $\mG^{\otimes}$. Since the payoff to player 2 from all strategies in $A^{\otimes}_2$ against $\t^{\times,k}_1$ in $\mG$ is the same,  
%
%\begin{multline}\label{eq:optimal2}
%\s_2\in\arg\max_{x_2\in E^{\otimes}_2}
%\varphi^k 
%\Big[ (1-{\beta})  \bar{\mathbb{G}}^{1,\d,\e}_2(\zeta^{1,k}_1, q^0_2(\bar{x}_2)) - {\beta}B\Big] +\\
%(1-\varphi^k)
%\Big[
%\mathbb{G}_2(\tau^{1,k}_1, x_2) + (1-\b)\e_k\big(1-\sum_{\bar{s}_2\in\bar{S}_2}{\bar{x}_2(\bar{s}_2)\g(\bar{s}_2)\big)}
%\Big].
%\end{multline}
%
\begin{equation}\label{eq:optimal2}
\s_2\in\arg\max_{y_2\in A^{\otimes}_2}
\varphi^k 
\Big[ (1-b)  \bar{\mathbb{G}}^{1,\e_k}_2(\bar\zeta^{1,k}_1, q^0_2(\bar{x}_2)) - b B\Big] +
(1-\varphi^k)
\bar\mG_2(\bar\tau^{1,k}_1, y_2).
\end{equation}

Now, $q^{\otimes}_2(\s_2)$ is equivalent to $\hat{\t}^{\otimes}_2$ in $\mG^{\otimes}$ which implies that $q^{\otimes}_2(\t_2)$ is also equivalent to $\hat{\t}^{\otimes}_2$ in $\mG^{\otimes}$, so that $\t_2 \in A^{\otimes}_2$. 
Therefore, there exists a product strategy $\hat{\t}^{\otimes}_2 * \bar{\t}^0_2$ which is equivalent to $\t_2$ in $\mG$. Hence, $\hat{\t}^{\otimes}_2 * \bar{\t}^0_2 \in \D(\bar{S}_2)$ and 
\begin{align}
\bar \mG_2(\bar \t^{1,k}_1, (1-\b)\bar \t_2 + \b(\hat \t^{\otimes}_2 * \bar \t^0_2)) &= \bar \mG_2(\bar \t^{1,k}_1, \s_2),\\
\bar \mG^{1,\d,\e_k}_2(\bar\zeta^{1,k}_1,(1-\b)\bar \t_2 + \b(\hat \t^{\otimes}_2 * \bar \t^0_2)) &> (1-\b)  \bar{\mathbb{G}}^{1, \d, \e_k}_2(\bar\zeta^{1,k}_1, q^0_2(\bar{\t}_2)) - \b B.
\end{align}

Therefore, it must be that $\b=0$, so $\s_2 = \hat \t^{\otimes}_2 * \bar \t^0_2 \in \D(\bar S_2)$. Furthermore, $\bar\t^0_2$'s support must be within the set of vertices of a polyhedron in $\P_2$, otherwise player 2 can do strictly better by Lemma~\ref{lm:excluded}~\eqref{part:support}. Hence, every $\bar{s}^0_2\in\bar{S}^0_2$ in $\bar\t^{0}_2$'s support satisfies $\bar{s}^0_2\not\in\bar{K}_2^{1,\eta}$ (because of~\eqref{2} in Lemma \ref{lm:excluded}). On the other hand, both $\bar\t^{1,k}_1$ and $\bar\zeta^{1,k}_1$ must be optimal against $\sigma_2^k$ for every $k$. The last two facts together with Lemma~\ref{lm:excluded}~\eqref{part:payoffs} imply that at the limit $\e_k \to 0$, for $k \to +\infty$, we have 
%
%Because of the choice for $B>0$ and the fact that every strategy in $\Delta{(S_2)}\cap E_2^{\otimes}$ has an equivalent strategy in $\Delta{(\bar{S}_2)}$ we have $\beta=0$ so that $\s_2=\bar\t_2=(\hat\s^{\otimes}_2,q^0_2(\bar\t_2))\equiv(\hat\s^{\otimes}_2,\bar\t^0_2)$.
%
%
%
%
\begin{equation*}
\bar\z^{1}_1\in\arg\max_{x_1\in\D(\bar{S}_1^1)} \bar\mG^{1}_1(x_1,\bar\tau^0_2)
\qquad\text{and}\qquad
\bar\t^{1}_1\in\arg\max_{x_1\in\D(\bar{S}_1^1)} \bar\mG_1(x_1,\bar\t_2).
\end{equation*}
Since $\bar\t_2=\hat\t^{\otimes}_2 * \bar\t^0_2$ strategy $\bar\t^1_1$ is, in fact, a best-reply against $\bar\t^0_2$ in the excluded game $\bar\mG^1$.
%Simlarly, $\zeta^{1}_1$ maximizes $\bar\mG^{1}_1(x_1,\sigma^0_2)$ with respect to $x_1$, but since $\bar\mG^{1,\d,\e}_1$ coincides with $\mG^1_1$ when $\sigma^0_2\not\in\bar K^{1,\eta}_2$, \red{check notation} strategy $\zeta^{1}_1$ is also a best-reply against $\sigma^0_2$ in $\mG^1$.
Thus, $\varphi\bar\zeta^1_1+(1-\varphi)\bar\t^1_1$ is also a best-reply against $\bar{\t}^0_2$ in $\bar\mG^1$. 
Recall $\bar{\t}^0_2\in\tilde\partial K_2^1$, hence by Assumption \ref{P2}, player 2 must have  a profitable deviation $\tilde\t^0_2$ in $\bar\mG^1$, i.e.,
\begin{multline*}
\bar\mG^1_2(\varphi\bar\zeta^1_1+(1-\varphi)\bar\t^1_1,\tilde\t^0_2)>
\bar\mG^1_2(\varphi\bar\zeta^1_1+(1-\varphi)\bar\t^1_1,\bar\t^0_2)=
\varphi\bar\mG^1_2(\bar \zeta^1_1,\bar\t^0_2)+(1-\varphi)\bar\mG^1_2(\bar\t^1_1,\bar\t^0_2)=\\
\lim_{k\to\infty}\Big[\varphi^k\bar\mG^{1, \d, \e_k}_2(\bar\zeta^{1,k}_1,\bar\t^0_2)+(1-\varphi^k)\bar\mG_2(\bar\t^{1,k}_1, \hat\t^{\otimes}_2 * \bar\t^{0}_2)\Big],
\end{multline*}
where the last equality uses Lemma~\ref{lm:excluded}~\eqref{part:payoffs} since $\bar{\t}^0_2\notin\bar{K}^{1,\eta}_2$.
Similarly, 
\begin{equation*}
\bar\mG^1_2(\varphi\bar\zeta^1_1+(1-\varphi)\bar\t^1_1,\tilde\t^0_2)=
\varphi\bar\mG^1_2(\bar\zeta^1_1,\tilde\t^0_2)+(1-\varphi)\bar\mG^1_2(\bar\t^1_1,\tilde\t^0_2)\leq
\lim_{k \to\infty}\Big[\varphi^k\bar\mG^{1,\delta,\e_k}_2(\bar\zeta^{1,k}_1,\tilde\t^0_2)+(1-\varphi^k)\bar\mG_2(\bar\t^{1,k}_1, \hat\t^{\otimes}_2 * \tilde\t^0_2)\Big],
\end{equation*}
where we have used again Lemma~\ref{lm:excluded}~\eqref{part:payoffs} in the last inequality. 
From this we have a contradiction given that $\hat{\t}^{\otimes}_2 * \tilde{\t}^0_2$ does strictly better than $\bar{\t}_2$ for large enough $k$.
\end{proof}

\begin{claim}\label{lastLemma} 
For $k$ large enough, $(\bar\zeta^{1,k}_1, q^0_2(\bar\t^k_2))$ is an equilibrium of $\bar{\mG}^{1, \d, \e_k}$. 
\end{claim}
\begin{proof}
From Claim \ref{auxLemma3}, for $k$ sufficiently large, $\alpha_k =0$. 
For such a $k$, player 2's expected payoff from playing $\s^k_2=(1-\beta_k)\bar\t^k_2+\b_k\t^k_2$ against $\s^{k}_1$ is equal to
\begin{multline}\label{eq:lasteq}
\e_k\left [(1-\beta_k)\bar{\mG}^{1,\delta,\e_k}_2(\bar\zeta^{1,k}_1, q^0_2(\bar{\t}^{k}_2)) - \beta_k B\right] + (1-\e_k)\left[\bar\mG_2(\t^{\times,k}_1, (1-\b_k)\bar\t^k_2 + \b_k\t^k_2)\right]=\\
\e_k\left [(1-\beta_k)\bar{\mG}^{1,\delta,\e_k}_2(\bar\zeta^{1,k}_1, q^0_2(\bar{\t}^{k}_2)) - \beta_k B\right] + (1-\e_k)\left[(1-\b_k)\mG^{\otimes}_2(q^\otimes_1(\t^{\times,k}_1), \hat\t^{\otimes}_2) + \b_k\bar\mG_2(\t^{\times,k}_1,\t^k_2)\right].
\end{multline}

Repeating similar arguments to the proof of Claim~\ref{auxLemma3} we obtain $\beta_k\to 0$.
From~\eqref{eq:lasteq} we therefore see that, for $k$ large enough, $q^0_2(\bar\t^k_2)$ must be a best-reply to $\bar\zeta^{1,k}_1$ in $\bar\mG^{1,\d,\e_k}$.
In turn, $\bar\zeta^{1,k}_1$ must also be best-reply to $q^0_2(\bar\t^k_2)$.
That is, $(\bar\zeta^{1,k}_1, q^0_2(\bar\t^k_2))$ is an equilibrium of $\bar{\mG}^{1,\d, \e_k}$. 
\end{proof}  

Let $(\bar \zeta^1_1, q^0_2(\bar\t_2))$ be a limit of $(\bar\zeta^{1,k}_1, q^0_2(\bar\t^k_2))_{k \in \mathbb{N}}$, that is, an equilibrium of $\bar{\mG}^{1,\d,0}$. By construction, $(\bar \zeta^1_1, q^0_2(\bar\t_2))$ belongs to the complement of $\bar K^1$. Therefore, by definition of $\bar K^1$, there exists a pure strategy $\bar s^1_1$ such that $\bar \mG^{1,\d,0}_1(\bar s^1_1, \hat \t^{\otimes}_2 * q^0_2(\bar\t_2))) = \bar \mG_1(\bar s^1_1,  \hat \t^{\otimes}_2 * q^0_2(\bar\t_2)) > G_1(Q)$, where the equality comes from Lemma \ref{lm:excluded} \eqref{part:payoffs}. 
Hence, $\s \notin \bar K$, which implies $\s^k \notin \bar{\mathcal{O}}$ for sufficiently large $k$. 
This contradicts our initial assumption. We conclude that there is no sequence of equilibria of $\bar{G}^{\e_k}$ that converges to some point in~$\bar{\mathcal{O}}$.

%Recall now that by Lemma~\ref{lm:excluded}~\eqref{part:noeq}, game $\bar{\mG}^{1,\e_k}$ does not have an equilibrium such that player 2's strategy is in $\bar{K}_2^{1}$.
%Hence, for some deviation $\bar{s}^1_1\in\bar{S}^1_1$ we have $\bar\mG_1(\bar{s}^1_1,(\hat{\t}^{\otimes}_2,q^0_2(\bar\t^k_2)))>G_1(Q)$.
%In other words, for $k$ large enough, player 1 must have a profitable deviation in game $\bar{G}^{\e_k}$ after $\kappa_1$.

%A consequence of Lemma \ref{lastLemma} is that $\s_2$ does not belong to the supporting polytope $\bar{K}^1_2$ which, in turn, implies $\s \notin \bar\mathcal{O}^0^0$ (from (3) in Lemma \ref{lm:excluded}). This is a contradiction and concludes Step 3. 

\begin{example}
We finish the example started in Section~\ref{Sec:Example} using the same construction as in Figure~\ref{Auxgame2}.
Since the game equivalent to player 1's excluded game has the same strategy space as the excluded game itself, it is enough to ``plug'' the perturbed version found in Example~\ref{ex:C} after $\kappa_2$.
See Figure~\ref{AuxgameExample} for the resulting game.
Following the same line of reasoning as in the general case above, note that player 1 cannot be made indifferent between \textit{In} and \textit{Out} at the limit of a sequence of equilibria as $\e$ goes to zero.
If that was the case, player 2 would be playing $\ell$ with probability close to 2/3 along the sequence which, in turn, implies that player 1 would be playing $T$ and $\bar{T}$ with probability 1 and, therefore, player 2 would have a profitable deviation to playing $\ell$ with probability~1.
Hence, if there is an equilibrium close to the ``\textit{Out}''-component then, for any $\e>0$, player 1 must strictly prefer playing \textit{Out} and player 2 would be choosing her strategy as if she was best-replying to player 1's move after $\kappa_2$.
That is, player 2 would be choosing $\ell$ with probability 1.
But then, player 1 would have an incentive to deviate to \textit{In} and then choose $T$.
We conclude that the perturbed game has no equilibrium close to the ``\textit{Out}''-component and, correspondingly, that such a component is not hyperstable in the original game.
\end{example}

\begin{figure}
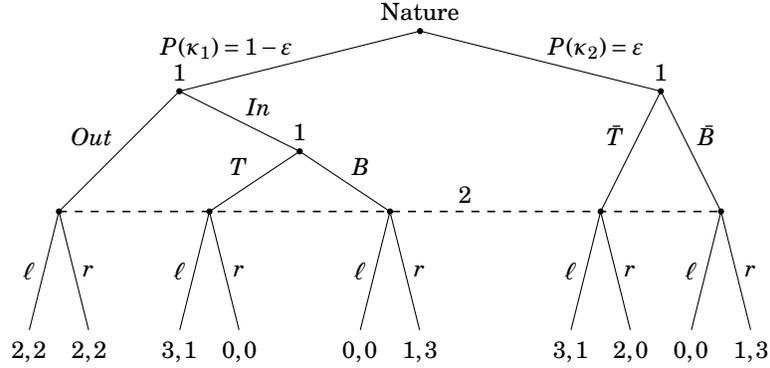

\caption{A perturbation of a game equivalent to Figure~\ref{Entry}}\label{AuxgameExample}

\vspace{1em}
\begin{istgame}[scale=0.8, every node/.style={scale=0.8}]
\xtdistance{10mm}{80mm}
\istroot(n){Nature}
\istb{P(\kappa_1)=1-\varepsilon}[al] 
\istb{P(\kappa_2)=\varepsilon}[ar]%{\bar{S}_1}[[yshift=-25mm]below] 
\endist

\xtdistance{10mm}{40mm}
\istroot(11)(n-1){1}
\istb<level distance=20mm>{\mathit{Out}}[al]
\istb{\mathit{In}}[ar]
\endist

\xtdistance{20mm}{20mm}
\istroot(12)(n-2){1}
\istb{\bar{T}}[al]
\istb{\bar{B}}[ar]
\endist

\xtdistance{10mm}{30mm}
\istroot(13)(11-2){1}
\istb{T}[al]
\istb{B}[ar]
\endist

\xtInfoset[dashed](11-1)(12-2){2}[xshift=12.5mm, above]

\xtdistance{20mm}{10mm}
\istroot(21)(11-1){}
\istb{\ell}[l]{2,2}
\istb{r}[r]{2,2}
\endist

\xtdistance{20mm}{10mm}
\istroot(22)(12-1){}
\istb{\ell}[l]{3,1}
\istb{r}[r]{2,0}
\endist

\xtdistance{20mm}{10mm}
\istroot(23)(12-2){}
\istb{\ell}[l]{0,0}
\istb{r}[r]{1,3}
\endist

\xtdistance{20mm}{10mm}
\istroot(24)(13-1){}
\istb{\ell}[l]{3,1}
\istb{r}[r]{0,0}
\endist

\xtdistance{20mm}{10mm}
\istroot(25)(13-2){}
\istb{\ell}[l]{0,0}
\istb{r}[r]{1,3}
\endist

\end{istgame}

\end{figure}

\section{Applications}\label{sec:application}
In this section we analyze some examples including a few prominent models in the economics literature and show how index computation can be used to eliminate non-hyperstable equilibria.
We begin by revisiting the game in Figure \ref{Entry} and showing that the recent refinement concept of \textit{sequentially stable outcomes} of \citet{D2024} diverges from hyperstability. 
We then move to study two finitely repeated games. 
In particular, Example~\ref{ex:MS} illustrates how relying only on pure-strategy subgame-perfect equilibria to analyze these games may lead to selecting equilibrium outcomes that have index zero. 
We conclude by analyzing some classical signaling games.
This serves both to compare hyperstabily to other equilibrium concepts and to highlight how hyperstable components are easy to compute through their index.  
In preparation for the examples, we note a novel way to compute the index of a component using the indexes of the supporting polytopes in their corresponding excluded games.
The following result is a corollary from the arguments in Section~\ref{sec:step1}.

\begin{proposition}\label{prop:factorization}
Let $G$ be a two-player extensive-form game with normal form $\mG$.
Let $K$ be a component for which assumptions~\ref{P1} and~\ref{P2} are satisfied.
Letting $K^{\otimes} = q^{\otimes}(K)$ be the projected component in game $\mG^{\otimes}$, the index of component $K$ in game $\mG$, written as  $\mathrm{Ind}_{\mG}(K)$, can be decomposed as follows
\begin{equation}
\mathrm{Ind}_{\mG}(K)=\mathrm{Ind}_{\mG^1}(K^1)\times\mathrm{Ind}_{\mG^{\otimes}}(K^{\otimes})\times\mathrm{Ind}_{\mG^2}(K^2),
\end{equation}
where, by convention, $\mathrm{Ind}_{\mG^n}(K^n)\equiv 1$ if player $n$ does not have an excluded game.
\end{proposition}

\begin{example}\label{ex:game3}

\begin{figure}[t!]
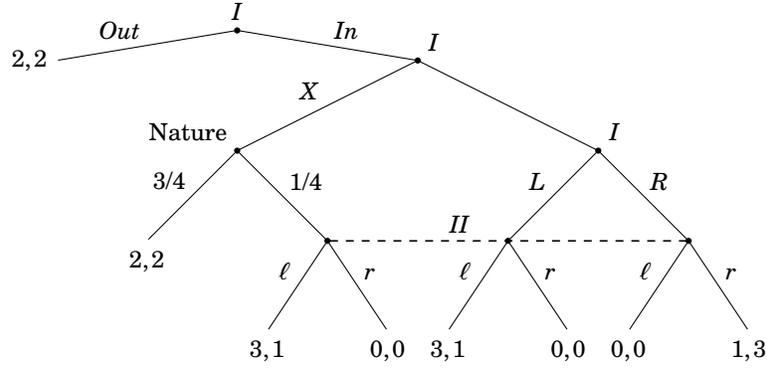

\caption{Equivalent Entry game}\label{Entrymod}
%\hfill
%\begin{minipage}{0.45\textwidth}
\vspace{1em}
\begin{istgame}[scale=0.8, every node/.style={scale=0.8}]
\centering

%\setistmathTF{1}{1}{1}
\xtdistance{5mm}{60mm}
\istroot(0){$I$}
\istb{\mathit{Out}}[al]{2,2}[west]
\istb{\mathit{In}}[ar]
\endist

\xtdistance{15mm}{60mm}

\istroot(1)(0-2)<above right>{$I$}
\istb{X}[al]
\istb{}[ar]
\endist

\xtdistance{15mm}{30mm}

\istroot(2)(1-1)<above left>{Nature}
\istb{3/4}[al]{2,2}
\istb{1/4}[ar]
\endist

\istroot(3)(1-2)<above right>{$I$}
\istb{L}[al]
\istb{R}[ar]
\endist

\xtdistance{15mm}{20mm}

\istroot(4a)(2-2)
\istb{\ell}[al]{3,1}
\istb{r}[ar]{0,0}
\endist

\istroot(4)(3-1)
\istb{\ell}[al]{3,1}
\istb{r}[ar]{0,0}
\endist

\istroot(5)(3-2)
\istb{\ell}[al]{0,0}
\istb{r}[ar]{1,3}
\endist

\xtInfoset[dashed](4a)(5){$\mathit{II}$}[xshift=-8mm]

\end{istgame}

\end{figure}

Sequential equilibrium is one of the most prominent solution concepts in extensive-form games as it extends the principle of backward induction to imperfect information games. 
However, it has shortcomings which have been extensively discussed in the literature \citep{KM1986,CK1987}.
\textit{Sequentially stable outcomes}, recently defined by \citep{D2024}, refine the set of sequential equilibrium outcomes.  

Roughly, an outcome $\omega$ is sequentially stable if for \textit{any} vanishing sequence of behavioral trembles, there exists a sequence of $\e$-sequential equilibria of the (trembled)-game (with vanishing $\e>0$) inducing outcomes which converge to $\omega$ (see \citealp[p.9]{D2024} for the definition).
Note that while the definition of hyperstability considers payoff perturbations explicitly, by invoking $\e$-sequential equilibria, sequentially stable outcomes considers payoff perturbations implicitly.

Sequentially stable outcomes are a natural refinement of sequential equilibrium.
They appropriately refine the set of equilibria in some examples while perhaps being equally tractable.
Unfortunately, similarly to sequential equilibrium, sequentially stable outcomes rely on the specifics of the extensive-form.
Hence, contrasting with hyperstability, they do not satisfy invariance.%
\footnote{~In addition, as already discussed by \citeauthor{D2024}, sequentially stable outcomes do not satisfy admissibility either.}
We can illustrate this difference returning to Figure~\ref{Entry}.
It can be easily checked that ``$\mathit{Out}$'' is a sequential equilibrium outcome.
It is also a sequentially stable outcome (which can be seen by applying \citealp[Proposition 4.3, Part 2]{D2024}).
However, ``$\mathit{Out}$'' is not selected by either concept in the well-known equivalent representation of the game provided in Figure \ref{Entrymod} (cf. \citealp{JH1994}). 
After adding the  redundant strategy $X$, the unique equilibrium of the subgame that follows $\mathit{In}$ is $(L,\ell)$, making $(\mathit{In}-L,\ell)$ the unique Sequential Equilibrium (or Sequentially Stable) outcome.
Not surprisingly, the forward induction outcome $(\mathit{In}-L,\ell)$ is the unique equilibrium selected by hyperstability both in in the representation of Figure~\ref{Entry} and in the representation of Figure~\ref{Entrymod}. 
(Hyperstability can be computed directly in the latter game by noticing that player 1's excluded game associated with ``$\mathit{Out}$'' does not have any equilibrium in the supporting polytope, implying that the supporting polytope has index 0 and, by Proposition~\ref{prop:factorization}, component ``$\mathit{Out}$'' also has index 0, i.e., it is not hyperstable.)

\end{example}

We move to study a finitely repeated game.
As we show below, finitely repeated games have a recursive structure that aids the computation of the index of equilibrium outcomes.\footnote{~An interesting result in the literature on finitely repeated games is \cite{MO1989}. Proposition 1 in this paper provides a sufficient condition for a pure Nash equilibrium outcome (i.e., an outcome where no randomizations are made at any period by both players) of a finitely repeated, two-player game \textit{not} to be Kohlberg-Mertens stable. 
In particular, this condition implies that this outcome has an index of zero: any non-zero index equilibrium component in mixed strategies contains a Kohlberg-Mertens stable set. If an outcome is not Kohlberg-Mertens stable, then the associated index is zero.
} 

\begin{example}\label{ex:MS}
The repeated game with stage-game in Figure~\ref{fig:rg2} is taken from \citet[Example 4.4.1]{MS2006}.
The game has 3 equilibria in pure strategies $(A,A)$, $(B,B)$ and $(C,C)$ and 4 equilibria in mixed strategies $(\frac{1}{5}A+\frac{4}{5}B,\frac{3}{7}A+\frac{4}{7}B)$, $(\frac{1}{5}A+\frac{4}{5}C,\frac{1}{3}A+\frac{2}{3}C)$, $(\frac{3}{4}B+\frac{1}{4}C,\frac{1}{4}B+\frac{3}{4}C)$ and  $(\frac{1}{17}A+\frac{12}{17}B+\frac{4}{17}C,\frac{3}{13}A+\frac{4}{13}B+\frac{6}{13}C)$.
Pure strategy equilibria are strict and therefore have index +1.
Mixed strategy equilibria that are not completely mixed have index -1.%
\footnote{~
The latter can be seen by deleting the unused pure strategy for both players for each equilibrium as it is an inferior response. 
The resulting truncated game has 3 equilibria, 2 of them strict and one in completely mixed strategies which, therefore, has index -1.
}
Since indexes must add up to one, the completely mixed equilibrium has index +1.

\begin{figure}[t]
\caption{Stage-game for a repeated game with player specific punishments}\label{fig:rg2}

\begin{center}
\setlength{\extrarowheight}{.1em}
\begin{tabular}{c|c|c|c|}
\multicolumn{1}{c}{}&\multicolumn{1}{c}{$A$}&\multicolumn{1}{c}{$B$}&\multicolumn{1}{c}{$C$}  \\ \cline{2-4}
$A$&$4,4$&$0,0$&$0,0$ \\ \cline{2-4}
$B$&$0,0$&$3,1$&$0,0$ \\ \cline{2-4}
$C$&$2,2$&$0,0$&$1,3$ \\ \cline{2-4}
\end{tabular}
\vspace{1em}
\end{center}
\end{figure}

Examples such as the twice-repeated Prisoners Dilemma do not allow for equilibria where in the first stage a profile different from the unique equilibrium of the stage game is played. This contrasts with this example
where the stage-game has more than one equilibrium that can be used differently to punish/reward players in the associated finitely repeated game, and allows for subgame perfect equilibrium outcomes in which a profile that is not necessarily an equilibrium of the stage-game to be played in early stages.
Indeed, \citet[p. 113]{MS2006} note that, if the game in Figure~\ref{fig:rg2} is played twice, then it is possible to support $(C,A)$ in the first period in subgame perfect equilibrium by playing $(A,A)$ in the second stage if no player deviated in the first stage, $(B,B)$ if the second player deviated, and $(C,C)$ if the first player deviated. 
We show in this particular case that supporting the non-equilibrium profile in the first stage is not plausible (from the point of view of hyperstability). 
To be precise, we show that the component $K$ which induces the outcome ``6,6'' in which players play on-path $(C,A)$ and then $(A,A)$ is not hyperstable.

First, we note that such an outcome is not isolated in the set of equilibrium outcomes.
Indeed, if ``2,6'' is the outcome in which players play on-path $(C,C)$ twice, the same equilibrium component induces the continuum of outcomes $\a$``6,6''$+(1-\a)$``2,6'' for $0\leq\a\leq 1$.
Therefore, Assumption~\ref{P1} is not satisfied and we cannot directly apply the results in this paper to analyze the equilibrium component $K$.
Nonetheless, we can use Property~\ref{I1} and perturb player 2's payoffs slightly so that, \textit{in the first stage}, her payoff is $3-\varepsilon$ under profile $(C,C)$ with $\varepsilon>0$.
Call the resulting perturbed game $G(\e)$.
The index of any small enough admissible neighbourhood $\mathcal{O}$ of $K$ in the original game coincides with the index of $\mathcal{O}$ in $G(\e)$ provided $\varepsilon>0$ is small enough.
Furthermore, in $G(\e)$, the unique component in $\mathcal{O}$ does induce a unique outcome and, in this unique outcome, players play on path $(C,A)$ and then $(A,A)$.

Let us compute player 1's excluded game under such an outcome. 
Since it does not depend on $\e$, we simply denote it as $\mG^1$.
Any deviation that involves player 1 playing $B$ in the first stage gives player 1 a payoff of at most 3, which is strictly smaller that her payoff of 6 under the equilibrium component.
Therefore, we may only consider player 1's deviations in which she plays $A$ in the first stage.
After eliminating duplicates, player 1's (simplified) excluded game is represented in Figure~\ref{fig:rg2b} (left). 
Note that payoffs are obtained from Figure~\ref{fig:rg2} after adding to each payoff vector the vector (4,4) and, consequently, their set of Nash equilibria coincide.
Player 1's supporting polytope consists of those strategy profiles in $\mG^1$ that yield player 1 a payoff smaller or equal than 6.
This polytope includes equilibrium $(\mathit{AC},C)$ as well as every mixed strategy equilibrium.
Hence, the index of the supporting polytope in player 1's excluded game is -1.

In turn, consider player 2's excluded game $\mG^{2}(\e)$ in Figure~\ref{fig:rg2b} (right).
Analogously to $\mG^{1}$, this (simplified) excluded game is constructed by noticing that every deviation in which player 2 plays $B$ is an inferior response to the equilibrium outcome and by eliminating duplicates.
Payoffs are obtained from Figure~\ref{fig:rg2} after adding to each payoff vector the vector $(1,3-\e)$. 
That is, both games also have the same Nash equilibria.
Player 2's supporting polytope consists of those strategy profiles in $\mG^2(\e)$ that yield player 2 a payoff smaller than 6.
This polytope includes every equilibrium of $\mG^2(\e)$ but $(A,\mathit{CA})$ and, therefore, the index of the supporting polytope in player 2's excluded game is zero.
Using Proposition~\ref{prop:factorization}, the index of the component $K$ that induces the outcome in which players play $(A,C)$ in the first stage is zero for $\varepsilon\geq 0$ small enough.
That is, $(A,C)$ cannot be sustained in the first stage in a hyperstable equilibrium component in the two-fold repetition of the game in Figure~\ref{fig:rg2}.
\end{example}

\begin{figure}[t]
\caption{Excluded games $\mG^1$ (left) and $\mG^2(\e)$ (right).}\label{fig:rg2b}
\vspace{1em}

\hfill
\setlength{\extrarowheight}{.1em}
\begin{tabular}{c|c|c|c|}
\multicolumn{1}{c}{}&\multicolumn{1}{c}{$A$}&\multicolumn{1}{c}{$B$}&\multicolumn{1}{c}{$C$}  \\ \cline{2-4}
$\mathit{AA}$&$8,8$&$4,4$&$4,4$ \\ \cline{2-4}
$\mathit{AB}$&$4,4$&$7,5$&$4,4$ \\ \cline{2-4}
$\mathit{AC}$&$6,6$&$4,4$&$5,7$ \\ \cline{2-4}
\end{tabular}
\hfill
\begin{tabular}{c|c|c|c|}
\multicolumn{1}{c}{}&\multicolumn{1}{c}{$\mathit{CA}$}&\multicolumn{1}{c}{$\mathit{CB}$}&\multicolumn{1}{c}{$\mathit{CC}$}  \\ \cline{2-4}
$A$&$5,7-\e$&$1,3-\e$&$1,3-\e$ \\ \cline{2-4}
$B$&$1,3-\e$&$4,4-\e$&$1,3-\e$ \\ \cline{2-4}
$C$&$3,5-\e$&$1,3-\e$&$2,6-\e$ \\ \cline{2-4}
\end{tabular}
\hfill{}
\vspace{1em}
\end{figure}

\begin{example}
\citet{VD1989} shows that in the two-fold repetition of the stage game in Figure~\ref{stagegame} the outcomes in which players either play $(T,L)$ or $(B,R)$ twice is not stable in the sense of \citet{KM1986} whereas the outcome in which players alternate is, in turn, stable.
In addition, \citet[p. 429]{VD1989} also points out that ``[c]omputing the set of all stable equilibrium paths seems to be laborious, however, even in this most simple conceivable case.''
Here, we illustrate how the properties of the index and the results in this paper can be used to compute the index of every equilibrium component in this repeated game.
We also note that every nonzero index component contains a stable set (both in the sense of~\citealp{KM1986}, and in the sense of \citealp{M1989}.)

Call $G$ the two-fold repetition of the game in Figure~\ref{stagegame}.
The stage game of $G$ has 3 equilibria.
The two pure equilibria $(T,L)$ and $(B,R)$ are strict and therefore have index $+1$.
The completely mixed equilibrium $(\frac{4}{5}T + \frac{1}{5}B, \frac{1}{5}L + \frac{4}{5}R)$ has index $-1$.
%Thus, the extensive-form game $G$ satisfies \ref{P1} and, as it will be clear below, also \ref{P2}.

First, every equilibrium component in which players play a completely mixed action profile in the first period has nonzero index.
Indeed, if in such a component players always play the completely mixed equilibrium of the stage game after any realization in the first stage, then this equilibrium component induces a completely mixed outcome which, therefore has nonzero index (cf. Proposition~\ref{nonzeroindex}). 
On the other hand, if for some realization in the first stage, players play a strict equilibrium in the second stage then every observable deviation consists of deviating in the second stage from the prescribed strict equilibrium.
Those deviations are inferior responses to the equilibrium outcome and after eliminating them the outcome is completely mixed and, therefore, has nonzero index.

Let ``8,2'' be the equilibrium outcome in which players play $(T,L)$ twice.
Using a similar argument as above, deviating only in the second stage is an inferior response to the equilibrium outcome. 
After eliminating those deviations from $G$, we can see that player 1 has (after eliminating duplicates) only 2 observable deviations, $\mathit{BT}$ and $\mathit{BB}$ (that is, playing $B$ in the first stage and then playing either $T$ or $B$ in the subgame that follows $(B,L)$).
Player 2's strategy set in player 1's excluded game is $\{L,R\}$.
Note that if player 1 deviates to either $\mathit{BT}$ or $\mathit{BB}$ in $G$ then the payoff vector accrued in the first stage is $(0,0)$ given that player 2 is playing $L$ in such a stage.
Therefore, the excluded game $\mathbb{G}^1$ coincides with the game in Figure~\ref{stagegame}.

Player 1's payoff in $\mathbb{G}^1$ is always strictly smaller than player 1's payoff of 8 if $(T,L)$ is played twice.
Hence, the supporting polytope $K^1$ coincides with the entire strategy space in $\mathbb{G}^1$ which implies that $K^1$ has index +1 in $\mG^1$.
In turn, consider player 2's observable deviations $\mathit{RL}$ and $\mathit{RR}$. 
The excluded game $\mathbb{G}^2$ again coincides with the game in Figure~\ref{stagegame}.

The supporting polytope $K^2$ requires player 1 to play $B$ with probability smaller or equal than $0.5$.
The only equilibrium that lies outside that region is $(B,R)$ which is strict and, therefore, has index $+1$. 
Thus, the supporting polytope $K^2$ has zero index in game $\mG^2$. 
Proposition~\ref{prop:factorization} implies that the component ``8,2'' has zero index and, therefore, it is not hyperstable.
The analogous argument of course holds for the equilibrium component ``2,8''.
It also holds for the equilibrium outcome in which players play the completely mixed equilibrium of the stage game after playing one of the strict equilibria in the first stage.

\begin{figure}[t]
\caption{Coordination stage-game}\label{stagegame}
\begin{center}
\setlength{\extrarowheight}{.1em}
\begin{tabular}{c|c|c|}
\multicolumn{1}{c}{}&\multicolumn{1}{c}{$L$}&\multicolumn{1}{c}{$R$}  \\ \cline{2-3}
$T$&$4,1$&$0,0$ \\ \cline{2-3}
$B$&$0,0$&$1,4$ \\ \cline{2-3}
\end{tabular}
\vspace{1em}
\end{center}
\end{figure}

Finally, consider now the two equilibrium outcomes ``5,5'', one in which players play $(T,L)$ first and $(B,R)$ second, and the other in which this order is reversed.
In both cases player 1 and player 2's excluded games coincide with the game in Figure~\ref{stagegame}.
Thus, every payoff in the two excluded games associated to the corresponding equilibrium component is smaller for both players than the payoff 5 induced by the equilibrium outcome.
This implies that both supporting polytopes have index +1 and, therefore, the two components that induce the equilibrium payoff vector $(5,5)$ also have index +1 and are hyperstable.
\end{example}

%We begin illustrating some implications of our results using  signaling games.
We now turn our attention to signaling games.
Signaling games have driven the creation of a wealth of refinement concepts to eliminate implausible equilibria. 
Some of the most well-known are the Intuitive Criterion, D1, D2, Never-weak-best-reply (NWBR), and Universal Divinity. 
Hyperstability is strictly stronger than any of these criteria and, as we show, it is easy to compute by using properties of the index. 
%
%We take as a definition of a Signaling Game the standard one used in Cho and Kreps (cf. \cite{CK1987}). 
The next proposition is known but we recall it for completeness. 

\begin{proposition}\label{sigprop}
Let $G$ be a finite signaling game. 
Suppose the terminal payoffs of $G$ are chosen according to \ref{P1} and \ref{P2} and $Q$ is an equilibrium outcome which is hyperstable. 
Then $Q$ satisfies the Intuitive Criterion, D1, D2, NWBR and Universal Divinity. 
\end{proposition}

The result follows from observing that a stable outcome in the sense of \citet{KM1986} satisfies all the criteria listed above and that every hyperstable outcome is also Kohlberg-Mertens stable. 

\begin{figure}[t!]
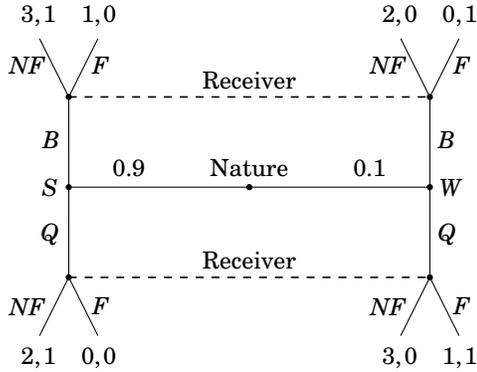

\caption{Beer-and-Quiche Game}\label{BandQ}
%\bigskip
\vspace{1em}
\begin{minipage}{0.55\textwidth}
\centering
\begin{istgame}[scale=0.8, every node/.style={scale=0.8}]
%\xtShowEndPoints[solid node,fill=red]
\xtdistance{0mm}{60mm}
%\setistmathTF{1}{1}{1}
\istroot(o)(0,0){Nature} % names the root as (0) at (0,0)
\istb{0.9}[above, xshift=-14pt]{S}[l] % endpoint will be (0-1), automatically
\istb{0.1}[above, xshift=14pt]{W}[r] % endpoint will be (0-2), automatically
\endist % end of simple (parent-child) structure

\xtdistance{0mm}{30mm}

\istroot[left](w)(o-1)
\istb{B}[l]
\istb{Q}[l]
\endist

\istroot'[right](s)(o-2)
\istb{B}[r]
\istb{Q}[r]
\endist

\xtInfoset[dashed](w-1)(s-1){Receiver}[a]
\xtInfoset[dashed](w-2)(s-2){Receiver}[a]

\xtdistance{10mm}{10mm}

\istroot[down](a)(w-2)
\istb{\mathit{NF}}[l]{2,1}[south]
\istb{F}[r]{0,0}[south]
\endist

\istroot[down](b)(s-2)
\istb{\mathit{NF}}[l]{3,0}[south]
\istb{F}[r]{1,1}[south]
\endist

\istroot[up](c)(w-1)
\istb{F}[r]{1,0}[north]
\istb{\mathit{NF}}[l]{3,1}[north]
\endist

\istroot[up](d)(s-1)
\istb{F}[r]{0,1}[north]
\istb{\mathit{NF}}[l]{2,0}[north]
\endist
\end{istgame}
\end{minipage}
\hfill
\begin{minipage}{0.40\textwidth}
\centering
Sender's excluded game in $\mathit{QQ-NF}$ component:

\setlength{\extrarowheight}{.1em}
\begin{tabular}{c|c|c|}
\multicolumn{1}{c}{}&\multicolumn{1}{c}{$\mathit{NF}$}&\multicolumn{1}{c}{$F$}  \\ \cline{2-3}
$\mathit{BQ}$&$3,9/10$&$6/5,0$ \\ \cline{2-3}
$\mathit{QB}$&$2,9/10$&$9/5,1$ \\ \cline{2-3}
$\mathit{BB}$&$29/10,9/10$&$9/10,1/10$ \\ \cline{2-3}
\end{tabular}
\end{minipage}
\vspace{1em}
\end{figure}

\begin{example} Consider the Beer-Quiche Game from \cite{CK1987} (see Figure \ref{BandQ}). 
The game has two equilibrium outcomes.
In the first one, Sender chooses Beer ($B$) for both of his types and Receiver, seeing this, does Not Fight ($\mathit{NF}$). 
Off the equilibrium path, Receiver Fights ($F$) with probability at least $0.5$. 
This equilibrium outcome survives all the classical refinements listed in Cho and Kreps and there are good reasons for it to be selected (see \citealp[pp. 184-185]{CK1987}). 
In the other equilibrium outcome, Sender chooses Quiche ($Q$) for both of his types and, upon seeing this, Receiver does not fight. 
To prevent the strong type ($S$) from deviating to Beer ($B$),  Receiver fights with probability at least $0.5$ off the equilibrium path. 
Cho and Kreps provide intuitive reasons under which Receiver's off-equilibrium beliefs are implausible (see p.185) and this equilibrium outcome does not survive the Intuitive Criterion. 
Therefore, it is not hyperstable (see Proposition \ref{sigprop}). 
But we can verify this directly by computing the index of the second equilibrium outcome.

Sender's excluded game associated to this equilibrium outcome is in the right-hand side of Figure~\ref{BandQ} where $\mathit{QB}$, for instance, is the deviation in which sender chooses $Q$ after $S$ and $B$ after $W$.
First note that the supporting polytope consists of those strategy profiles in the excluded game whose payoff to Sender is smaller than or equal to $\frac{21}{10}$ (i.e. her payoff under the $\mathit{QQ}-\mathit{NF}$ component).
In this excluded game, $\mathit{BB}$ is strictly dominated.
Furthermore, the excluded game has three equilibria, $(\mathit{BQ},\mathit{NF})$, $(\mathit{QB},\mathit{F})$, and $(\frac{1}{10}\mathit{BQ}+\frac{9}{10}\mathit{QB},\frac{3}{8}\mathit{NF}+\frac{5}{8}\mathit{F})$.
Sender's payoff under these equilibria is, respectively, $3$, $\frac{9}{5}$, and $\frac{15}{8}$.
That is, only the last two belong to the supporting polytope and, since strict equilibria have index +1 and the sum of all the indexes must be equal to +1, the index of the supporting polytope is zero. 
%
%and the other is mixed, their 
%Any pure strategy of the Sender which assigns Beer after $S$ is a strictly inferior reply to any mixed strategy in the connected component of equilibria inducing such an outcome. 
%Eliminating these pure strategies does not alter the index Property \ref{I3}. 
%After eliminating action $B$ at type $S$'s information set (and node and branches that follow), Receiver has a subgame after Beer and it is clear that Outcome $\mathit{QQ-NF}$ is not subgame-perfect. 
%(If it was then Receiver must choose $\mathit{NF}$ off equilibrium path, which prompts the strong Sender to deviate from $Q$ to $B$, upsetting the equilibrium). Since nonzero index outcomes always satisfy backward induction, 
Proposition~\ref{prop:factorization} implies that $\mathit{QQ-NF}$ has index zero and, consequently, it is not hyperstable. 
%
%There are other ways of computing the index of this equilibrium outcome without resorting to verifying properties such as backward induction. 
%The next example shows a different method that could also be applied to the Beer-Quiche game in the same way.  
\end{example}

% signaling game

\begin{example}
Consider the game in \citet[Figure IV, p. 207]{CK1987} and represented in Figure~\ref{fig:FigIV}. 
The game has two equilibrium outcomes.
In the first, both types send $m'$ and, off the equilibrium path, Receiver plays in the convex hull of $(2/3,0, 1/3)$, $(1/6, 1/2, 1/3)$, $(0, 1/2, 1/2)$ and $(0,0,1)$ where probabilities in each vector correspond, in order, to strategies $r_1$, $r_2$ and $r_3$. 
This first component satisfies the Intuitive Criterion, $D1$, $D2$, and Universal Divinity.
In the second equilibrium outcome both types send $m$ and Receiver plays $r_2$.
Since the second equilibrium outcome is induced by a strict equilibrium (hence with index$ +1$) and the indexes of all equilibria always add up to +1 (Property~\ref{I3}) the first component has zero index.
While simple, this reasoning requires knowing all the equilibrium outcomes but, as before, we can also directly compute its index.

\begin{figure}
\caption{Figure IV in Cho \& Kreps}\label{fig:FigIV}
\vspace{1em}
\begin{minipage}{0.55\textwidth}
\centering
\begin{istgame}[scale=0.8, every node/.style={scale=0.8}]

\xtdistance{30mm}{0mm}

\istroot(0)[chance node]{Nature}
\istb<grow=left>{1/2}[a]
\istb<grow=right>{1/2}[a]
\endist

\xtdistance{15mm}{30mm}

\istroot(1)(0-1)<180>{$t_1$}
\istb<grow=north>{m'}[l]{0,0}[north]
\istb<grow=south>{m}[l]
\endist

\istroot(2)(0-2)<0>{$t_2$}
\istb<grow=north>{m'}[r]{0,0}[north]
\istb<grow=south>{m}[r]
\endist

\istroot'[north](a1)(1-1)

\xtdistance{15mm}{10mm}

\istroot(b1)(1-2)
\istb{r_1}[al]{-1,3}
\istb{r_2}[right,yshift=-3mm]{1,2}
\istb{r_3}[ar]{-1,0}
\endist

\istroot(a2)(2-2)
\istb{r_1}[al]{1,0}
\istb{r_2}[right,yshift=-3mm]{1,2}
\istb{r_3}[ar]{-2,3}
\endist
%\xtInfoset(a1)(b2){2}

\xtInfoset[dashed](b1)(a2){Receiver}
\end{istgame}
\end{minipage}
\hfill
\begin{minipage}{0.40\textwidth}
\centering
Sender's excluded game in $\mathit{\mathit{m'm'}}$ component:
\setlength{\extrarowheight}{.1em}
\begin{tabular}{c|c|c|c|}
\multicolumn{1}{c}{}&\multicolumn{1}{c}{$r_1$}&\multicolumn{1}{c}{$r_2$}&\multicolumn{1}{c}{$r_3$} \\ \cline{2-4}
$\mathit{m'm}$&1/2, 0&1/2,1&-1,3/2 \\ \cline{2-4}
$\mathit{mm'}$&-1/2,3/2&1/2,1 &-1/2,0 \\ \cline{2-4}
$\mathit{mm}$&0, 3/2&1, 2& -3/2, 3/2 \\ \cline{2-4}
\end{tabular}
\end{minipage}
\end{figure}

%, but not NWBR.%
%\footnote{~The authors are cautious about applying the NWBR criterion, see p. 208 of Cho and Kreps.
%\red{What does this mean?}} 

%Since, they the only two equilibrium outcomes, it is immediate to see that the first component has zero index as the second equilibrium is strict (hence, index +1 by Property~\ref{I3}) and the sum of the indexes must be +1. 

Sender's excluded game in the component is on the right-hand side of Figure~\ref{fig:FigIV}.
Its unique equilibrium is $(\mathit{mm}, r_2)$ and it lies outside the supporting polytope which, consequently, has index zero.
Therefore, by Proposition~\ref{prop:factorization}, the $m'm'$ component has index zero and, therefore, it is not hyperstable. 
\end{example}

\begin{example}
We move to the example in \citet[p.216]{CK1987} represented in Figure~\ref{fig:3types}.
%We consider the version in which Nature chooses each of the three types with equal probability. 
This example illustrates how Kohlberg-Mertens stability is strictly stronger than NWBR, D1, D2, Universal Divinity, and the Intuitive Criterion. 
Similarly to the previous example, this game has two equilibrium outcomes.
The first one is the strict equilibrium $(\mathit{mm'm'},r_1$), where $mm'm'$ means player 1 plays $m$ in his leftmost information set, and $m'$ in the middle and rightmost information sets. 
Hence, the other equilibrium outcome in which all three types send $m'$ must have index $0$ and not be hyperstable. 
The formal argument sketched by Cho and Kreps that verifies that the second equilibrium outcome is not Kohlberg-Mertens stable (although satisfying all other refinement criteria) is more involved than what we showed here and, because of the unintuitive nature of their characterization of stability, less convincing.%
\footnote{~This is a position Cho and Kreps express themselves, see last paragraph of p. 220.} 
In contrast, using the properties of the index, identifying hyperstability is easy.
Its intuitiveness comes from the fact that it captures payoff robustness in every equivalent game.  

%However, as is in the previous example, it is not necessary to know both equilibrium outcomes to compute the index of the pooling equilibrium outcome. 
The supporting polytope in Sender's excluded game is a polyhedron obtained by permuting the vector $(2/3, 1/3,0)$ in Receiver's mixed strategy set.
% of the Receiver (i.e., denoting a mixed strategy of the Receiver by a $1 \times 3$ vector where the first entry corresponds to the probability of $r_1$, the second, the probability of $r_2$, the third the probability of $r_3$, the vertices of the supporting polytope are permutations of the vector $(2/3, 1/3, 0)$). 
This supporting polytope contains two equilibria in its interior.
Both indices can be explicitly computed using Shapley's formula (\citealp{S1974}).
One has index $+1$ and the other index $-1$, which implies that the supporting polytope in Sender's excluded game has index zero and, by Proposition~\ref{prop:factorization}, the outcome in which every type chooses $m'$ has zero index and, therefore, fails to be hyperstable.
\end{example}

\begin{figure}[t]
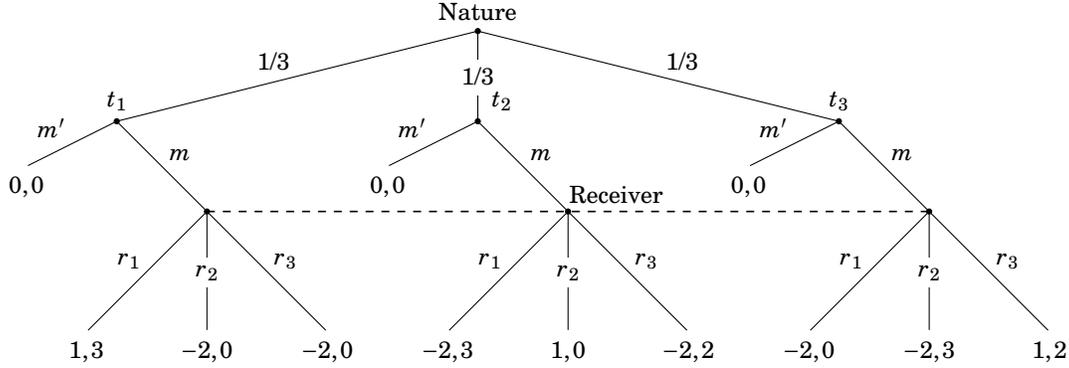

\caption{A signaling game with three types}\label{fig:3types}

\vspace{1em}

\begin{istgame}[scale=0.8, every node/.style={scale=0.8}]
\xtdistance{15mm}{60mm}
\istroot(n1){Nature}
\istb{1/3}[al]
\istb{1/3}[fill=white]
\istb{1/3}[ar]%{\bar{S}_1}[[yshift=-25mm]below] 
\endist

\xtdistance{15mm}{30mm}
\istroot(11)(n1-1){$t_1$}
\istbA<level distance=.5*\xtlevdist>{m'}[al]{0,0}
\istbA<level distance=\xtlevdist>{m}[ar]
\endist
%\xtActionLabel(11)(11-3){s^1_1\in S^1_1}[yshift = -10pt, fill=white]

\istroot(12)(n1-2){\qquad$t_2$}
\istbA<level distance=.5*\xtlevdist>{m'}[al]{0,0}
\istbA<level distance=\xtlevdist>{m}[ar]
\endist

\istroot(13)(n1-3){$t_3$}
\istbA<level distance=.5*\xtlevdist>{m'}[al]{0,0}
\istbA<level distance=\xtlevdist>{m}[ar]
\endist

\xtInfoset[dashed](11-2)(13-2){Receiver}[xshift=8mm]

\xtdistance{20mm}{20mm}
\istroot(23)(11-2){}
\istb{r_1}[al]{1,3}
\istb{r_2}[fill=white]{-2,0}
\istb{r_3}[ar]{-2,0}
\endist

\istroot(23)(12-2){}
\istb{r_1}[al]{-2,3}
\istb{r_2}[fill=white]{1,0}
\istb{r_3}[ar]{-2,2}
\endist

\istroot(23)(13-2){}
\istb{r_1}[al]{-2,0}
\istb{r_2}[fill=white]{-2,3}
\istb{r_3}[ar]{1,2}
\endist

\end{istgame}
\end{figure}

\begin{example}
We conclude with a finite game version of Spence's job-market model in Figure \ref{SJM}. 
There are two equilibrium outcomes, one pooling and one separating and hyperstability selects the separating equilibrium outcome. 
We can analyze the game in the same manner as the Beer-and-Quiche game to show that the pooling equilibrium (where the Sender sends $e$ for both types, and the Receiver replies with $w$ on equilibrium path, and off equilibrium path, with a probability of $w$ at least $1/2$) has index~$0$.

Interestingly, when we consider the general Spence model \citep[Section V]{CK1987} with any number of types and a continuum of education levels and wages, the same reasoning can be used to eliminate any pooling or semi-pooling equilibrium outcome, leaving only the fully separating one (the so-called Riley outcome) left.
To show that, consider a (semi-)pooling equilibrium outcome where a certain number of types pool at education level $e_p$. 
To prevent the highest of the pooling types, denoted $\theta^*$, from deviating from $e_p$ to a higher education level $e > e_p$ where the associated equilibrium wage $w(e) < (\theta^* e)$ is offered (off-path) after $e$ renders a deviation of all pooling types but $\theta^*$ to $e$ a strictly inferior reply (this follows from single-crossing). 
Therefore, we can eliminate all corresponding pure strategies that assign $e$ to all such pooling types. 
After this elimination, the extensive-form has a subgame after education level $e$. Backward induction now implies that firms best-reply with wage level $(\theta^*e) > w(e)$, which upsets the (semi-)pooling equilibrium. 

We emphasize that the characterization in our main theorem is not known to hold in this infinite dimensional setting.
The purpose here is only to remark that the same two principles invoked previously to compute the index in the finite game setting (elimination of strictly inferior replies and backward induction) can also be invoked in the Spence model to eliminate pooling and semi-pooling equilibria.

\begin{figure}
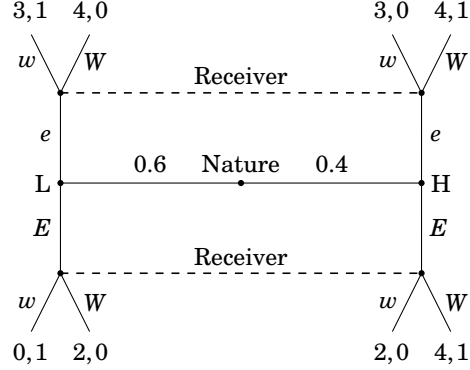

\caption{Spence job-market game}\label{SJM}
\bigskip
\bigskip
\begin{istgame}[scale=0.8, every node/.style={scale=0.8}]%\xtShowEndPoints[solid node,fill=red]
\xtdistance{0mm}{60mm}
%\setistmathTF{1}{1}{1}
\istroot(o)(0,0){Nature} % names the root as (0) at (0,0)
\istb{0.6}[a]{$L$}[l] % endpoint will be (0-1), automatically
\istb{0.4}[a]{$H$}[r] % endpoint will be (0-2), automatically
\endist % end of simple (parent-child) structure

\xtdistance{0mm}{30mm}

\istroot[left](w)(o-1)
\istb{e}[l]
\istb{E}[l]
\endist

\istroot'[right](s)(o-2)
\istb{e}[r]
\istb{E}[r]
\endist

\xtInfoset[dashed](w-1)(s-1){Receiver}[a]
\xtInfoset[dashed](w-2)(s-2){Receiver}[a]

\xtdistance{10mm}{10mm}

\istroot[down](a)(w-2)
\istb{w}[l]{0,1}[south]
\istb{W}[r]{2,0}[south]
\endist

\istroot[down](b)(s-2)
\istb{w}[l]{2,0}[south]
\istb{W}[r]{4,1}[south]
\endist

\istroot[up](c)(w-1)
\istb{W}[r]{4,0}[north]
\istb{w}[l]{3,1}[north]
\endist

\istroot[up](d)(s-1)
\istb{W}[r]{4,1}[north]
\istb{w}[l]{3,0}[north]
\endist

\end{istgame}
\end{figure}
\end{example}

\appendix

\section{Proof of Proposition~\ref{nonzeroindex}}\label{app:nonzeroindex}
We introduce enabling strategies.
Fix a two-player finite game tree $\Gamma$.
For each player $n=1,2$ and $z \in Z$, let $A_n(z)$ be the set of Player $n$'s actions that precede $z$.
The set of Player $n$'s actions that are $\prec$-maximal in $A_n(z)$ for some $z \in Z$ are denoted $L_n$ and called \textit{Player $n$'s last actions in $\Gamma$}. 
If $L_n = \emptyset$ then $n$ is a dummy player.
Player $n$'s last action that precedes $z$ is $\ell_n(z)$.  
Given last action $i \in L_n$, let $u \in U_n$ satisfy $i \in A_n(u)$. 
The subset $S_n(i)\subset S_n$ consists of those strategies that at every information set $u'$ of Player $n$ that precede $u$ prescribe the unique action leading to $u$. 
Define $q^e_n: \S_n \to [0,1]^{L_n}$ so that for every $\sigma_n \in \S_n$ and $i \in L_n$ we have $q^e_n(\s_n) \equiv (\sum_{s_n \in S_n(i)}\s_n(s_n))_{i \in L_n}$. 
The map $q^e_n$ is affine and therefore $C_n \equiv q^e_n(\S_n)$ is a polytope in $\Re^{L_n}$. 
The polytope $C_n$ is the enabling strategy set of Player $n$. 
If Nature moves in $\Gamma$ we analogously denote the set of Nature's last moves by $L_0$.
Denote $L \equiv L_0\times L_1 \times L_2.$

Given extensive-form payoffs $G \in \mathcal{G}$ and a profile of enabling strategies $(c_1,c_2)\in C\equiv C_1\times C_2$, we construct payoffs are as follows.
First, for every Nature's last move $i_0\in L_0$, define $c_0(i_0)\equiv\prod_{i_0'\in A_0,i'\preceq i_0}\rho(i'_0)$.
Each of Player $n$'s last action in $L_n$ is associated with a unique sequence of Player $n$'s past actions that lead to the information set in which such a last action is available.
Therefore, a vector of last actions $i = (i_0,i_1,i_2)$ either defines a path from the root to a terminal node or it does not. 
In the latter case, define $g_n(i) \equiv 0$. 
In the former case, letting $z$ be the terminal nodes that is reached by such path, define $g_n(i) \equiv G_n(z)$. 
Player $n$'s enabling payoff function $V^{e}_n: C \to \Re$ is defined by
\begin{equation}
V^{e}_n(c) \equiv \sum_{i \in L}g_n(i)c_0(i_0)c_1(i_1)c_2(i_2).
\end{equation}
Note that the payoff function $V^{e}_n$ is affine in each coordinate and, therefore, defines a polytope-form game $\mathfrak{V}^{e} = (C_1,C_2,V^e_1,V^e_2)$ which is called the \textit{enabling-form of $G$}.

To prove Proposition~\ref{nonzeroindex}, let $d_n$ be the dimension of $C_n$.
% (the enabling strategy set of player $n$). 
%Let $C = C_1 \times C_2$. 
Since $K$ induces a unique outcome with full support, then $q^e(K) = p$ is an isolated equilibrium in enabling strategies, located in the interior of $C$. 
For each $n=1,2$, take a $d_n$-simplex $F_n$ contained in the relative interior of the enabling strategy set $C_n$ and containing $p_n$ in its relative interior.
If we restrict the enabling payoff function $V^{e}_n$ to $D_n$, then this defines a normal-form game with an isolated completely mixed equilibrium $p$, with the same index as the polytope-form index of the enabling-form (cf. \citealp{LP2023}). It is known that completely mixed isolated equilibria of normal-form games have indexes $+1$ or $-1$, since the payoff matrices are non-singular (and Shapley's formula gives their index, see \citealp{S1974}).  
In particular, $p$ has non-zero index. 
The equivalence between the index in normal and polytope forms (cf. \citealp{LP2023}) then implies that $K$ has non-zero index.

\section{Proof of Proposition~\ref{prop:genericity}}\label{app:genericity}

Fix a two-player game-tree $\Gamma$ with perfect recall. The set of terminal nodes is $Z$ and the set of terminal payoffs is denoted $\Re^{2|Z|}$. A typical element of $\Re^{2|Z|}$ is denoted $G$. We describe a system of polynomial equations whose zeros essentially correspond to an assignment of payoffs $G$ and: (i) an equilibrium of $G$ in behavior strategies; (ii) a \textit{deviation} behavior strategy of player 1 that grants him the same payoff as on-path; (iii) player 2 also best-replies to the deviation strategy of player 1. We prove that set of payoffs at which these three conditions are satisfied has dimension strictly lower than $2|Z|$. The proof where the roles of players 1 and 2 are reversed is obviously analogous. This implies our desired result.

We start with preliminary definitions required to describe the system.  Recall that, given player $n=1,2$, a behavior strategy $b_n = (b_n(i))_{i \in A_n}$ satisfies $(b_n(i))_{i\in A_n(u)}\in\Delta(A_n(u))$ for every $u\in U_n$. Given player 1's set of actions $A_1$, consider subsets $C_1$ and $C'_1$ such that (1) $C_1\cup C'_1=A_1$ and $C_1 \cap C'_1 = \emptyset$ and (2) each action in $C_1$ has all predecessors in $C_1$ or no predecessor. Note that there are finitely many pairs $(C_1,C'_1)$ satisfying this condition. For each $n=1,2$, the collection $U^+_n(C_1)\subset U_n$ is made of those information sets of player $n$ that can be reached (for some strategy of player 2) when player 1 is restricted to choosing actions in $C_1$. Let $U^0_n(C_1)\equiv U_n\setminus U^+_n(C_1)$, i.e., the set of information sets of player $n$ that can be reached when player 1 plays an action in $C'_1$ with positive probability. If $b_n$ is a behavior strategy of player $n$, then $b^+_n$ and $b^0_n$ are the restrictions to, respectively, $U_n^+(C_1)$ and $U^0_n(C_1)$. For $u\in U^+_n(C_1)$, let $G^+_{n,u,i}(b^+_1,b^+_2)$ be the payoff to player $n$ when (i) at every information set of player $n$ that precedes $u$, player $n$ plays the unique action that leads to $u$, (ii) player $n$ chooses action $i$ at $u$; and (iii) play at every other information set in $U^+_1(C_1)\cup U^+_2(C_1)$ is prescribed by $(b^+_1,b^+_2)$. Note that the definition of $G^+_{1,u,i}(b^+_1,b^+_2)$ can be generalized to $G^0_{1,u,i}(b^+_1,b^+_2,b^0_1,b^0_2)$, by requiring (i) and (ii) for $u\in U^0_1(C_1)$ and $i\in A_n(u)$, or $u\in U^+_1(C_1)$ and $i\in A_n(u)\cap C'$ and letting play at all other information sets be directed by $(b^+_1,b^+_2,b^0_1,b^0_2)$. Suppose now that $\vert{U}^0_2(C_1)\vert\geq 1$ so that: 

$$\tilde{C}_1\equiv  C'_1 \bigcap\left(\bigcup_{u\in U^+_1(C_1)}{A_1(u)}\right)\neq\varnothing.$$

Let $\a$ denote an arbitrary probability distribution over $\tilde{C}_1$. Note that for any $\a$, there exists a behavior strategy $b'_1$ for player 1 that induces $\a>0$ in the following sense:

\begin{equation}
\a(i)=\frac{\prod_{j\preceq i}b'_1(j)}{\sum_{i'\in\tilde{C}_1}\prod_{j\preceq i'}b'_1(j)}\text{ for each }i\in\tilde{C}_1.
\end{equation}

For $u \in U^0_2$, let $G^0_{2,u,i}(\a,b^0_1,b^0_2)$ be the payoff to player $2$ when (i) at every information set of player $2$ that precedes $u$, player $2$ plays the unique action that leads to $u$, (ii) the probability distribution $\a$, together with $b^0_1$, is used to compute beliefs at $u$, in the sense that any $b'_1$ inducing $\a$ is used together with $b^0_1$ to compute beliefs at $u$, (iii) player $2$ chooses action $i\in A_2(u)$; and (iv) play at every other information set in $U^0_1(C_1)\cup U^0_2(C_1)$ is prescribed by $(b^0_1,b^0_2)$.\footnote{~The expression $G^0_{2,u,i}(p,b^0_1,b^0_2)$ ignores player 2's utility at final nodes $z$ that do not follow some information set in $U^0_2$ as they do not affect optimal behavior at those information sets.}

Fix from now on a pair $(C_1, C'_1)$, let $C = C_1 \cup A_2$. For each $n$ and each $u\in U_n$ fix some action $k_u\in A_n(u)\setminus\tilde{C}_1$. Let $\hat{U}^+_n(C_1)$ (respectively, $\hat{U}^0_n(C_1)$) be the collection of information sets in $U^+_n(C_1)$ (respectively, $U^0_n(C_1)$) such that no other information set of player $n$ in $U^+_n(C_1)$ (respectively, $U^0_n(C_1)$) precedes them. We call this collection initial. The polynomial system presented below has variables $(G,(\alpha,\beta),(\mu_1,\mu_2),\mu^0_2)\in\mathbb{R}^{2\vert Z\vert}\times\mathbb{R}^{A}_{++}\times\mathbb{R}^{2}\times\mathbb{R}$:

\begin{alignat*}{3}
&g^+_{n,u,i}&&=G^+_{n,u,i}(\beta^+_1,\beta^+_2)-\mu_n,&&\quad\forall n\in N ,u\in U_n^+(C_1), i\in A_n(u)\cap C, i\neq k_u\text{ if }u\notin \hat{U}_n^+(C_1),\\
&g^0_{1,u,i}&&=G^0_{1,u,i}(\beta)-\mu_1,&&\quad\forall u\in U_1^+(C_1), i\in A_1(u)\cap\tilde{C}_1,\\
&g^0_{1,u,i}&&=G^0_{1,u,i}(\beta)-\mu_1,&&\quad\forall u\in U_1^0(C_1), i\in A_1(u), i\neq k_u,\\
&g^0_{2,u,i}&&=G^0_{2,u,i}(\alpha,\beta^0_1,\beta^0_2)-\mu^0_2,&&\quad\forall u\in U_2^0(C_1), i\in A_2(u), i\neq k_u\text{ if }u\notin\hat{U}^0_2(C_1),\\
&f_{n,u}&&=\sum_{i\in A_n(u)\cap C}{\beta_{n,u,i}}-1,&&\quad\forall n\in N,u\in U^+_n(C_1), \\
&f^0_{n,u}&&=\sum_{i\in A_n(u)}{\beta_{n,u,i}}-1,&&\quad\forall n\in N,u\in U^0_n(C_1), \\
&\tilde{f}&&=\sum_{i\in\tilde{C}_1}{\alpha_i}-1.&&
\end{alignat*}

Note that $(\alpha,\beta)$ is in $\mathbb{R}^{A}_{++}$: $\beta$ in particular can be seen as a collection of two subcollections of variables, $\beta_1$ for player 1 and $\beta_2$ for player $2$. $\beta_1$ comprises the collection of real variables of player 1 $(\b_{1,u,i})$ where $u$ is an information set of player 1 in $U_1^+(C_1)$ and $i$ an action available at that information set; or $u$ is an information set in $U^0_1(C_1)$ and $i$ an action available at that information set. Note that the collection has no variables in $\b_{1,u,i}$ where $u$ is an information set in $U_1^+(C_1)$ and $i$ is an available action in $u$ that is in $C'_1$. $\beta_2$ comprises the collection of $(\b_{2,u,i})$ for all information sets of player 2 and all available actions $i$ at that information set. Furthermore, $\alpha \in \Re^{\tilde C_1}_{++}$, which justifies $(\alpha,\beta) \in \mathbb{R}^{A}_{++}$. 

The set of zeros of the polynomial system is denoted $M$. The set $M$ corresponds precisely to the games defined by $G$, their equilibria $\b$ in behavior strategies where player 1 only plays actions in $C_1$ and player 2 completely mixes on all his information sets; and an equilibrium of the excluded game $(b'_1, \b_2)$ (where $b'_1$ induces $\a$) paying the same as $\b$ and where (i) all actions in $\tilde{C}_1$ are played with positive probability by $b'_1$; and (ii) at all information sets $U^0_{1}(C_1)$, $b'_1$ completely mixes. Note that any $b'_1$ which induces $\a$ is such that the beliefs induced on each information set of $U^0_2(C_1)$ that is reached in $(b'_1, \b_2)$ with positive probability is identical (because the relative probabilities assigned to the nodes in each such information set are always the same for any such $b'_1$). The probability $\a$ therefore summarizes a deviation strategy $b'_1$ in what is relevant for player $2$, namely, the belief it induces in the information sets of player 2. Polynomial $g^+_{n,u,i}$ comes from optimality on-path, $g^0_{1,u,i}$ comes from the optimality of players 1's deviation, and $g^0_{2,u,i}$ comes from the optimality of player 2's off-path reply to player 1's deviation.

Similarly to \citet{GW2001}, we have a key condition on the system that is imposed on three groups of polynomials: in the first group, it is  $i\neq k_u\text{ if }u\notin \hat{U}_n^+(C_1)$, which is the exact same condition as in that paper: the optimality of all but the $k_u$ action at a non-initial information set which is reachable by actions in $C_1$ implies the optimality of $k_u$, since all information sets of $U^{+}_n(C_1)$ are on path of $\beta^+_1$ and $\beta^+_2$. The analogous condition is also invoked for player $2$ in the fourth group of polynomials. Finally, the condition is also present in the third group of polynomials (note the absence of the "if" part: it would be redundant since all information sets in $U^0_1(C_1)$ are non-initial by definition.

If $i\in A_n(u)$, $u\in U_n$, let $\bar{Z}(i)$ be the set of terminal nodes that come after $i$ and after each action $k_{u'}$ available at the corresponding information set $u'$ of player $n$ that comes after $u$.
For each $z\in \bar Z(i)$, $\partial g^+_{n',u',i'}/\partial G_{n,z}\neq 0$ if and only if $n'=n$ and either $u'=u$ and $i'=i$ or $u'$ is an information set of $n$ that precedes $u$ and $i'$ is the unique choice that from $u'$ leads to $u$. The same is true about $\partial g^0_{n',u',i'}/\partial G_{n,z}\neq 0$. Similarly to \citet{GW2001}, we can make an assignment of rows to columns in the Jacobian matrix of the polynomials above that shows that it has an upper triangular structure.

We now show that the dimension of $M$ is less than $2|Z|-1$. First note that we have $2\vert Z\vert+\vert A\vert+2+1$ variables in the system. In the first group of polynomials, there are, for player $n$,  $\sum_{u \in U^{+}_n(C_1)}|A_n(u)| - U^+_n(C_1) + \hat{U}^+_n(C_1) $ equations. The second group has $\tilde{C}_1$ equations. The third group of has  $\sum_{u \in U^{0}_1(C_1)}|A_n(u)| - |U^{0}_1(C_1)|$ and the fourth has $\sum_{u \in U^{0}_2(C_1)}|A_2(u)| - U^0_2(C_1) + \hat{U}^0_2(C_1)$ equations. The fifth, sixth and seventh have, together, $|U_1| + |U_2| + 1$ equations. Therefore, the dimension of $M$ is:
\begin{multline*}
2\vert Z\vert+\vert A\vert+2+1-(\vert C\vert-\vert U^+_1(C_1)\vert-\vert U^+_2(C_1)\vert+\vert\hat{U}^+_1(C_1)\vert+\vert\hat{U}^+_2(C_1)\vert)-\vert \tilde{C}_1\vert \\
-(\vert A\setminus (C\cup \tilde{C}_1)\vert - \vert U^0_1(C_1) \vert -\vert U^0_2 (C_1)\vert + \vert\hat{U}^0_2(C_1)\vert) - \vert U_1\vert - \vert U_2 \vert-1\leq 2\vert Z\vert-1,
\end{multline*}
where we have used $\vert\hat{U}^+_n(C_1)\vert\geq 1$, $\vert\hat{U}^0_2(C_1)\vert\geq 1$, $\vert U_n\vert=\vert U^+_n(C_1)\vert+\vert U^0_n(C_1)\vert$, and $\vert A\vert=\vert {C}\vert+\vert \tilde{C}_1\vert+\vert A\setminus (C\cup \tilde{C}_1)\vert$.

If the assumption we made above (i.e. $|\hat{U}^0_2(C_1)| \geq 1$) does not hold, then no information set of player 2 is excluded when player 1 randomizes in $C_1$. Therefore, the deviations of player 1 lead immediately (with the exception of moves by Nature) to a terminal node and, applying the result from \citet{GW2001}, we obtain that only a non-generic choice of terminal payoffs could make player $1$ indifferent between deviating or not. 

We have therefore obtained that $M$ is a semi-algebraic surface of dimension less than $2|Z|-1$. To finalize the argument, we need the following lemma \citep[Theorem 2.8.8]{BCR1998}.
\begin{lemma}\label{lm:sa}
If $X$ is a semi-algebraic set and $f:X\to\mathbb{R}^m$ is a semi-algebraic function then $\dim(f(X))\leq\dim(X)$.
\end{lemma}

\begin{theorem}
The subset of terminal payoffs of $\Re^{2|Z|}$ for which the induced game has a Nash equilibrium outcome with player 1 indifferent between such an outcome and some equilibrium in the excluded game $\mathbb{G}^1$ is a semi-algebraic set of dimension lower than $2|Z|$.
\end{theorem}
\begin{proof}
Our previous construction is such that $A_1$ is a disjoint union of $C_1$ and $C'_1$ and only considers equilibria $(b_1,b_2)$ and a deviation $b'_1$ such that the support of $b_1$ is $C_1$, the support of $b_2$ is $A_2$, and $b'_1$ induces a full support $\a$ and completely mixes after all $\tilde{C}_1$ at all subsequent information sets. In this case, we consider the projection $\pi: M \to \Re^{2|Z|}$ over $\Re^{2|Z|}$ and apply Lemma~\ref{lm:sa}, thus obtaining that $\pi(M)$ is a semi-algebraic set with dimension less than $2|Z|$. For the general case, note that any equilibrium $(b_1,b_2)$ and deviation $b'_1$ can be made to satisfy the restrictions above by eliminating choices that are not in their corresponding support. Since the set of all choices is finite, we only need to take the intersection of finitely many, full-dimensional, semi-algebraic subsets of games to prove the result.
\end{proof}

Proposition~\ref{prop:genericity} readily follows from the theorem above.

\section{Proof of Lemma \ref{lm:excluded}}\label{sec:proofLemmaExcluded}

The proof of the lemma in Step 2 is divided in three parts. Before we present it, a preliminary fixed-point theoretic result must be established (Lemma \ref{lm:function} below). The proof of this auxiliary result resembles closely the proof of Proposition 3.2 in \cite{GLP2023} with minor changes. 

\begin{lemma}\label{lm:function}
For every open neighborhood $W^1$ of the graph of the best-reply correspondence $BR^1$ of game $\mG^1$, there are functions $\tilde{f},f^0:\S^1\to\S^1$ such that $\mathrm{Graph}(\tilde{f})\subset W^1$, $f^0\vert_{\S^1\setminus\mathrm{int}(K^{1,2\eta})}=\tilde{f}\vert_{\S^1\setminus\mathrm{int}(K^{1,2\eta})}$, and $f^0$ has no fixed point in $K^{1,2\eta}$.
\end{lemma}

\begin{proof}Fix $W^1$ a neighborhood of the graph of the best-reply correspondence  $BR^1$ of game $\mG^1$, such that any continuous map $f$ whose graph is in $W^1$ has a displacement $d_f|_{K^{1, 2\eta}} = (id - f)|_{K^{1, 2\eta}}$ with degree $0$.

Denote by $\partial K^{1, 2\eta}$ the boundary of $K^{1, 2\eta}$ relative to the affine space generated by $\S^1$.  Note that $(K^{1, 2\eta}, \partial K^{1, 2\eta})$ is trivially homeomorphic to a ball with boundary. 
Define $\tilde f_t: \S^1 \to \S^1$ as follows. 
Letting $\t_n^0$ be the barycenter of $\S^1_n$ for each $n$, define $\forall \s \in \S^1 , \tilde f_t(\s) \equiv (1- t)f(\s) + t \t^0$. 
Taking $t>0$ sufficiently small, we ensure that $\partial K^{1, 2\eta}$ has no fixed points of $\tilde f_t$, and $K^{1, 2\eta}$ has index $0$ w.r.t. $\tilde f_t$. 
We fix such a $t$ and omit the parameter $t$ from $\tilde f_t$, denoting it by $\tilde f$. 
Let $A_1$ and $A_2$ be, respectively, the hyperplane in $\Re^{\vert S^1_1\vert}$ and $\Re^{\vert S^0_2\vert}$ through the origin and with normal $(1, \ldots, 1)$, and let $A = \prod_n A_n$. 
The map $d_{\tilde f} = id - \tilde f$ maps $\S^{1}$ into $A$. Therefore, $d_{\tilde f}: (K^{1, 2\eta}, \partial K^{1, 2\eta}) \to (A, A - 0)$ has degree zero. By the Hopf Extension Theorem (cf. Corollary 8.1.18, \citealp{SP1966}) there exists a map $\hat d^0$ from $K^{1, 2\eta}$ to $A - 0$ such that its restriction to $\partial K^{1,2\eta}$ coincides with $d_{\tilde f}$.  Extend $\hat d^0$ to the whole of $\S^1$ by letting it be $d_{\tilde f}$ outside $K^{1, 2\eta}$.  
%and therefore $\s^*$ is the unique fixed point of $\tilde f$ in $V$. The set $\S \backslash V$ is mapped by $\tilde f$ to $\S \backslash \partial \S$. Hence all the other fixed points of $\tilde f$ belong to $X \backslash \partial X$.

%Let $\tilde d$ be the displacement of $\tilde f$: $\tilde d(\s) \equiv \s - \tilde f(\s)$. 

For each $\s \in \S^1$, there exists $\xi \in (0, 1]$ such that $\s - \xi \hat{d}^0(\s) \in \S^1$.
Indeed, this is obvious for $\s \in K^{1,2\eta} \setminus \partial K^{1, 2\eta}$, since $\s$ belongs to $\S^1 \setminus \partial \S^1$.
For $\s \in \partial K^{1,2\eta} \cup (\S^1 \setminus K^{1,2\eta})$, we can take $\xi = 1$ as $\hat d^0 = d_{\tilde f}$.  
Now for each $\s$, let $\xi(\s)$ be the largest $\xi \in [0, 1]$ such that $\s - \xi(\s)\hat d^0(\s) \in \S^1$. 
Note that the map $\s \mapsto \xi(\s)$ is continuous. 
Define $f^0: \S^1 \to \S^1$ by $f^0(\s) \equiv \s - \xi(\s)\hat d^0(\s)$. 
The map $f^0$ is continuous, coincides with $\tilde f$ on $\S^1 \setminus (K^{1, 2\eta} \setminus \partial K^{1, 2\eta})$ and has no fixed point in $K^{1, 2\eta}$. 
\end{proof}

The reasoning in the three-part proof presented below resembles closely the reasoning presented in the proof of Theorem 3 in \citetalias{GW2005}. For the sake of completeness we include details here, with the necessary explicit modifications required for our setting.

\subsection*{Notation}All norms presented from now on are $\ell_{\infty}$-norms. Given a finite game $\mG = (S_1, S_2, \mG_1, \mG_2)$, let $g \in \mathbb{R}^{|S_1| + |S_2|}$. Denote by $\mG \oplus g$ the finite game where pure strategies are $S_n, n=1,2$ and payoffs are defined by $\mG_n(s_1, s_2) + g_{s_n}, n =1,2$. In Part II, we will use an analogous construction: given a continuous function $g: \S \to \Re^{|S_1| + |S_2|}$, where $\S = \S_1 \times \S_2$ and $\S_n = \D(S_n)$, for each $\s \in \S$, the finite game $\mG \oplus g(\s)$ is then defined. The payoff of a mixed strategy profile $\s'$ to player $n$ is then $\mG_n(\s') + \s'_n \cdot g_n(\s)$, where  $g_n(\s) \in \Re^{|S_n|}$. Given $\mu>0$ and $\s \in \S$, $\t \in \S$ is a $\mu$-reply to $\s$ if it satisfies: $\forall n, s_n, \mG_n(\t_n, \s_{-n}) \geq \mG_n(s_n, \s_{-n}) - \mu$. A duplicate strategy $\bar s_n$ in $\mG$ is a pure strategy that is equivalent to a mixed strategy $\s_n$ of $\mG$, i.e., for all $m, s_{-n} \in S_{-n}, \mG_m(\bar s_n, s_{-n}) = \mG_m(\s_n, s_{-n})$.

\subsection{Part I}

\begin{lemma}\label{partI}
Without loss of generality, game $\mG^1$ satisfies the following property: for every neighborhood $W^1$ of Graph($\BR^1$), there exists a map $h: \S^1 \to \S^1$ s.t. 

\begin{enumerate} 
	\item \label{proximity} Graph($h) \cap [(\S^1 \setminus K^{1,2\eta}) \times \S^1]  \subseteq W^1 \cap [(\S^1 \setminus K^{1, 2\eta}) \times \S^1]$
	
	\item \label{independence} For each player $n$, the $n$-th coordinate map $h_n$ of $h$ depends only on $\S_{-n}$

	\item \label{nofp} $h$ has no fixed point in $K^{1,2\eta}$  

\end{enumerate} 
\end{lemma}

\begin{proof} 
This lemma's proof parallels the reasoning in Step 1 of \citetalias{GW2005} very closely, with no major changes. We quickly recall the details for completeness. First let $\mG^*$ be defined by letting $S^*_n = S^1_n \times S^1_{n+1}$, where $n=1,2$ is taken modulo $2$, be the strategy set of player $n$. For each $n$, and $m \in \{n,n+1\}$, denote the natural projection $p_{n,m}$ from $S^*_n$ to $S^1_m$. Define payoffs $\mG^*_n(s^*) = \mG^1_n(s),$ where for each $m$, $s_m = p_{m,m}(s^*_m)$. Extend $p_{n,m}$ linearly to $\Delta(S^*_n)$ so that $p_{n,m}$ computes the marginal over $S^1_{m}$. Let $p: \S^* \to \S^1$ be the map $p(\s^*) = (p_{1,1}(\s^*_1), p_{2,2}(\s^*_2))$. The map $p^*$ computes the payoff-relevant coordinates of $\s^*$ in the game $\mG^*$. Let $P: \S^* \times \S^* \to \S^1 \times \S^1$ be defined by $P(\s^*,\t^*) = (p(\s^*), p(\t^*))$.

Let $\BR^*$ denote the best-reply correspondence of game $\mG^*$. Let $K^*$ denote the set of $\mG^*$ that is equivalent to $K^{1,2\eta}$. Fix a neighborhood $W^*$ of Graph($\BR^*$). For $\mu>0$, consider $W(\mu) = \{ (\s^1, \t^1) \in \S^1 \times \S^1 \mid \tau^1 \text{ is a $\mu$-reply to }\s \}$. The collection $\{W(\mu)\}_{\mu>0}$ is a basis of neighborhoods of Graph($\BR^*$). Choose $\mu>0$ such that $P^{-1}(W(\mu)) \subseteq W^*$. 

Apply now Lemma \ref{auxStep1} to obtain $f^0$ with (i) Graph($f^0|_{\S^1 \setminus K^{1,2\eta}}) \subseteq W(\mu) \cap [(\S^1 \setminus K^{1, 2\eta}) \times \S^1]$ and (ii) $f^0$ has no fixed point in $K^{1, 2\eta}$.  Define now the map $h^*: \S^* \to \S^*$ such that for each $n$, $h^*_n(\s^*) = \t_n(\s^*) * p_{n+1,n+1}(\s^*_{n+1})$, where $\t_n(\s^*) * p_{n+1,n+1}(\s^*_{n+1})$ is a product probability (measure) with 
$$ \t^{*}_1(\s^*) \equiv  f^0_1 (p_{2,2}(\s^*_2), p_{2,1}(\s^*_2)),$$ 
$$\t^{*}_2(\s^*) \equiv f^0_2 (p_{1,1}(\s^*_2), p_{1,2}(\s^*_2)).$$
Note that each coordinate map $h^*_n$ is defined over $\S^*_{-n}$ and therefore satisfies \eqref{independence}.

We claim that \text{Graph}$(h^*|_{\S^* \setminus K^*}) \subseteq W^* \cap ([\S^* \setminus K^*] \times \S^*)$. From (i) above, it follows that $\t_n(\s^*)$ is a $\mu$-reply to $p_{-n}(\s^*)$. Therefore $(p(\s^*), \t(\s^*))$ belongs to $W(\mu)$. Hence $(\s^*, h^*(\s^*)) \in P^{-1}(W(\mu)) \subseteq W^*$, which concludes the proof of the claim. Therefore,  $h^*$ satisfies \eqref{proximity}.

To conclude the proof of the lemma, suppose now that $\s^*$ is a fixed point of $h^*$. Then $\s^*_n$ is a product strategy with $p_{n,n+1}(\s^*_n) = p_{n+1,n+1}(h^*_n(\s^*_n)).$ Hence, $p_{n,n}(\s^*_{n}) = p_{n,n}(h^*_{n}(\s^*)) = f^0_n(p_{-n, -n}(\s^*), p_{n-1,n}(\s^*_{n-1})) = f^0_n(p(\s^*)).$ Therefore, $p(\s^*)$ is a fixed point of $f^0$. Since from (ii) above $f^0$ has no fixed point in $K^{1,2\eta}$, $\s^* \notin K^*$. So $h^*$ satisfies \eqref{nofp}. \end{proof}

\subsection{Part II}

\begin{lemma}\label{partII}
For any $\d_0>0$, there exists $\d>0$ and an equivalent game $\bar{\mG}^1$ to $\mG^1$ obtained by adding duplicates to $\mG^1$ such that the following holds:  there exists a map $g: \bar{\S}^1 \to \mathbb{R}^{\bar{R}}$  (where $\bar{R} = |\bar{S}_1| + |\bar{S}_2|$) satisfying: 

\begin{enumerate} 
\item \label{independence2} For each player $n$, $g_n$ depends on $\bar{\S}^1_{-n}$.

\item \label{noeq2} No profile $\bar{\s} \in \bar{K}^{1}$ is an equilibrium of the game $\bar{\mG} \oplus g(\bar \s)$.

\item \label{bounddelta0} $\Vert g(\bar \s) \Vert \leq \d_0$, for each $\bar \s \in \bar{\S}^1 \setminus \text{int}_{\bar{\S}^1}(\bar{K}^{1,2\eta})$.

\item \label{bounddelta} $\Vert  g(\sigma) \Vert \leq \d$, for each $\bar \s \in \bar{K}^{1,2\eta}$.
 
\end{enumerate} 
\end{lemma} 

The proof will require a simple auxiliary result whose proof can be found in \citetalias{GW2005}. Let $M$ be a constant satisfying: for each $\s,\t$, $\Vert \mG^1_n(\s) - \mG^1_n(\t) \Vert \leq M \Vert \s - \t \Vert$, for each $n=1,2$.

\begin{lemma}[Lemma, \citetalias{GW2005}]\label{lemmaaux}
If $\t_n$ is a $\beta_1$-reply against $\s$, $\Vert \t_n -\t'_n \Vert \leq \beta_2$, $\Vert \s - \s' \Vert \leq \beta_3$, then $\t'_n$ is a $\beta_1 + M\beta_2$-reply to $\s$, and $\t_n$ is a $(2M\beta_3 + \beta_1)$-reply to $\s'$.
\end{lemma}

\begin{proof} See \citetalias{GW2005}.\end{proof}

\begin{lemma}\label{auxlemma2}
There exists $\d_1>0$ such that for any equivalent game $\bar{\mG}^1$ to $\mG^1$ and any $\bar{g} \in \mathbb{R}^{|\bar{S}_1| + |\bar{S}_2|}$ with $\Vert \bar{g} \Vert \leq \d_1$, $\bar{\mG}^1 \oplus \bar{g}$ does not have an equilibrium in $\bar{K}^1 \setminus \bar{K}^{1,2\eta}$.
\end{lemma}

\begin{proof}[Proof of Lemma \ref{auxlemma2}] Observe first that $K^1 \setminus \text{int}_{\S^1}(K^{1,2\eta})$ is a compact set. For each $\s \in K^1 \setminus \text{int}_{\S^1}(K^{1,2\eta})$, there exists $\e(\s)>0$, $\rho(\s)>0$, $n \in \{1,2\}$ and $s_n \in S_n$ such that for each $\s' \in B_{\e(\s)}(\s)$, we have
$$ \mG^1_n(s_n, \s'_{-n}) > \mG^1_n(\s') + \rho(\s)$$

The balls $\{B_{\e(\s)}(\s)\}_{\s \in K^1 \setminus \text{int}_{\S^1}(K^{1,2\eta})}$ form an open cover of $K^1 \setminus K^{1,2\eta}$, from which we extract a finite subcover: there exists $\s^1,...,\s^k$ with corresponding $\e(\s^i)>0$, $\rho(\s^i)>0$, $n_i \in \{1,2\}$ and $s_{n_i} \in S_{n_i}$ such that for each $\s' \in B_{\e(\s^i)}(\s^i)$
$$\mG^1_{n_i}(s_{n_i}, \s'_{-n_{i}}) > \mG^1_{n_i}(\s') + \rho(\s^i).$$

Define $\rho = \text{min}_i \rho(\s^i)$. Fix an equivalent game $\bar{\mG}^1$ to $\mG^1$, and $\bar{\s} \in \bar{K}^1 \setminus (\text{int}_{\bar{\S}^1} \bar{K}^{1,2\eta}$). Let $\s$ be the equivalent profile in $\mG$. Then there exist $n, s_n \in S_n \subseteq \bar{S}_n$, 
$$ \bar{\mG}^1_n(s_n, \bar{\s}_{-n}) = \mG^1_n(s_n, \s_{-n}) > \mG^1_n(\s) + \rho = \bar{\mG}^1_n(\bar{\s}) + \rho.$$

Hence, for any $\bar{g} \in \mathbb{R}^{|\bar{S}_1|+|\bar{S}_2|}$ with $\Vert \bar{g} \Vert \leq \frac{\rho}{2} =: \d_1$:
$$ \bar{\mG}^1_n(s_n, \bar{\s}_{-n}) + \bar{g}_{s_n}> \bar{\mG}^1_{n}(\bar{\s}) + \bar{\s}_n \cdot \bar{g}_n.$$ Therefore, $\bar{\s}$ is not an equilibrium of $\bar{\mG}^1 \oplus \bar{g}$. \end{proof}

\begin{proof}[Proof of Lemma \ref{partII}]

Fix $0<\d_0 < \d_1$, where $\d_1$ is obtained from Lemma \ref{auxlemma2}. Fix $\eta_0 = \frac{\d_0}{6M}$. For each $\s$, there exists an open ball $B(\s)$ around $\s$ of radius strictly less than $\eta_0$ such that for each $\s' \in B(\s)$, the set of pure best replies to $\s'$ is a subset of those that are best replies to $\s$. Since the set of best replies for each player $n$ to a strategy profile is the face of $\S^1_n$ spanned by her pure strategies $\BR(\s') \subseteq\BR(\s), \forall \s' \in B(\s)$.

The collection of balls $\{B(\s)\}_{\s}$ define an open covering of $\S^1$. By compactness, there exists a finite set of points $\s^1, \s^2,...,\s^k$ whose corresponding balls form a subcover. For each $\s^i$ let $W(\s^i)$ be the $\eta_0$-neighborhood of $\BR^1(\s^i)$.

Let $W = \bigcup_i (B(\s^i) \times W(\s^i))$. The set $W$ is a neighborhood of Graph($\BR^1$). From Part I, there exists $h: \S^1 \to \S^1$ such that: (i) Graph($h) \cap (\S^1 \setminus K^{1,2\eta}) \times \S^1 \subset W \cap (\S^1 \setminus K^{1,2\eta}) \times \S^1$; (ii) For each player $n$, the coordinate map $h_n$ of $h$ depends on $\S^1_{-n}$; 
(iii) $h$ has no fixed point in $K^{1}$.

Let $\t = h(\s)$, where $\s \in \S^1 \setminus K^{1,2\eta}$. Then there exist $\s^i, \t^i$ with $\s \in B(\s^i),$ $\t^i$ is a best reply to $\s^i$ and $\t$ is within $\eta_0$ of $\t^i$. Therefore, Lemma \ref{lemmaaux} implies that $\t^i$ is a $2M\eta_0$-reply against $\s$ and therefore $\t$ is a $3M\eta_0$-reply against $\s$.

Fix now $\alpha>0$ such that if $\s \in K^{1,2\eta}$, then $\Vert \s - h(\s) \Vert >\alpha$. For each $n=1,2$, take a simplicial subdivision  $\mathcal{T}_n$ of $\S_n$ such that: (a) $K^{1,2\eta}_2$ is the space of a subcomplex of $\mathcal{T}_2$; (b) the  diameter of the subdivision is less than $\eta_0$ and $\alpha$; (c) no simplex that intersects $K^{1,2\eta}_2$ intersects the boundary of $K^{1,\eta}_2$ w.r.t $\S^1_2$. Let $T_n$ be the set of vertices of $\mathcal{T}_n$. Let $T \equiv T_1 \times T_2$. 

We now define a game $\bar{\mG}^1$ that is equivalent to $\mG^1$. For each player $n$, the pure strategy set of player $n$ is $T_n$. The game $\bar{\mG}^1$ is equivalent to $\mG^1$ and the mixed strategy set is denoted $\bar{\S}^1= \bar{\S}^1_1 \times \bar{\S}^1_2$.

For each $t_n \in T_n$, let $St(t_n)$ denote the star of $t_n$ in the simplicial complex $\mathcal{T}_n$ and $ClSt(t_n)$ is the closed star of $t_n$ with respect to $\mathcal{T}_{n}$. Let $X(t_n) \equiv h^{-1}_n(ClSt(t_n)) \subseteq \S^1_{-n}$; let $Y(t_n) \equiv \{ \s_{-n} \in \S^1_{-n} | \Vert t_n - h_n(t_n) \Vert \geq 2\eta_0 \}$. Since the diameter of each simplex of $\mathcal{T}_n$ is less than $\eta_0$, $X(t_n) \cap Y(t_n) = \emptyset$. Use Urysohn's Lemma to define a function $\pi_{t_n}: \S^1_{-n} \to [0,1]$, with $\pi^{-1}_{t_n}(1) \equiv X(t_n)$ and $Y(t_n) \subset \pi^{-1}_{t_n}(0)$.

Let $R' = |T_1| + |T_2|.$ We define a map $g: \bar{\S}^1 \to \mathbb{R}^{R'}$, by first defining it on $\S^1$ and extending to $\bar{\S}^1$ by equivalence. Let $v_n(\s_{-n}) = \text{max}_{s \in S_n} \mG^1_n(s, \s_{-n})$. If $t_2 \notin K^{1,2\eta}_2$, we define $g_{t_2} \equiv 0$. For any other $t_2$: 
$$ g_{t_2}(\s_{1}) \equiv \pi_{t_2}(\s_{1})[v_2(\s_{1}) - \mG^1_2(t_2, \s_{1}) + M\eta_0].$$

Now, for $\s_2 \in \S^1_2 \setminus \text{ int}_{\S^1_2}(K^{1,\eta}_2)$, define, for each $t_1 \in T_1$, $g_{t_1}(\s_2) \equiv 0$. If $\s_2 \in K^{1, 2\eta}_2$, let  
$$ g_{t_1}(\s_{2}) \equiv \pi_{t_1}(\s_{2})[v_1(\s_{2}) - \mG^1_1(t_1, \s_{2}) + M\eta_0].$$ 

Fix $n, t_n, \s \in \S^1 \setminus K^{1,2\eta}$. We want to prove that $0 \leq g_{t_n}(\s) \leq 6M\eta_0 = \d_0$. %If $n=2$ and $t_2 \notin K^{1,2\eta}_2$, then there is nothing to prove, since $g_{t_2} \equiv 0$. %Likewise, if $\s_2 \in \S^1_2 \setminus K^{1,\eta}_2$, then for each $t_1 \in T_1, g_{t_1}(\s_2) = 0$. 

If $\s_1 \in Y(t_2)$, then $g_{t_2}(s_2) =0$. If $\s_1 \notin Y(t_2)$, then $\Vert t_2 - h_2(\s_1)\Vert \leq 2\eta_0$.  Since $h_n(\s_{-n})$ is a $3M\eta_0$-reply to $\s_{-n}$, Lemma \ref{lemmaaux} implies that $t_n$ is a $5M\eta_0$-reply to $\s_{-n}$, i.e., $0 \leq f_n(\s_{-n}) - \mG_{n}(t_n, \s_{-n}) \leq 5M\eta_0$. Hence, $0 \leq g_{t_n}(\s) \leq 6M\eta_0 = \d_0$. Hence $g_{t_2}(\s_1) \in [0,\d_0]$. 

Now assume $\s_2 \in (\S^1_2 \setminus \text{int}_{\S^1_2}K^{1,\eta}_2) \cup K^{1,2\eta}_2$. Note the union is disjoint. If $\s_2 \in Y(t_1)$, then $g_{t_1}(\s_2)=0$. If $\s_2 \notin Y(t_1)$, then repeating the same reasoning as above and considering the definition of $g_{t_1}$ one gets $g_{t_1}(\s_2) \in [0,\d_0]$. Using now Urysohn's Lemma, extend $g_{t_1}$ continuously to $\S^1_2$ so that image of $g_{t_1}$ is contained in $[0,\d_0]$. This proves \eqref{bounddelta0}. Note that, in addition, since $K^{1,2\eta}$ is compact and $g$ continuous, there exists $\d>0$ such that $\Vert  g(\s) \Vert \leq \d$ for each $\s \in K^{1,2\eta}$. This shows \eqref{bounddelta}. It is also clear from the construction of $g$ that it satisfies \eqref{independence2}.

To finish the proof we show \eqref{noeq2}. Suppose $\bar{\s} \in \bar{K}^{1,2\eta}$ is an equilibrium and let $\s$ be the corresponding equivalent strategy in $\S^1$. Following the exact same reasoning as in the last paragraph of Step 2 of \citetalias{GW2005}, the pure best-replies $T'_n$ to $\s_{-n}$ are those $t_n$ for which $\s_{-n} \in X(t_n)$. Obviously, there exists at least one such $t_n$, i.e., a vertex of the carrier of $h_n(\s_{-n})$ in $\mathcal{T}_n$. Hence the distance between $t_n$ and $h_{n}(\s_{-n})$ is $ < \alpha$. The support of $\bar{\s}_n$ being a subset of $T'_n$ implies that $\Vert \s_n - h_n(\s_{-n}) \Vert \leq \alpha$. This is a contradiction and proves that $\bar{K}^{1,2\eta}$ contains no profile $\bar{\s}$ that is an equilibrium of $\bar{\mG}^1 \oplus g(\bar{\s})$.

Suppose now that $\bar{\s} \in \bar{K}^1 \setminus \bar{K}^{1,2\eta}$ and consider the game $\bar{\mG} \oplus g(\bar{\s})$. Then, by construction, $\Vert g(\bar \s) \Vert \leq \d_0 < \d_1$. Our choice of $\d_0$ then implies that $\bar{\s}$ is not an equilibrium of $\bar{\mG}^1 \oplus g(\bar{\s})$. This concludes the proof of this step. 
\end{proof}

\subsection{Part III}
We now conclude by proving Lemma \ref{lm:excluded}. 

Let $\bar{\mG}^1$ to be the game from Lemma \ref{partII}, satisfying properties (1)-(4). Note that $\bar{K}^{1,2\eta}$ is compact, $g$ is continuous and if $\bar{\s} \in \bar{K}^1$, then $\bar{\s}$ is not an equilibrium of $\bar{\mG} \oplus g(\bar \s)$. There exists then $\xi>0$ such that no $\bar{\s} \in \bar{K}^1$ is an equilibrium of $\bar{\mG} \oplus g$, with $\Vert g - g(\bar{\s}) \Vert < \xi$. Choose $\zeta>0$ such that: 
$$ \Vert g(\bar{\s}) - g(\bar{\s}') \Vert \leq \xi \text{, if } \Vert \bar{\s} - \bar{\s}' \Vert < \zeta.$$

Consider now a simplicial subdivision $\mathcal{S}_n$ of $\bar{\S}^1_n$ with diameter less than $\zeta>0$ and satisfying the requirement that no simplex of $\mathcal{S}_2$ that intersects $\tilde{\partial}K^1_2$ has a point in common with $\bar{K}^{1,\eta}$.

Let $\mathcal{P}_n$ be a polyhedral complex generated by $\mathcal{S}_n$, with $\g_n$ the convex function that is linear on each polyhedron ($\mathcal{P}_n$ refines $\mathcal{S}_n$) (cf. Appendix B in \citetalias{GW2005}). Consider now the game $\tilde{\mG}$ where the strategy set of each player $n$ is $P_n$, the set of vertices of $\mathcal{P}_n$. Let $\tilde{\S}_n$ be the mixed strategies of $\tilde{\mG}$. It is clear that $\tilde{\mG}$ satisfies Lemma \ref{lm:excluded}~\eqref{1} and \eqref{2}.

Define now the matrix $A_n \in \mathbb{R}^{|\bar{S}_n| \times |P_n|}$, by letting the column $p_n \in P_n$ have the coordinates of the mixed strategy $p_n \in \bar{\S}_n$ in $\bar S_n$. Therefore, given $\tilde{\s}_m$, $\bar{\s}_m = A_m\tilde{\s}_m$ gives the equivalent strategy $\bar{\s}_m$ to $\tilde{\s}_m$ in $\bar{\mG}$. %Let $g^*_{p_n}: P_{-n} \to \Re$ be defined by: (i) $g^*_{\bar s_2}(\cdot) =0, \bar s_2 \notin \bar{K}^{1,\eta}$; (ii) $g^{*}_{\bar s_2} = g_{\bar s_2}, \forall \bar s_2 \in \bar{K}^{1,\eta}$; (iii) $g^*_{\bar s_1}(p_2) = 0, \forall p_2 \notin \bar{K}^{1,\eta}_2$; $(iv) g^*_{\bar s_1}(p_2) = g_{\bar s_1}(p_2), \forall p_2 \in  \bar{K}^{1,\eta}_2$.

For each $n$, let now $B_n: P_{-n} \to \mathbb{R}^{P_n}$ be defined by $B_n(p_{-n}) = A^T_n g_n(p_{-n})$. We define, 
$$ \tilde \mG^{1,\d}_n(p) \equiv \tilde{\mG}^1_{n}(p) + B_{n,p_n}(p_{-n}).$$

$\tilde{\mG}^{1,\d}$ is a $\d$-perturbation of $\tilde{\mG}^1$, which is equivalent to $\bar{\mG}^1$. For $\e>0$, define: 
$$ \tilde{\mG}^{1,\d,\e}_n \equiv \tilde{\mG}^1_n(p) + B_{n,p_n}(p_{-n}) - \e \g_n(p_n).$$

$\tilde{\mG}^{1,\d,\e}$ is an $\e$-perturbation of $\tilde{\mG}^{1,\d}$. Our construction of $g$ in Part II now implies that $\tilde{G}^{1,\d}$ and $\tilde{G}^{1,\d,\e}$ satisfies Lemma \ref{lm:excluded}~\eqref{part:payoffs}.

We claim that for sufficiently small $\e>0$, the game $\tilde{\mG}^{1,\d,\e}$ has no equilibrium in $\tilde{K}^1$, the set equivalent to $\bar{K}^1$ (and therefore to $K^1$). Suppose to the contrary that there exists $(\e_k)_{k \in \mathbb{N}}$ converging to zero and a corresponding sequence $\tilde{\s}^{k}$ in $\tilde{K}^{1}$. For each $k$, let $\bar{\s}^k$ be the equivalent profile in $\bar{\S}^1$. For each $k$ and each player $n$, if $\tilde{\t}^k_n \in \tilde{\S}^1_n$ such that $A_n\tilde{\t}^k_n = \bar{\s}^k_n$, then $\sum_{p_n \in P_n}\tilde{\t}^k_n(p_n)\g_n(p_n) \geq \sum_{p_n \in P_n}\tilde{\s}^k_n(p_n)\g_n(p_n)$. Thus $\tilde{\s}^k_n$ solves the linear programming problem $\text{min}_{\tilde{\t}^k_n}\sum_{p_n \in P_n}\tilde{\t}^k_n(p_n)\g_n(p_n)$ subject to $A_n\tilde{\t}^k_n = \bar{\s}^k_n$. Let $L^k_n$ be the unique polyhedron of $\mathcal{P}_n$ that contains $\bar{\s}^k_n$ in its interior. Since $\g_n$ is a convex function $\g_n(\tilde \s^k_n) \leq \sum_{p_n \in P_n}\tilde{\t}^k_n(p_n)\g_n(p_n)$. Moreover, the construction of $\g_n$ ensures that this inequality is strict unless the support of $\tilde{\t}^k_n$ is included in $L^k_n$. Therefore Lemma \ref{lm:excluded}~\eqref{part:support} is satisfied and the equilibrium strategy $\tilde{\s}^k_n$ assigns positive probability only to points in $L^k_n$.

Now let $\tilde{\s}$ be a limit of $\tilde{\s}^k$ as $\e^k \to 0$, and let $\bar{\s}$ be an equivalent mixed strategy. Therefore $\bar{\s}$ is an equilibrium of the game $\bar{\mG} \oplus b$, where $b_{n, \bar{s}_n} = \sum_{p_{-n}}g_{n, \bar{s}_n}(p_{-n})\tilde{\s}_{-n}(p_{-n})$. By the arguments above, there exists a polyhedron $P^{\circ}_n \in \mathcal{P}_n$ such that $\tilde{\s}_n$ assigns positive probability only to points in $P^{\circ}_n$. Viewing a vertex $p^*_{-n}$ of $P^{\circ}_{-n}$ as a point in $\bar{\S}^1_{-n}$, we have 
$$ \Vert p^*_{-n} - \bar{\s}_{-n} \Vert < \zeta \implies \Vert g_n(p^*_{-n}) - g_n(\bar{\s}_{-n}) \Vert \leq \xi.$$ This implies 
$$ \Vert \sum_{p_{-n}}g_{n}(p_{-n})\tilde \s_{-n}(p_{-n}) - g_{n}(\bar{\s}_{-n}) \Vert \leq \xi. $$

Letting $\bar{g}_n \in \mathbb{R}^{\bar S_n}$ be defined by $\bar{g}_{n,\bar{s}_n} =  \sum_{p_{-n}}g_{n, \bar{s}_n}(p_{-n})\tilde \s_{-n}(p_{-n})$. Therefore, $\bar \s$ is an equilibrium of $\bar{\mG} \oplus \bar{g}$. However, 
$$\Vert \bar{g} - g(\bar \s) \Vert \leq \xi.$$ Therefore, $\bar{\s} \notin \bar{K}^1$, by our assumption on $\xi$. Therefore, for sufficiently small $\e>0$, $\tilde{\mG}^{1,\d,\e}$ has no equilibrium in $\tilde{K}^{1}$, so Lemma \ref{lm:excluded}~\eqref{part:noeq} is satisfied.

\renewcommand{\bibfont}{\small}
\bibliographystyle{abbrvnat}
\bibliography{references}

\end{document}